\documentclass[10pt]{article}

\usepackage{amssymb}
\usepackage{latexsym}
\usepackage{amsmath}
\usepackage{amsthm}
\usepackage{euscript}
\usepackage{graphics}
\usepackage{color}

\newcommand{\vio}{\color{black}}

\numberwithin{equation}{section}

\newcommand\gH{{\mathfrak{H}}}

\newcommand\gN{{\mathfrak{N}}}

\newcommand\fS{{\mathfrak{S}}}

\newcommand\gotD{\mathfrak{D}}
\newcommand\gotH{\gH}
\newcommand\gotK{\mathfrak{K}}

\newcommand\gotL{\mathfrak{L}}
\newcommand\gotN{\gN}

\newcommand\gotS{\fS}
\newcommand\gotX{\mathfrak{X}}
\newcommand\gotY{\mathfrak{Y}}

\newcommand{\ga}{{\alpha}}
\newcommand{\gd}{{\delta}}
\newcommand{\gD}{{\Delta}}
\newcommand{\gga}{{\gamma}}
\newcommand{\gG}{{\Gamma}}

\newcommand{\gk}{{\kappa}}

\newcommand{\gl}{{\lambda}}
\newcommand{\gL}{{\Lambda}}
\newcommand{\gO}{{\Omega}}
\newcommand{\go}{{\omega}}
\newcommand{\gs}{{\sigma}}
\newcommand\gS{{\Sigma}}
\newcommand{\gth}{{\theta}}
\newcommand{\gT}{{\Theta}}

\newcommand{\R}{{\mathbb{R}}}
\newcommand{\C}{{\mathbb{C}}}
\newcommand{\D}{{\mathbb{D}}}
\newcommand{\bO}{{\mathbb{O}}}
\newcommand{\T}{{\mathbb{T}}}
\newcommand{\N}{{\mathbb{N}}}
\newcommand{\Z}{{\mathbb{Z}}}

\newcommand\cA{{\mathcal{A}}}
\newcommand\cB{{\mathcal{B}}}
\newcommand\cC{{\mathcal{C}}}

\newcommand\cF{{\mathcal{F}}}
\newcommand\cH{{\mathcal{H}}}

\newcommand\cL{{\mathcal{L}}}
\newcommand\cN{{\mathcal{N}}}
\newcommand\cO{{\mathcal{O}}}
\newcommand\cP{{\mathcal{P}}}
\newcommand\cS{{\mathcal{S}}}

\newcommand\cU{{\mathcal{U}}}
\newcommand\cV{{\mathcal{V}}}

\newcommand\res{{\rm res}}

\def\h#1{{{\hat #1} }}
\def\wt#1{{{\widetilde #1} }}
\def\wh#1{{{\,\widehat #1\,} }}

\def\bm\chi{\mbox{\boldmath$\chi$}}

\def\RE{{\rm Re\,}}
\def\IM{{\rm Im\,}}
\def\Ext{{\rm Ext\,}}

\def\ran{{\rm ran\,}}
\def\dom{{\rm dom\,}}

\def\col{{\rm col\,}}
\def\dim{{\rm dim\,}}

\def\graph{{\rm gr\,}}
\def\rank{{\rm rank\,}}
\def\tr{{\rm tr\,}}

\let\xker=\ker \def\ker{{\xker\,}}

\def\supp{{\rm supp\,}}
\def\ord{{\rm ord\,}}

\def\sign{{\rm sign}}

\def\loc{{\rm loc\,}}
\def\adj{{\rm adj\,}}
\def\dist{{\rm dist\,}}

\newtheorem{theorem}{Theorem}[section]
\newtheorem{proposition}[theorem]{Proposition}
\newtheorem{corollary}[theorem]{Corollary}
\newtheorem{lemma}[theorem]{Lemma}
\newtheorem{definition}[theorem]{Definition}

\newtheorem{remark}[theorem]{Remark}

\newtheorem{Hypothesis}[theorem]{Hypothesis}

\numberwithin{equation}{section}
\newcommand{\ba}{\begin{array}}
\newcommand{\ea}{\end{array}}
\newcommand{\bea}{\begin{eqnarray}}
\newcommand{\eea}{\end{eqnarray}}
\newcommand{\bead}{\begin{eqnarray*}}
\newcommand{\eead}{\end{eqnarray*}}
\newcommand{\be}{\begin{equation}}
\newcommand{\ee}{\end{equation}}
\newcommand{\bed}{\begin{displaymath}}
\newcommand{\eed}{\end{displaymath}}
\newcommand{\bl}{\begin{lemma}}
\newcommand{\el}{\end{lemma}}
\newcommand{\bt}{\begin{theorem}}
\newcommand{\et}{\end{theorem}}
\newcommand{\Label}{\label}
\newcommand{\bc}{\begin{corollary}}
\newcommand{\ec}{\end{corollary}}
\newcommand{\la}{\Label}

\newcommand{\br}{\begin{remark}}
\newcommand{\er}{\end{remark}}
\newcommand{\bd}{\begin{definition}}
\newcommand{\ed}{\end{definition}}

\newcommand{\slim}{\,\mbox{\rm s-}\hspace{-2pt} \lim}

\title{Perturbation determinants and trace formulas
for singular perturbations}

\author{Mark~Malamud\\
Institute of Applied Mathematics and Mechanics  \\
National Academy of science of Ukraine\\
R. Luxemburg  str. 74 \\
Donetsk 83114, Ukraine\\
E-mail: mmm@telenet.dn.ua \and
Hagen~Neidhardt\\
Institut f\"ur Angewandte Analysis und Stochastik\\
Mohrenstr. 39\\
D-10117 Berlin, Germany\\
E-mail: neidhard@wias-berlin.de
}

\begin{document}

\maketitle

\begin{abstract}
\noindent
We use the boundary triplet approach to extend the classical concept of perturbation determinants
to a more general setup. In particular, we examine the concept of perturbation determinants to pairs of proper
extensions of closed symmetric operators. For an ordered pair of extensions
we express the perturbation determinant in terms of the abstract Weyl function and the corresponding
boundary operators. A crucial role in our approach plays so-called almost solvable extensions.
We obtain  trace formulas for  pairs of self-adjoint, dissipative and other pairs of extensions
and express the spectral shift function in terms of the abstract  Weyl function and 
the characteristic function of almost solvable extensions.

We emphasize that for \emph{pairs of dissipative extensions our results are new even
for the case of additive  perturbations}. In this case we improve and complete  some classical 
results of M.G. Krein for pairs of self-adjoint and dissipative operators.
We apply the main results  to ordinary differential operators and to elliptic operators as well.\\


\end{abstract}

\newpage
{\footnotesize{
\tableofcontents}
}

\newpage

\section{Introduction}

The perturbation determinant (perturbation determinant) was introduced and used
by Krein in \cite{K53b,K59,K60,K62a,K64}. Independently
from Krein it was also introduced by Kuroda in \cite{Kur61}.
The perturbation determinant is an important tool
to study the spectral shift function for pairs of
self-adjoint operators \cite{BirKrei62,BirKrei63,K53b,K62a,K64} as well as
for non-selfadjoint operators \cite{K87}. Moreover, it was also used
to analyze certain other properties of non-selfadjoint operators as the completeness
of the root vectors etc, cf.  \cite{K59,K60}. Recently, perturbation
determinants in application to Schr\"odinger operators were studied in
\cite{Gesztesy2003,Gesztesy2005,Gesztesy2007,Gesztesy2007a,Gesztesy2007b,Gesztesy2008,Gesztesy2009,Gesztesy2012}.
The properties of perturbation determinants are summarized in
\cite{BY92b,GK69,Yaf92}.

A perturbation determinant relates
a scalar-valued holomorphic function to an ordered  pair of closed
linear operators $\{H',H\}$ defined on some Hilbert space $\gotH$.
In the simplest case if $H'$ and $H$ are bounded operators such
$V := H' - H$ is a trace class operator the perturbation determinant is defined by
\be\la{1.0}
\gD_{H'/H}(z) := \det(I + V(H -z)^{-1}), \quad z \in \rho(H).
\ee
Obviously, the definition extends to unbounded operators $H'$ and $H$
if $\dom(H') = \dom(H)$ and $V := \overline{H' - H}$ is a trace class
operator provided the resolvent set $\rho(H)$ is not empty.
In \cite{K62a,K87} Krein has introduced a class $\gotD$ of pairs of
densely defined closed operators $\{H',H\}$ given by\\[-1ex]

\noindent
(i) $\rho(H') \cap \rho(H)$ is not empty,\\[-1ex]

\noindent
(ii) $\dom(H') = \dom(H)$,\\[-1ex]

\noindent
(iii) $(H'-H)(H-z)^{-1} \in \gotS_1(\gotH)$ for $z \in \rho(H)$.\\[-1ex]

\noindent
to which the definition \eqref{1.0} can be extended.
Pairs $\{H',H\}$ of densely defined closed
operators satisfying the conditions (i) and (ii) are called regular in
the following. Therefore, Krein's theory of perturbation determinant
makes sense for regular pairs of operators.

However, non-regular pairs $\{\wt A',\wt A\}$ naturally appear  in consideration
of boundary value problems for differential operators.
For instance, any pair $\{\wt A',\wt A\}$ of
different proper extensions of a densely defined  closed symmetric operator $A$
is not regular: $\dom(\wt A') = \dom(\wt A)$ if and only if  $\wt A' = \wt A$. In the sequel  a pair
$\{\wt A',\wt A\}$ of closed operators is called singular if there is a densely
defined closed symmetric operator $A$ such that $\wt A'$ and $\wt A$
are proper extensions of $A$. The aim of the paper is to extended
Krein's theory of perturbation determinants from regular to
singular pairs.

To do this there are several approaches. For instance (see  \cite{BY92b}, \cite{Yaf92})
for a pair $\{H',H\}$ of densely defined closed operators with the trace class resolvent difference
%
%
the following concept of generalized perturbation determinant was proposed
\be\la{1.01}
\wt \gD_{H'/H}(z,\xi) := \det\left((H' - z)(H' - \xi)^{-1}(H -\xi)(H - z)^{-1}\right),
\ee
$z \in \rho(H)$ and $\xi \in \rho(H')$. However, this  definition has few drawbacks.  
One of them is that it cannot be applied to boundary value problems.

In the present paper we propose a new  concept of generalized perturbation determinants
by applying the technique of
boundary triplets and the corresponding Weyl functions (see
Section \ref{sec.II}  for precise definitions). This new
approach to extension theory of symmetric operators has been
appeared and elaborated  during the last three decades
(see \cite{BGP07,DM87,DM91,GG84,Mal92} and references therein).

Recall that a triplet $\Pi = \{\cH, \Gamma_0,
\Gamma_1\}$, where $\cH$ is an auxiliary Hilbert space and
$\Gamma_0,\Gamma_1:\  \dom(A^*)\rightarrow \cH$ are linear
mappings,  is called a  boundary triplet for the adjoint $A^*$ of a symmetric operator $A$  if
the "abstract Green's identity"
    \be\label{Intro2.0}
(A^*f,g) - (f,A^*g) = (\gG_1f,\gG_0g)_{\cH} -
(\gG_0f,\gG_1g)_{\cH}, \qquad f,g\in\dom(A^*),
   \ee
holds and the mapping $\gG:=(\Gamma_0,\Gamma_1)^\top:  \dom(A^*)
\rightarrow \cH \oplus \cH$ is surjective.

 A boundary triplet $\Pi = \{\cH, \Gamma_0,
\Gamma_1\}$ for $A^*$ always exists whenever $n_+(A)= n_-(A)$, though it is not unique. Its
role in extension theory is similar to  that of a coordinate system in analytic
geometry. It  leads to a natural  parametrization   of the  set $\Ext_A$ of proper  extensions $\wt A$ of $A$ ($A\subset \wt A \subset A^*$)
by means of the set ${\wt\cC}(\cH)$ of linear relations (multi-valued operators) in $\cH$, see \cite{GG84} and \cite{DM91} for detailed treatments. 
In this paper  we consider  only the case of boundary  relation $\gT$ being  be the graph $\graph(B)$
of a closed linear operator $B$  in $\cH$, i.e. assume that the extension  $\wt A$ is given by
\be \label{Intro2.2}
A_B :=
A^*\!\upharpoonright\ker(\gG_1 - B\gG_0).
\ee
In this case the extension parameter is often called the boundary operator.

The main analytical tool in this approach is  the abstract Weyl
function $M(\cdot)$ which was introduced and studied in
\cite{DM91}. This  Weyl function plays a similar role in the
theory of  boundary triplets as the classical Weyl-Titchmarsh
function does in the  theory of Sturm-Liouville operators (see \cite{BMN02, DM91, Mal92, MalNei09}). 

This approach to the perturbation determinant for singular
pairs allows to express it in terms of the Weyl function
\cite{DM87,DM91,DM95} and the corresponding boundary operator. 


%
%
\bd\la{I.1}
{\em
Let $A$ be a densely defined closed symmetric operator in $\gotH$,
let $\Pi = \{\cH,\gG_0,\gG_1\}$ be a boundary triplet for $A^*$
and $M(\cdot)$ the corresponding  Weyl function. We say that the
ordered pair $\{\wt A',\wt A\}$ of proper extensions of $A$ belongs to the
class $\gotD^\Pi$ if $\wt A'$ and $\wt A$  admit representations \eqref{Intro2.2} 
with closed boundary operators $B'$ and $B$, respectively,
and the following conditions are valid

\item[\;\;\rm (i)] the set $\{z \in \rho(A_0): 0 \in \rho(B - M(z))\}$ is not empty,

\item[\;\;\rm (ii)] $\dom(B') = \dom(B)$,

\item[\;\;\rm (iii)] $(B' - B)(B - M(z))^{-1} \in \gotS_1(\gotH)$
    for $z \in \rho(A_0)$ obeying $0 \in \rho(B - M(z))$\\[-1.2ex]

\noindent
where $A_0$ is a self-adjoint extension of $A$ given by
$A_0 := A^*\upharpoonright\ker(\gG_0)$.
If $\{\wt A',\wt A\} \in \gotD^\Pi$, then the scalar-valued function
\be\la{1.102}
\gD^\Pi_{\wt A'/\wt A}(z) := \det(I_\cH + (B'-B)(B - M(z))^{-1})
\ee
defined for all those $z \in \rho(A_0)$ satisfying $0 \in \rho(B -
M(z))$ is called the perturbation determinant of the pair $\{\wt A',\wt A\}$
with respect to $\Pi$.
}
\ed
In the following we verify that the so defined
perturbation determinant $\gD^\Pi_{\wt A'/\wt A}(\cdot)$
fulfills all standard properties of Krein's perturbation
determinant listed in Appendix \ref{B}, see also \cite[Chapter 8.1.1]{Yaf92}.

The definition of the perturbation determinant for proper extensions
depends on the chosen boundary triplet $\Pi$. However, it turns out that if
$\{\wt A',\wt A\} \in \gotD^\Pi$ and $\{\wt A',\wt A\} \in \gotD^{\Pi'}$,
then the perturbation determinants  $\gD^{\Pi}_{\wt A',\wt A}(z)$
and $\gD^{\Pi'}_{\wt A',\wt A}(z)$ differ only by a multiplicative
complex constant. Hence, the definition of the perturbation in the sense of
Definition \ref{I.1} does not depend so much on $\Pi$ as it seems to
be at first glance.

Further, the natural question arises whether for two extensions $\wt
A'$ and $\wt A$ of $A$ satisfying
\be\la{1.103}
(\wt A' - z)^{-1} - (\wt A - z)^{-1} \in \gotS_1(\gotH), \qquad z
\in \rho(\wt A') \cap \rho(\wt A).
\ee
there is always a boundary triplet $\Pi$ for
$A^*$ such that $\{\wt A',\wt A\} \in \gotD^\Pi$.
We show that this holds if $\wt A$ is almost solvable,
cf. Section \ref{sec.III}.

The most simple expression for the perturbation determinant we obtain if the symmetric
operator $A$ has finite deficiency indices $n_+(A)=n_-(A)<\infty$. Namely, in
this case, as an immediate consequence of Definition \ref{I.1}, we
arrive at the following formula for the perturbation determinant
\be\label{0.5}
\gD^\Pi_{\wt A'/\wt A}(z) := \frac{\det(B'-  M(z))}{\det(B -
M(z))}, \qquad z \in \rho(\wt A') \cap \rho(\wt A).
\ee
For instance, let $A := A_{\min}$ be a minimal symmetric operator
generated in $L^2(\R_+)$ by the Sturm-Liouville differential
expression
\bed
\cL = -D^2+q,\qquad  q=\overline{q}\in l^1_{\rm loc}[0,\infty),
\eed
assuming  the limit point case at infinity. Let also
$L_j:=A_{h_j},j\in\{1,2\}$, be a proper extension of $A$ determined
by 
\bed
\dom(A_{h_j})=\{y\in\dom(A^*):\  y'(0)=h_j y(0)\},\qquad
j\in\{1,2\}.
\eed
Then
\bed
\Delta_{L_2/L_1}(z) = \frac{m(z)-h_2}{m(z)-h_1},
\eed
where $m(\cdot)$ is the Weyl function corresponding to the
Dirichlet extension.

In the case of  infinite deficiency indices $n_{\pm}(A) =\infty$ we consider
proper extensions  $\wt A',  \wt A \in \Ext_A$  satisfying
$\rho(\wt A') \cap \rho(\wt A) \not= \emptyset$  and ($A_0 := A^*\upharpoonright\ker(\gG_0)$)
\bed
(\wt A' - \zeta)^{-1} - (A_0 -\zeta)^{-1} \in \gotS_1(\gotH) \quad
\mbox{and} \quad (\wt A - \zeta)^{-1} - (A_0 -\zeta)^{-1} \in
\gotS_1(\gotH)
\eed
for  $\zeta \in \rho(\wt A') \cap \rho(\wt A) \cap \C_\pm$.
Then we can choose a boundary triplet $\Pi = \{\cH,\gG_0,\gG_1\}$  for $A^*$
such that $\wt A' = A_{B'}, \  \wt A = A_B$ with
$B', B \in \cC(\cH)$ and satisfying
$(B'-\mu)^{-1}, \  (B-\mu)^{-1} \in \gotS_1(\cH)$
for some $\mu \in \rho(B') \cap \rho(B) \cap \R$.

Denoting  by $M(\cdot)$ the corresponding  Weyl function,
we show that  there exists  a boundary triplet
 $\wt \Pi = \{\wt \cH,\wt \gG_0,\wt
\gG_1\}$ for $A^*$ such that $\{\wt A',\wt A\} \in \gotD^{\wt \Pi}$,
$\{\wt A',A_0\} \in \gotD^{\wt \Pi}$ and $\{\wt A,A_0\} \in
\gotD^{\wt \Pi}$, and  the perturbation determinant $\gD^{\wt \Pi}_{\wt A/A_0}(\cdot)$
 admits a representation
\be\la{0.11}
\gD^{\wt \Pi}_{\wt A/A_0}(z) = \det(I - (\mu - B)^{-1}(\mu - M(z))), \qquad z \in \rho(A_0).
\ee
Combining this formula with similar formula for $\gD^{\wt \Pi}_{\wt A'/A_0}(z)$
and using the chain rule we arrive at the formula for $\gD^{\wt \Pi}_{\wt A'/\wt A}(\cdot)$
    \be\label{0.7}
\gD^{\wt \Pi}_{\wt A'/\wt A}(z) = \frac{\det(I - (\mu -
B')^{-1}(\mu - M(z)))}{\det(I - (\mu - B)^{-1}(\mu - M(z)))},
\qquad z \in \rho(\wt A) \cap \C_\pm.
\ee
If $0\in \rho(B')\cap \rho(B)$,  we can put  $\mu = 0$ in both formulas \eqref{0.11} and
\eqref{0.7}.

Formula  \eqref{0.7} gives a desired extension of \eqref{0.5}  to the case $n_{\pm}(A) =\infty$
while it is also useful in the case $n_{\pm}(A) < \infty$ (see Section  \ref{sec.VII.2}).

Formulas \eqref{0.11} and  \eqref{0.7} can be applied to boundary value problems for partial differential equations. For instance,  consider the Schr\"odinger  symmetric  operator
in a domain $\Omega \subset \R^2$ with smooth compact boundary,
\be\label{0.14A}
{\mathcal A}:= - \Delta + q(x) =  \frac{\partial^2}{\partial x_1^2} + \frac{\partial^2}{\partial x_2^2} + q(x),\quad  q=\overline q  \in C^{\infty}({\overline\Omega}).
\ee
Consider  Robin-type realizations of the expression ${\mathcal A}$,
\be\label{6.12c}
\begin{split}
\wh A_{\gs_j} &:= A_{\max}\upharpoonright \dom(\wh A_{\gs_j}),\\
\dom(\wh A_{\gs_j}) &:= \{f\in H^2(\gO): G_1f = \gs G_0f\}, \quad j\in \{1,2\},
\end{split}
\ee
and denote  by $A_0$ the Dirichlet realization of ${\mathcal A}$ given by $\dom(\wh A_0)= \{f\in H^2(\gO): G_0f =0 \}$. Here $G_0$ and $G_1$ are trace operators,  $G_0u := \gamma_0u := u|_{\partial \Omega}$ and $G_1u :=
\gamma_0 \bigl({\partial u}/{\partial\nu} \bigr), \  u\in \dom(A_{\max}).$
It is known that $\wh A_0 = ({\wh A_0})^*$  and 
the realization $\wh A_{\sigma_j}$  is  closed whenever $\sigma_j\in C^2(\partial\Omega)$ and self-adjoint if $\sigma$ is real.

Denote by  $\wh \gs_j$  the multiplication operator induced by $\gs_j$  in $L^2(\partial\gO)$.
Assuming that  $0 \in \rho(\wh{A_{\gs_1}}) \cap \rho(\wh{A_{\gs_2}})\cap \rho(\wh{A_0})$, 
we indicate  a boundary triplet $\wt \Pi$ for $A_{\max}$ such that
 $\{\wh A_{\gs_j}, \wh{A_0}\} \in \gotD^{\wt \Pi}$ and the corresponding
perturbation determinants  $\gD^{\wt \Pi}_{\wh A_{\gs_j}/A_0}(\cdot)$  
and   $\gD^{\wt \Pi}_{\wh A_{\gs_2}/A_{\gs_1}}(\cdot)$ are given by
\bed
\gD^{\wt \Pi}_{\wh A_{\gs_j}/A_0}(z) =
{\det}_{L^2(\partial\gO)}\left(I - (\gL_{0,0}(z) - \gL_{0,0}(0))(\wh \gs_j - \gL_{0,0}(0))^{-1}\right),\; j\in \{1,2\}, \
\eed
and 
\bed
\gD^{\wt \Pi}_{\wh A_{\gs_2}/A_{\gs_1}}(z) =
\frac{{\det}_{L^2(\partial\gO)}\left(I - (\gL_{0,0}(z) - \gL_{0,0}(0))(\wh \gs_2 - \gL_{0,0}(0))^{-1}\right)}{{\det}_{L^2(\partial\gO)}\left(I - (\gL_{0,0}(z) - \gL_{0,0}(0))(\wh \gs_1 - \gL_{0,0}(0))^{-1}\right)},
\eed
$z \in \rho(\wh A_{\gs_1}) \cap \rho(\wh A_{\gs_2})\cap \rho(\wh{A_0})$, respectively.
Here $\gL_{0,0}(\cdot)$ is the Dirichlet to Neumann map restricted to ${L^2(\partial\gO)}$ (see Section 6.3 for details).

In a second part of the paper we use the definition of perturbation
determinants to find trace formulas for extensions, in particular, for
pairs of self-adjoint, accumulative and arbitrary closed
extensions.

The paper is organized as follows. In Section \ref{sec.II} we give a
short introduction into the theory of boundary triplets.
Section \ref{sec.III} is devoted to almost solvable extensions. The properties of
the perturbation determinant of pairs of proper extensions are
investigated and verified in Section \ref{sec.IV}. Trace formulas are proven in
Section \ref{sec.V}. In Section \ref{sec.VI} we compare the
results with those ones of regular pairs.
Finally, in Section \ref{sec.VII} we given certain examples of
perturbation determinants for partial differential operators.
In the Appendix we have collected several results for the convenience
of the reader which are necessary for proofs or for
understanding.

{\bf Notation.}
By $\gotH$ and $\cH$ we denote separable Hilbert spaces. Linear
operators in $\gotH$ or $\cH$ are always denoted by capital Latin
letters, for example by $H$, $A$, etc. By $\dom(A)$, $\ran(A)$ and
$\rho(A)$ we denote the domain, range and spectrum of $A$,
respectively. The symbols $\gs_p(\cdot)$, $\gs_c(\cdot)$ and
$\gs_r(\cdot)$ stand for the point, the continuous and the
residual spectrum of a linear operator. Recall that $z \in
\rho_c(H)$ if $\ker(H-z) = \{0\}$ and $\ran(H-z) \not=
\overline{\ran(H-z)} = \gotH$;  $z \in \gs_r(H)$ if $\ker(H-z) =
\{0\}$ and $\overline{\ran(H-z)} \not= \gotH$.

The set of bounded linear operators from  $\gotH_1$ to
$\gotH_2$ is denoted by $[\gotH_1,\gotH_2]$; $[\gotH] :=
[\gotH,\gotH]$. The Schatten-v.Neumann ideals                   
of compact operators on a Hilbert space $\gotH$ is
denoted by $\gotS_p(\gotH)$, $0 < p \le \infty$; in particular,
$\gotS_\infty(\gotH)$ denotes  the ideal of all compact
operators in $\gotH$.
The set of closed linear operators and the set  of
closed linear relations  in $\cH$ is denoted  by
$\cC(\cH)$ and $\widetilde\cC(\cH)$, respectively.  

$C^k_b({\Omega}),$\ $k\in\Z_+\cup \{\infty\}$,  the set
of $C^k$-functions bounded in $\Omega$ with all their derivatives
of order $\le k,$
 $C_b({\Omega}) :=C^0_b({\Omega})$; $C^k_{u}(\Omega)$,
$k\in\Z_+\cup \{\infty\}$,  the set of $C^k$-functions uniformly
continuous in $\Omega$ with all their derivatives of order $\le
k$, $C_{u}(\Omega):= C^0_{u}(\Omega)$; $C^k_{ub}(\Omega) :=
C^k_{u}(\Omega) \cap C^k_{b}(\Omega)$, $C_{ub}(\Omega):=
C^0_{ub}(\Omega)$; $H^s(\Omega) \ s\in \R$, the usual Sobolev
spaces.

\section{Preliminaries}\la{sec.II}

\subsection{Relations}\la{sec.II.1.1}

For any linear relation $\Theta$ in $\cH$  the adjoint
relation $\Theta^*\in\widetilde\cC(\cH)$ is defined by
\bed
\Theta^*= \left\{
\begin{pmatrix} k\\k^\prime
\end{pmatrix}: (h^\prime,k)=(h,k^\prime)\,\,\text{for all}\,
\begin{pmatrix} h\\h^\prime\end{pmatrix}
\in\Theta\right\}.
\eed
A linear relation $\Theta$ is called symmetric if
$\Theta\subset\Theta^*$ and self-adjoint if $\Theta=\Theta^*$. The
relation $\Theta$ is called dissipative if $\{h,h'\} \in \gT$
yields $\IM(h',h) \ge 0$ and accumulative if $-\Theta$ is
dissipative. If a dissipative (accumulative) relation $\gT$ does
not admit any closed dissipative (accumulative) extension, then
$\gT$ is called maximal dissipative or $m$-dissipative (maximal
accumulative or $m$-accumulative).

We usually consider $\cC(\cH)$ as a subset of $\widetilde\cC(\cH)$
by identifying an operator $T\in \cC(\cH)$ with its graph $\graph(T).$
In particular, an operator $T\in \cC(\cH)$ is called dissipative
if $\IM((Tf,f)) \ge 0$, $f \in \dom(T)$, and accumulative if
$\IM(Tf,f) \le 0$, $f \in \dom(T)$. A dissipative  operator $T$ is
called maximal dissipative ($m$-dissipative) if it does not admit
any closed dissipative extension.

\subsection{Boundary triplets and proper extensions}\la{sec.II.1.2}

Let $A$ be a densely defined closed symmetric operator
in  $\gH$ with equal deficiency indices $n_\pm(A)= \dim(\gotN_{\mp
i})$,\  $\gotN_z := \ker(A^* -z)$.
\bd\la{def2.2}
{\em
\item[\;\;\rm (i)]
 A closed extension $\wt A$ of $A$ is called a proper
extension, in short $\wt A \in \Ext_A,$ if $A \subseteq \wt A
\subseteq A^*;$

\item[\;\;\rm (ii)]
Two proper extensions $\wt A'$, $\wt A$ are called
disjoint if $\dom(\wt A')\cap \dom(\wt A) = \dom( A)$ and
transversal if in addition $\dom(\wt A') + \dom(\wt A) = \dom(
A^*)$.
}
\ed

Any  extension $\wt A= {\wt A}^*$ of $A$ is proper, $\wt A\in
\Ext_A$. Moreover, any dissipative (accumulative)
extension $\wt A$ of $A$ is proper,
(cf. \cite[Theorem III.1.3]{GG91}, \cite{Mal92}). In the following we
we call also $A$ and $A^*$ be proper extensions. This sounds strange
but makes sense with respect of Proposition \ref{prop2.1} below.
\bd[\cite{GG91}]\la{II.1}
{\em
A triplet $\Pi = \{\cH, \Gamma_0,
\Gamma_1\}$, where $\cH$ is an auxiliary Hilbert space and
$\Gamma_0,\Gamma_1:\  \dom(A^*)\rightarrow \cH$ are linear
mappings,  is called an (ordinary) boundary triplet for $A^*$ if
the "abstract Green's identity"
    \be\label{2.0}
(A^*f,g) - (f,A^*g) = (\gG_1f,\gG_0g)_{\cH} -
(\gG_0f,\gG_1g)_{\cH}, \qquad f,g\in\dom(A^*),
   \ee
holds and the mapping $\gG:=(\Gamma_0,\Gamma_1)^\top:  \dom(A^*)
\rightarrow \cH \oplus \cH$ is surjective.
}
\ed
A boundary triplet $\Pi=\{\cH,\gG_0,\gG_1\}$ for $A^*$ exists
whenever  $n_+(A) = n_-(A)$. Note also that  $n_\pm(A) =
\dim(\cH)$ and $\ker(\Gamma_0) \cap \ker(\Gamma_1)=\dom(A)$.

With any boundary triplet $\Pi$ one associates  two canonical
self-adjoint extensions $A_j:=A^*\!\upharpoonright\ker(\gG_j)$, $j = 0,1$.
Conversely, for any extension $A_0=A_0^*\in \Ext_A$ there exists
a (non-unique) boundary triplet $\Pi=\{\cH,\gG_0,\gG_1\}$ for $A^*$
such that $A_0:=A^*\!\upharpoonright\ker(\gG_0)$.

Using the concept of boundary triplets one can parameterize all
proper extensions of $A$ in the following way.
\begin{proposition}[{\cite{DM91, Mal92}}]\label{prop2.1}
Let  $\Pi=\{\cH,\gG_0,\gG_1\}$  be a
boundary triplet for  $A^*.$  Then the mapping
\begin{equation}\label{bij}
(\Ext_A\ni)\ \widetilde A \to  \Gamma \dom(\widetilde A)
=\{\{\Gamma_0 f,\Gamma_1f \} : \  f\in \dom(\widetilde A) \} =:
\Theta \in \widetilde\cC(\cH)
\end{equation}
establishes  a bijective correspondence between $\Ext_A$
and  $\widetilde\cC(\cH)$. We write  $\wt A = A_\Theta$
if $\wt A$ corresponds to $\Theta$ by \eqref{bij}. Moreover, the following holds:

\item[\;\;\rm (i)] $A^*_\Theta= A_{\Theta^*}$, in particular,
  $A^*_\Theta= A_\Theta$ if and only if $\Theta^* = \Theta$.

\item[\;\;\rm (ii)]  $A_\Theta$ is symmetric (self-adjoint) if
  and only if  $\Theta$ is symmetric (self-adjoint).

\item[\;\;\rm (iii)]
  $A_\Theta$ is $m$-dissipative  ($m$-accumulative) if
  and only if so is $\Theta$.

\item[\;\;\rm (iv)] The extensions $A_\Theta$ and $A_0$ are
disjoint (transversal) if and only if $\Theta\in \cC(\cH)$
($\Theta \in [\cH]$). In this case \eqref{bij} takes the form
\be \label{2.2}
A_\Theta := A_{\graph(\Theta)} =
A^*\!\upharpoonright\ker(\gG_1- \Theta\gG_0).
\ee
\end{proposition}

In particular, $A_j:=A^*\!\upharpoonright\ker(\gG_j) =
A_{\Theta_j},\ j\in \{0,1\},$ where $\gT_0:= \{0\} \times \cH$ and
$\gT_1 := \cH \times \{0\} = \graph(\bO)$ where $\bO$ denotes the zero
operator in $\cH$.

We note that $\wt\cC(\cH)$ contains the linear relations $\{0\} \times
\{0\}$ and $\cH \times \cH$. It turns out that
the corresponding closed extensions of $A$ are $A$ and $A^*$,
respectively. This is the reason why we include $A$ and $A^*$ into the
sets of proper extensions of $A$.

\subsection{Weyl functions and spectra of proper extensions}\la{sec.II.1.3}

It is well known that Weyl functions are an important tool in the
direct and inverse spectral theory of singular Sturm-Liouville
operators. In \cite{DM87,DM91,DM95} the concept of Weyl function
was generalized to the case of an arbitrary symmetric operator $A$
with  $n_+(A) = n_-(A)\le \infty$. Here following \cite{DM91}  we
briefly recall basic facts on Weyl functions and $\gamma$-fields
associated with a boundary triplet  $\Pi.$
\begin{definition}[{\cite{DM87,DM91}}]\label{Weylfunc}
{\em
Let $\Pi=\{\cH,\gG_0,\gG_1\}$ be a boundary triplet  for $A^*$ and
$A_0=A^*\!\upharpoonright\ker(\gG_0)$. The operator valued
functions $\gamma(\cdot) :\ \rho(A_0)\rightarrow  [\cH,\gH]$ and
$M(\cdot) :\ \rho(A_0)\rightarrow  [\cH]$ defined by
\be\label{2.3A}
\gamma(z):=\bigl(\Gamma_0\!\upharpoonright\gotN_z\bigr)^{-1}
\qquad\text{and}\qquad M(z):=\Gamma_1\gamma(z), \quad
z\in\rho(A_0), \ee
are called the $\gamma$-field and the  Weyl function,
respectively, corresponding to the boundary triplet $\Pi.$
}
\end{definition}
Clearly, the Weyl function can equivalently be defined by
\bed
M(z)\Gamma_0 f_z = \Gamma_1f_z,\qquad  f_z \in \gotN_z, \quad z\in\rho(A_0).
\eed
The $\gamma$-field $\gamma(\cdot)$ and the Weyl function
$M(\cdot)$ in \eqref{2.3A} are well defined. Moreover, both
$\gamma(\cdot)$ and $M(\cdot)$ are holomorphic on $\rho(A_0)$ and
the following relations hold
\be\la{2.5}
\gamma(z)=\bigl(I+(z-\zeta)(A_0-z)^{-1}\bigr)\gamma(\zeta),
\qquad z,\zeta\in\rho(A_0),
\ee
and
\begin{equation}\label{mlambda}
M(z)-M(\zeta)^*=(z-\overline\zeta)\gamma(\zeta)^*\gamma(z), \qquad
z,\zeta\in\rho(A_0).
\end{equation}
Identity \eqref{mlambda} yields that $M(\cdot)\in (R_{\cH})$, i.e.
$M(\cdot)$  is a  $[\cH]$-valued Nevanlinna function, that
is, $M(\cdot)$ is a ($[\cH]$-valued) holomorphic function on
$\C\backslash \R$ and
\bed
M(z)=M(\overline z)^*\qquad\text{and}\qquad
\frac{\IM(M(z))}{\IM(z)}\geq 0, \qquad
z\in \rho(A_0).
\eed
It follows also from \eqref{mlambda} that
$M(\cdot)$ satisfies $0\in \rho(\IM(M(z)))$
for all $z\in\C\backslash\R$.
\begin{proposition}[{\cite{DM87,DM91}}]\label{t1.12}
Let $A$ be a simple symmetric operator in $\gH$
simple closed densely defined symmetric operator in $\gH$
and let $\Pi=\{ {\cH}, \Gamma_0,\Gamma_1 \}$ be a boundary triplet
for $A^*$, $M(\cdot)$ the corresponding Weyl function. Assume that
$\Theta\in \wt\cC({\cH})$ and $z \in \rho(A_0)$. Then the
following holds:

\item[\;\;{\em (i)}] $z \in \rho(A_\gT)$ if and only if  $0 \in
\rho(\gT -M(z))$;

\item[\;\;{\em (ii)}] $z \in \gs_\tau(A_\gT)$ if  and only if $0
\in \gs_\tau(\gT - M(z))$,
$\tau = p,c,r$. Moreover,  $\dim(\ker(A_\gT -z)) = \dim(\ker(\gT -
M(z)))$.
\end{proposition}

\subsection{Krein-type formula for resolvents and comparability}\la{sec.II.1.4}

Let $\Pi=\{\cH,\gG_0,\gG_1\}$ be a boundary triplet for $A^*,$
$M(\cdot)$ and  $\gamma(\cdot)$ the  corresponding Weyl function
and the  $\gamma$-field, respectively.
For any proper (not necessarily self-adjoint)  extension
${\widetilde A}_{\Theta}\in \Ext_A$ with non-empty resolvent set
$\rho({\widetilde A}_{\Theta})$ the following Krein-type formula
holds (cf. \cite{DM87,DM91,DM95})
\begin{equation}\label{2.30}
(A_\Theta - z)^{-1} - (A_0 - z)^{-1} = \gamma(z) (\Theta -
M(z))^{-1} \gamma^*({\overline z}), \quad z\in \rho(A_0)\cap
\rho(A_\Theta).
\end{equation}
Formula \eqref{2.30} extends
the known Krein formula for canonical resolvents to the case of any $A_\Theta\in \Ext_A$
with $\rho(A_\Theta)\not = \emptyset.$  Moreover,
formulas  \eqref{bij}, \eqref{2.2} and \eqref{2.3A} express all
parameters  in \eqref{2.30} in terms of the boundary triplet $\Pi$ (cf.
\cite{DM87,DM91,DM95}). Namely, these expressions make it possible
to apply formula \eqref{2.30} to boundary value problems.

The following result is deduced from formula \eqref{2.30}.
\begin{proposition}[{\cite[Theorem 2]{DM91}}]\la{prop2.9}
Let $\Pi=\{\cH,\gG_0,\gG_1\}$  be a boundary triplet for $A^*$,
$\gT',\gT \in \wt\cC(\cH)$ and
$\rho(A_{\gT'})\cap \rho(A_{\gT})\not = \emptyset.$  If
$\rho({\gT'})\cap \rho({\gT})\not = \emptyset$, then for
any Neumann-Schatten ideal ${\mathfrak S}_p$, $1 \le p \le \infty$,
the following holds:

\item[\;\;\rm (i)] The relation
\be\la{2.31}
(A_{\gT'} - z)^{-1} - (A_{\gT} - z)^{-1}\in{\gotS}_p(\gH),
\quad  z\in  \rho(A_{\gT'})\cap \rho(A_{\gT}),
\ee
is equivalent to
\be\la{2.31a}
\bigl(\gT' - \zeta\bigr)^{-1}-
\bigl(\gT - \zeta \bigr)^{-1}\in{\mathfrak S}_p(\cH),  \quad
\zeta\in \rho({\gT'})\cap \rho({\gT}).
\ee
In particular,  $(A_{\gT}- z)^{-1} - (A_0 - z)^{-1}\in{\mathfrak
S}_p(\gH)$ for $z\in  \rho(A_{\gT})\cap \rho(A_{0})$ if and only
if $\bigl(\gT - \zeta\bigr)^{-1} \in{\mathfrak S}_p(\cH)$ for
$\zeta \in \rho(\gT)$.

\item[\;\;\rm (ii)] If $\dom(\gT') = \dom(\gT)$, then the
following implication holds:
\begin{equation}\label{2.32}
\overline{\gT' - \gT} \in{\mathfrak S}_p(\cH) \Longrightarrow
(A_{\gT'}-z)^{-1} - (A_{\gT}-z)^{-1}\in{\mathfrak S}_p(\gH),
\end{equation}
$z\in  \rho(A_{\gT'})\cap \rho(A_{\gT})$. In particular, if
$\gT',\gT\in[\cH]$, then \eqref{2.31}
is equivalent to $\gT' - \gT \in \gotS_p(\cH)$.
\end{proposition}

\section{Almost solvable extensions}\la{sec.III}

\subsection{Basic facts}\la{sec.III.1}

In the following $A$ denotes a densely defined closed symmetric
operator in $\mathfrak H$. The concept of  almost solvable extensions
of $A$ was introduced in \cite{DM85} (see also \cite{DM92a,DM92,DM95}).
Let us recall basic facts on these extensions and let us extend this concept to a family of proper
extensions.
\bd\label{def2.10}
{\em
\item[\;\;\rm (i)] An extension  $\wt A\in \Ext_A$ is called
almost solvable if there exists a self-adjoint extension $\wh A$ of
$A$ such that $\wh A$ and $\wt A$  are transversal, see Definition \ref{def2.2}(ii).

\item[\;\;\rm (ii)]
The family $\{\wt A_j\}^N_{j=1}$, $\wt A_j \in \Ext_A$, $j\in
\{1,\ldots,N\}$, $2 \le N \le \infty$, is called jointly almost solvable
if there exists a self-adjoint  extension $\wh A$ of $A$ such that
$\wh A$ is transversal to each $\wt A_j, \  j\in \{1, \ldots,
N\}$.

\item[\;\;\rm (iii)]
Let $\gl \in \R$. The family $\{\wt A_j\}^N_{j=1}$,
$\wt A_j \in \Ext_A$, $j\in \{1, \ldots, N\}$,  $2 \le N \le \infty$, is called
jointly almost solvable with respect to $\gl$ if there exists a
self-adjoint extension $\wh A$ of $A$ which is transversal to any
$\wt A_j$, $j\in \{1, \ldots, N\}$, $2 \le N \le \infty$, and
satisfies in addition the condition $\gl \in \rho(\wh A)$.
}
\end{definition}

We note that Definition \ref{def2.10}(i) coincides with Definition 3 of \cite{DM92}.
\bd\la{III.3}
{\em
Let $\wt A_j \in \Ext_A$, $  j\in \{1, \ldots, N\}.$ A boundary
triplet $\Pi =\{\cH,\gG_0,\gG_1\}$ for $A^*$ is called regular for
$\{\wt A_j\}^N_{j=1}$ if there exist  operators $B_j \in [\cH]$, $
j\in \{1, \ldots, N\}$, such that $\wt A_j = A_{B_j} :=
A^*\!\upharpoonright\ker(\gG_1- B_j\gG_0),$\  $j\in \{1, \ldots,
N\}.$
}
\ed
\begin{proposition}\la{prop2.11}
Let $\wt A_j \in \Ext_A$,\  $j\in \{1, \ldots, N\}.$

\item[\;\;\rm (i)] The family $\{\wt A_j\}^N_{j=1}$ is jointly
almost solvable if and only if there exists a boundary triplet
$\Pi=\{\cH,\Gamma_0,\Gamma_1\}$ for $A^*$ which is regular for
$\{\wt A_j\}^N_{j=1}$.

\item[\;\;\rm(ii)] The family $\{\wt A_j\}^N_{j=1}$ is jointly
almost solvable with respect to $\gl \in \R$ if and only if there
exists  a boundary triplet $\Pi=\{\cH,\Gamma_0,\Gamma_1\}$ for
$A^*$ which is regular for $\{\wt A_j\}^N_{j=1}$ and $\gl \in
\rho(A_0)$.
\end{proposition}
\begin{proof}
(i) The proof follows immediately from  Proposition \ref{prop2.1}(iv)
and \cite[Proposition 7.1]{DM95}.

(ii) By Definition \ref{def2.10}(iii) there exists  a self-adjoint
extension $\wh A$ which is transversal to $\wt A_j$, $j\in
\{1,2,\ldots,N\}$, and  such that $\gl \in \rho(\wh A)$. Choosing
a boundary triplet $\Pi = \{\cH,\gG_0,\gG_1\}$ for $A^*$ such that
$\wh A := A_0 = A^*\upharpoonright\ker(\gG_0)$
and applying Proposition \ref{prop2.1}(iv) we get the necessity.
Sufficiency is again implied by \cite[Proposition 7.1]{DM95}.
\end{proof}

Proposition \ref{prop2.11}(i)  makes it possible to introduce the
real part and the imaginary part of an almost solvable extension
$\wt A$.
\bd\la{defIII.4} \cite{DM85} {\em
Let $\wt A\in \Ext_A$ be almost solvable and let
$\Pi=\{\cH,\Gamma_0,\Gamma_1\}$ be a boundary triplet for $A^*$
which is regular for $\wt A,$ i.e.  $\wt A = A_B,\ B\in [\cH].$
Then the self-adjoint extensions $\wt A_R, \  \wt A_I\in \Ext_A$
defined by $\wt A_R := \wt A_{B_R}$ and $\wt A_I := \wt A_{B_I}$,
are called the real and the imaginary parts of $\wt A$,
respectively.
}
\ed

It can be shown (see \cite{DM92a}, \cite{DM95})  that the
definitions of $\wt A_R$ and $\wt A_I$ depend only on  $\wt A$ and
do not depend on the  choice of the regular boundary triplet.

It follows from Proposition \ref{prop2.11} that in the case
$n_+(A)= n_-(A)<\infty$  any $\wt A \in \Ext_A$ is almost
solvable. The following statement  demonstrates that for
$n_{\pm}(A)= \infty$ the class of almost solvable extensions is
also rather wide.
\begin{proposition}[{\cite[Theorem 1]{DM92}}]\la{II.9}
Let $\wt A \in \Ext_A$.  If the condition $(\rho(\wt A) \cup
\gs_c(\wt A))\cap \C_\pm \not = \emptyset$ is satisfied, then $\wt
A$ is almost solvable. In particular, $\wt A $ is almost solvable
if $\rho(\wt A) \cap \R\not = \emptyset$.
\end{proposition}
Recall that an extension $\wt A \in \Ext_A$ is called solvable if
$0\in \rho(\wt A)$. Hence any solvable extension  is almost
solvable. Furthermore, we note that the sufficient condition of
Proposition \ref{II.9} is not necessary. It might even happen that
$n_+(A_+)= n_-(A_-)<\infty$ and  $\wt A$ is almost solvable
although $(\rho(\wt A) \cup \gs_c(\wt A))\cap \C_\pm = \emptyset$.
Such extensions can easily be constructed for $A=A_+\oplus A_-$
where $A_{\pm}$  are  simple symmetric operators with deficiency
indices $n_+(A_+)= n_-(A_-)=1$ and $n_-(A_+)= n_+(A_-)=0$.

Finally, we indicate a criteria which easily follows from
\cite[Proposition 1.5]{DM92a}.
\bl\label{lem2.4}
Let $A$ be as above and let $\Pi = \{\cH,\gG_0,\gG_1\}$ be a
boundary triplet for $A^*$, $M(\cdot)$  the corresponding Weyl
function. Let $\wt A' := A_{\gT'}$ and $\wt A := A_{\gT}$ where
$\gT', \gT\in \widetilde\cC(\cH)$ and   $\zeta \in \rho(\wt A')
\cap \rho(\wt A).$  Then the  extensions $\wt A'$ and $\wt A$ are
transversal if and only if
\bed 0 \in \rho\left((\gT' - M(\zeta))^{-1} - (\gT -
M(\zeta))^{-1}\right).
      \eed
\el
\subsection{Compactness and almost solvable extensions}\la{sec.III.2}

In general, even two almost solvable extensions are
not necessarily jointly almost solvable. However, the following
result, which is very important in applications to the
perturbation determinants, is valid.
\begin{proposition}\la{prop3.3}
Let $A$ be a densely defined closed symmetric operator and let
$\wt A_j \in \Ext_A$, $j\in \{1, \ldots, N\}$, $2 \le N \le\infty$. If
at least one $\wt A_{j_0}$, $j_0 \in \{1,2,\ldots,N\}$, is almost
solvable and there exists a non-real $\zeta \in
\bigcap^N_{j=1}\rho(\wt A_j)$ such that
\be\label{7.7A}
(\wt A_j - \zeta)^{-1} - (\wt A_{j_0} - \zeta)^{-1} \in \mathfrak S_\infty(\gH),
\qquad j \in \{1,2,\ldots,N\},
\ee
then  the family $\{\wt A_j\}^N_{j=1}$ is jointly  almost
solvable.
\end{proposition}
\begin{proof}
Without loss of generality we assume $j_0 = 1$. Since $\wt A_1$ is
almost solvable there is a self-adjoint extension $\wh A$ of $A$
which is transversal to $\wt A_1$. We choose a boundary triplet
$\Pi = \{\cH,\gG_0,\gG_1\}$ for $A^*$ such that $\wh A = A_0 :=
A^*\upharpoonright\ker(\gG_0)$ and denote by $M(\cdot)$ the Weyl
function. By Proposition \ref{prop2.11}, there exists  a $B_1\in
[\cH]$ such that $\wt A_1 = A_{B_1}$. Since $\zeta \in \rho(\wt
A_1)$ one gets from Proposition \ref{t1.12} that $0 \in \rho(B_1 -
M(\zeta))$ which yields the existence of the operator $D_1 := (B_1
- M(\zeta))^{-1}$. Moreover, for any $j = 2,3,\ldots,N$, there are
closed relations $\gT_j$ in $\cH$ such that $\wt A_j = A_{\gT_j}$,
cf. Proposition \ref{prop2.11}. Again, by Proposition \ref{t1.12},
$0 \in \rho(\gT_j - M(\zeta))$ which shows that $D_j := (\gT_j -
M(\zeta))^{-1}$, $j \in \{2,3,\ldots,N\}$, exists and is bounded
From condition \eqref{7.7A} we get that
\be\la{3.20}
D_j - D_1 \in \gotS_\infty(\cH), \quad j =2,3,\ldots,N.
\ee
Notice, that $D_1$ is invertible, that is $0 \in \rho(D_1)$.

Without loss of generality we assume $\zeta \in \C_+$. We set
\bed
B_\mu := M_R(\zeta) + \mu^{-1}M_I(\zeta), \quad \mu \in \R \setminus \{0\},
\eed
where $ M_R(\zeta) := \RE(M(\zeta))$ and $M_I(\zeta) := \IM(M(\zeta))$. The operators
$B_\mu$ are bounded and self-adjoint. Hence $\wh A_\mu := A_{B_\mu}$
defines a family of self-adjoint extensions of $A$. Obviously, we have
\bed
(B_\mu - M(\zeta))^{-1} = \frac{\mu}{1 - i\mu}\frac{1}{M_I(\zeta)}
\eed
where we have used that for non-real $z \in \C_+$ the operator
$M_I(z)$ is invertible. Since $M_I(z)$ is also non-negative we have
\bed
D_j - (B_\mu - M(\zeta))^{-1} =
\frac{1}{\sqrt{M_I(\zeta)}}\left(D'_j -
\frac{\mu}{1 - i\mu}\right)\frac{1}{\sqrt{M_I(\zeta)}},
\quad j = 1,2,\ldots,N,
\eed
where $D'_j := \sqrt{M_I(\zeta)}D_j\sqrt{M_I(\zeta)}$, $j = 1,2,\ldots,N$.
From \eqref{3.20} we immediately get that
\bed
D'_j - D'_1 \in \gotS_\infty(\cH), \quad j = 2,3,\ldots,N.
\eed
Since $D_1$ is invertible the operator $D'_1$ is also invertible, that is
$0 \in \rho(D'_1)$. Hence there is a neighborhood
$\cU$ of zero such that $\cU \subseteq \rho(D'_1)$. The set $\gS_j :=\gs(D'_j)
\cap \cU$ consists of isolated eigenvalues of $D'_j$ of finite
algebraic multiplicity, that is, the set  $\gS_j$ is countable for
each $j = 2,3,\ldots,N$. Hence the set $\gS := \bigcup^N_{j=2}\gS_j$
is countable. Setting $\zeta(\mu) := \mu(1 - i\mu)^{-1}$, $\mu \in \R
\setminus \{0\}$, one has $\lim_{\mu\to 0}\zeta(\mu) = 0$.
Since the curve $\zeta(\mu)$ is continuous there is a least one $\mu_0 \in
\R \setminus \{0\}$ such that $\zeta(\mu_0) \in \cU \setminus \gS
\subseteq \cU \setminus \gS_j$, $j = 2,3,\ldots,N$, which yields
$\zeta(\mu_0) \in \cU \cap \rho(D'_j)$, $j = 2,3\ldots,N$, and, of
course, $\zeta(\mu_0) \in \cU$. Hence
the operators $D'_j - \zeta(\mu_0)$, $j = 1,2\ldots,N$, are invertible which shows that
$0 \in \rho(D_j - (B_{\mu_0} - M(\zeta))^{-1})$, $j = 1,2,\ldots,N$.
Hence  $0 \in \rho\left((\gT_j - M(\zeta))^{-1} - (B_{\mu_0} - M(\zeta))^{-1}\right)$,
$j = 1,2,\ldots,N$.
Taking into account Lemma \ref{lem2.4} we complete the proof.
\end{proof}

\begin{proposition}\la{II.14}
Let $A$ be a densely defined closed symmetric operator and let
$\wt A_j \in \Ext_A$, $1 \le N < \infty$. If there is a
real $\gl$ such that $\gl \in \rho(\wt A_{j_0})$ for some $j_0
\in \{1,2,\ldots,N\}$ and $\zeta \in \cap^N_{j=1}\rho(\wt A_j)$
such that \eqref{7.7A} holds, then the family
$\{\wt A_j\}^N_{j=1}$ is jointly almost solvable with respect to $\gl$.
\end{proposition}
\begin{proof}
Without loss of generality we assume $j_0 = 1$. Since
$\gl \in \rho(\wt A_1)$ one gets that $\gl$ is a regular point of $A$.
Hence there is a neighborhood $\gd := (\gl - \epsilon,\gl
+\epsilon)$, $\epsilon > 0$, of $\gl$ such that
$\gd \subseteq \rho(\wt A_1)$ is a gap for $A$, that is,
\bed
\|\bigl(A-\gl)\bigr)f\|\ge \epsilon\|f\|, \qquad   f\in\dom(A).
\eed
From \cite{K47} we get that there is a self-adjoint extension
$\wh A$ of $A$ such that $\rho(\wh A) \supseteq \gd$. We choose a
boundary triplet $\Pi = \{\cH,\gG_0,\gG_1\}$ of $A^*$ with Weyl
function $M(z)$ such that $\wh A = A_0 =
A^*\upharpoonright\ker(\gG_0)$. We recall that the Weyl function
$M(z)$ is well-defined and bounded on  $\gd$ because $\gd$ is a
gap of $\wh A = A_0$. Let $B_\mu = M(\gl) + \mu^{-1}$ where
$\mu \in \R \setminus \{0\}$. The operator $B_\mu$ is
bounded and self-adjoint for $\mu \in \R \setminus \{0\}$.
We set $\wh A_\mu := A_{B_\mu}$. Obviously, $\wh A_\mu$ is
self-adjoint. Moreover, we have $B_\mu - M(\gl) = \mu^{-1}$.
Hence $(B_\mu - M(\gl))^{-1} = \mu$ which yields that
$0 \in \rho(B_\mu - M(\gl))$ for any $\mu \in \R \setminus \{0\}$.
Using Proposition \ref{t1.12} we find $\gl \in \rho(\wh A_\mu)$ for any $\mu \not= 0$.

By Proposition \ref{prop2.1} we find a closed linear relation
$\gT_1$ in $\cH$ such that $\wt A_1 = A_{\gT_1}$. Moreover, from
Proposition \ref{t1.12} we get that $0 \in \rho(\gT_1 - M(\gl))$.
Hence $D_1 := (\gT_1 - M(\gl))^{-1}$ is a bounded operator.
Furthermore, we have
\bed
(\gT_1 - M(\gl))^{-1} - (B_\mu - M(\gl))^{-1} = D_1 - \mu.
\eed
Obviously, there is a $\mu \in \R \setminus
\{0\}$ such that $\mu \in \rho(D_1)$ holds.
Hence $0 \in \rho\left((\gT_1 - M(\gl))^{-1} - (B_\mu - M(\gl))^{-1}\right)$.
Applying Lemma \ref{lem2.4} we obtain that the extensions $\wt A_1$ and
$\wh A$ are transversal for sufficiently large $\mu$.

By assumption \eqref{7.7A} the sets $\gS_j := \gs(\wt A_j) \cap \gd$, $j = 2,3,\ldots,N$,
are countable which yields that $\gS = \bigcup^N_{j=2}\gS_j$ is
countable. Moreover, by Proposition \ref{prop2.1} there are closed
relations $\gT_j$ in $\cH$ such that $\wt A_j = A_{\gT_j}$, $j
=2,3,\ldots,N$. Let $\gl' \in (\gl,\gl + \epsilon) \setminus \gS
\subseteq (\gl,\gl + \epsilon) \setminus \gS_j$, $j
= 2,3,\ldots,N$, which yields $\gl' \in (\gl,\gl+\epsilon) \cap \rho(\wt
A_j)$, $j =2,3,\ldots,N$, and $\gl' \in \rho(\wt A_1)$. From Proposition
\ref{t1.12} we find $0 \in \rho(\gT_j - M(\gl'))$ for $j =
1,2,\ldots,N$. Hence the operators $D'_j := (\gT_j - M(\gl'))^{-1}$,
$j = 2,3,\ldots,N$, exist and are bounded.
We set $B'_\mu := M(\gl') + \mu^{-1}$ and $\wh{A'_\mu} := A_{B'_\mu}$,
$\mu \in \R \setminus \{0\}$. Obviously, $\wh{A'_\mu}$ is
self-adjoint. We have
\bed
(\gT_j - M(\gl'))^{-1} - (B'_\mu - M(\gl'))^{-1} = D'_j - \mu, \quad
j = 1,2,\ldots,N.
\eed
Hence, there is a sufficiently large real number $\mu_0 > 0$ such that
$0 \in \rho\left((\gT_j - M(\gl'))^{-1} - (B'_{\mu_0} - M(\gl'))^{-1}\right)$,
$j = 1,2,3,\ldots,N$.
By Lemma \ref{lem2.4} the self-adjoint extension $\wh{A'_{\mu_0}}$ is transversal to $\wt A_j$,
$j = 1,2,3,\ldots,N$.

It remains to show that $\gl \in \rho(\wh{A'_{\mu_0}})$. We have
$B'_{\mu_0} - M(\gl) = \mu^{-1}_0 + M(\gl') - M(\gl) \ge \mu^{-1}_0$
where we have used that $M(\gl) \le M(\gl')$ for $\gl,\gl' \in \gd$ and $\gl < \gl'$.
Hence $0 \in \rho\left( B'_{\mu_0} - M(\gl)\right)$. From
Proposition \ref{t1.12} we obtain $\gl \in \wh{A'_{\mu_0}}$.
\end{proof}

\subsection{Characteristic function and almost solvable extensions}\la{sec.III.3}

It is known several approaches to the definition of the
characteristic function (CF) of an unbounded operator with
non-empty resolvent set. The most relevant to our considerations
definitions have been proposed in  \cite{Str60} and
\cite{DM92a,DM92}. In general, the CF might  have some exotic
properties.  However, it was shown in \cite{DM92a,DM92,DM85} that
the CF of an almost solvable extension of $A$ takes values in
$[\cH]$ and has some nice properties similar to that of the CF of
bounded operators (cf. \cite{Bro71}). We will not present a strict
definition of CF since in what follows we need only the following
formula expressed CF in terms of the  Weyl function.
\begin{proposition}[{\cite[Theorem 2]{DM92a}}]\la{II.15}
Let $A$ be  a densely defined closed symmetric operator and  let $\wt A$ be an almost solvable
extension of $A$.  Let also $\Pi=\{\cH,\Gamma_0,\Gamma_1\}$ be a
boundary triplet for $A^*$  which is regular for $\wt A$, i.e.
$\wt A = A_B = A^*\!\upharpoonright\ker(\gG_1- B\gG_0)$ and $B\in
[\cH].$ Then the characteristic function of the operator
$A_B$ admits the representation
\be\label{2.33}
W^\Pi_{\wt A}(z) := I +
2i|B_I|^{1/2}\bigl(B^*-M(z)\bigr)^{-1}|B_I|^{1/2}J, \quad z \in
\rho(\wt A^*) \cap \rho(A_0),
\ee
where  $B_I = J|B_I|$, $J = \sign(B_I)$, is the polar decomposition of $B_I := \IM(B)$.
\end{proposition}

It follows from \eqref{2.33} that $W^\Pi_{\wt A}(\cdot)$
takes values in $[\cH]$ and is $J$-contractive in
$\C_+$, respectively $J$-expansive in  $\C_-$. In particular,
it is contractive in $\C_+$ if ${\wt A} = {A}_B$ is $m$-dissipative,
that is, $B$ is $m$-dissipative,
cf. Proposition \ref{prop2.1}(iii). Notice that $J=I$ in this case.

\section{Perturbation determinants for extensions}\la{sec.IV}

\subsection{Elementary properties}\la{B.I}

Through this section we always assume that $A$ is a densely defined closed
symmetric operator in $\gH$ with equal deficiency indices.
We are going to show that the perturbation determinant for extensions $\gD^\Pi_{\wt A'/\wt A}(\cdot)$, cf.
Definition \ref{I.1}, has similar properties as the perturbation determinant
for additive perturbations $\gD_{H'/H}(\cdot)$.
\bl \la{def7.3a}
Let $\Pi = \{\cH,\gG_0,\gG_1\}$ be a boundary triplet for $A^*$,
$M(\cdot)$ the corresponding  Weyl function.   Let also $\wt A',
\wt A\in \Ext_A$ and  $\wt A' = A_{B'}$ and $\wt A = A_B$. If
$\{\wt A',\wt A\} \in \gotD^\Pi$, then the following holds:
\item[\rm\;\;(i)]
$\{z \in \rho(A_0): 0 \in \rho(B - M(z)\} = \rho(\wt A) \cap \rho(A_0)
\not= \emptyset$ and \eqref{1.103} holds for $\zeta \in \rho(\wt
A') \cap \rho(\wt A)$.

\item[\rm\;\;(ii)]
The perturbation determinant $\gD^\Pi_{\wt
A'/\wt A}(z)$ is well defined on the open set $\rho(\wt A) \cap
\rho(A_0)$ and is holomorphic there.

\item[\rm\;\;(iii)]
If $\Pi$ is regular for $\{\wt A',\wt A\}$, then $B',B \in [\cH]$ and
$B'-B \in \gotS_1(\cH)$.

\item[\rm \;\;(iv)]
If $n_{\pm}(A) < \infty$, then
\be\label{3.2A}
\gD^\Pi_{\wt A'/\wt A}(z) =
\frac{\det\bigl(B'-M(z)\bigr)}{\det\bigl(B-M(z)\bigr)}, \qquad z
\in \rho(\wt A) \cap \rho(A_0).
\ee
\el
\begin{proof}
(i) This statement follows immediately from Proposition \ref{t1.12}(i).

To prove \eqref{1.103}  we start with the identity
\be\la{4.2AA}
(B - M(z))^{-1} - (B' - M(z))^{-1} =  (B' - M(z))^{-1}(B'-B)(B - M(z))^{-1}
\ee
valid for $z \in \rho(\wt A') \cap \rho(\wt A)$.
It follows  that the right-hand side is a trace class operator due to the assumption
$(B'-B)(B - M(z))^{-1} \in \mathfrak S_1(\cH)$.
Now \eqref{1.103}  is implied by combining
Krein-type formula \eqref{2.30} with the identity \eqref{4.2AA}.

(ii) By Proposition \ref{t1.12} we have $0\in\rho\bigl(B-M(z)\bigr)$ if
$z\in\rho(A_B) \cap \rho(A_0)$. Taking into account Definition \ref{I.1} we find
that the  perturbation determinant is defined on $\rho(\wt A) \cap \rho(A_0)$. Moreover, since the operator-valued function
$(B - M(z))^{-1}$ is holomorphic on the open set $\{z \in \C: 0\in\rho\bigl(B-M(z)\bigr)\}$ the perturbation determinant
$\gD^\Pi_{\wt A'/\wt A}(\cdot)$ is holomorphic on $z\in\rho(A_B) \cap \rho(A_0)$.

(iii) If $\Pi$ is regular for $\{\wt A',\wt A\}$, then
$B',B \in [\cH]$ by definition. Therefore, by Proposition \ref{prop2.9}(ii),
the inclusion \eqref{1.103} is equivalent to
$B'- B \in \gotS_1(\cH)$.

(iv) If $n_{\pm}(A) < \infty$, then $\dim(\cH) = n_{\pm}(A) <
\infty$. Hence $B'$ and $B$ are bounded operators and \eqref{3.2A}
is implied by combining  the identity
\bed
I_\cH + (B'-B)(B - M(z))^{-1} = (B' - M(z))(B - M(z))^{-1}, \quad  z \in \rho(\wt A') \cap \rho(A_0).
\eed
with Proposition \ref{classprop}(ii).
\end{proof}
\br\la{prob}
{\em
Lemma \ref{def7.3a} rises several problems.
\item[\;\;(a)]
By  Lemma \ref{def7.3a}(i) the assumption $\{\wt A',\wt A\} \in \gotD^\Pi$
yields  \eqref{1.103}.
Is the converse true? In other words, is there exits a boundary triplet $\Pi$ for $A^*$ such that
$\{\wt A',\wt A\} \in \gotD^\Pi$ whenever the extensions
 $\wt A',\wt A (\in \Ext_A)$  satisfy
condition \eqref{1.103}?
The answer is not obvious since due to \eqref{4.2AA} the inclusion \eqref{1.103}
is in general implied by the inclusion  $(B'-B)(B - M(z))^{-1} \in \mathfrak S_1(\cH)$ but not vice versa.

\item[\;\;(b)] The perturbation determinant $\gD^\Pi_{\wt A'/\wt A}(\cdot)$  depends
on a  chosen boundary triplet $\Pi$. What is character of this dependence?

\item[\;\;(c)]
In vie of  Lemma \ref{def7.3a}(ii),  $\gD^\Pi_{\wt A'/\wt A}(\cdot)$ is holomorphic on
$\rho(\wt A) \cap \rho(A_0)$.
We will show the next section  that $\gD^\Pi_{\wt A'/\wt A}(\cdot)$ admits holomorphic
continuation to the domain  $\rho(\wt A)$, as it takes place in the classical definition.
}
\er

\subsection{Existence of an appropriate boundary triplet}\la{sec.IV.2}

We are going to answer problem (a) of Remark \ref{prob}, at least, partially.
\begin{proposition}\la{IV.30}
Let $\wt A',\wt A \in \Ext_A$,  $\rho(\wt A') \cap \rho(\wt A) \not=
\emptyset$ and let \eqref{1.103} be valid. If $\wt A$ is an almost
solvable extension, then there exists a boundary triplet $\Pi =
\{\cH,\gG_0,\gG_1\}$ for $A^*$ such that $\{\wt A',\wt A\} \in
\gotD^\Pi$. Moreover, the boundary triplet $\Pi$ can be chosen regular for $\{\wt A',\wt A\}$.
\end{proposition}
\begin{proof}
By Proposition \ref{prop3.3}, the pair $\{\wt A',\wt A\}$ is
jointly  almost solvable. By  Proposition \ref{prop2.11}(i),  we can
find a regular boundary triplet $\Pi$. Finally, using Lemma
\ref{def7.3a}(iii)  we get $B'-B \in \gotS_1(\cH)$ which yields
$\{\wt A',\wt A\} \in \gotD^\Pi$.
\end{proof}
\bc\la{IV.40}
Let $\wt A',\wt A \in \Ext_A$,  $\rho(\wt A') \cap \rho(\wt A) \not=
\emptyset$ and let \eqref{1.103} be valid.

\item[\;\;\rm (i)] If $(\rho(\wt A) \cup \gs_c(\wt A)) \cap \C_+$
and $(\rho(\wt A) \cup \gs_c(\wt A)) \cap \C_-$ are not empty,
then there exists a boundary triplet $\Pi$ for $A^*$  which is
regular for $\{\wt A',\wt A\}$ and  such that  $\{\wt A',\wt A\}
\in \gotD^\Pi$.

\item[\;\;\rm (ii)] If $\rho(\wt A) \cap \R$ is not empty, then
there exists a boundary triplet $\Pi = \{\cH,\gG_0,\gG_1\}$ for
$A^*$  such that $\{\wt A',\wt A\} \in \gotD^\Pi$. Moreover, this
triplet can be chosen regular for $\{\wt A',\wt A\}$ such that the condition
$\gl \in \rho(A_0)$, $A_0 := A^*\upharpoonright\ker(\gG_0)$, for some
$\gl \in  \rho(\wt A) \cap \R$ is satisfied.
\ec
\begin{proof}
(i) The proof follows from Proposition \ref{II.9},
Proposition  \ref{prop3.3}, Proposition \ref{prop2.11}(i) and
Proposition \ref{IV.30}.

(ii) From Proposition \ref{II.9} and Proposition \ref{II.14} we get that $\{\wt A',\wt A\}$ is almost
solvable with respect to $\gl$. Proposition \ref{prop2.11}(ii) and Proposition \ref{IV.30} yield the
existence of the desired boundary triplet.
\end{proof}

\subsection{Dependence on the boundary triplet}

Let us answer problem (b) of Remark \ref{prob}.
\begin{proposition}\la{III.5}
Let $\wt A', \wt A\in \Ext_A$ and let
$\Pi = \{\cH,\gG_1,\gG_0\}$ and $\Pi' = \{\cH',\gG'_1,\gG'_0\}$  be boundary triplets for $A^*$ such that $\{\wt A',\wt A\} \in \gotD^\Pi$ and $\{\wt A',\wt
A\} \in \gotD^{\Pi'}$.  If $\wt A$ is almost solvable, then
there exists  a constant $c \in \C$ such that
\be\la{4.108}
\gD^\Pi_{\wt A'/\wt A}(z) = c\;\gD^{\Pi'}_{\wt A'/\wt A}(z), \quad
z \in \rho(\wt A) \cap \rho(A_0),
\ee
where $A_0 = A^*\upharpoonright\ker(\gG_0)$ and $A'_0 =
A^*\upharpoonright\ker(\gG'_0)$. If $\wt A'$ and $\wt A$
are self-adjoint, then $c$ is real.
\end{proposition}
\begin{proof}
By Proposition \ref{IV.30} there exists  a boundary triplet $\wt
\Pi = \{\wt \cH,\wt \gG_0,\wt \gG_1\}$ for $A^*$ such that $\{\wt
A',\wt A\} \in \gotD^{\wt \Pi}$ and  which is regular for $\{\wt
A',\wt A\}$. Since $\{\wt A',\wt A\} \in \gotD^\Pi$ and $\{\wt
A',\wt A\} \in \gotD^{\wt \Pi}$ there exist operators $B',B \in
\cC(\cH)$ and $\wt B', \wt B \in [\wt \cH]$ which satisfy the
requirements of Definition \ref{I.1}. In particular, one has $\wt
A' = A_{B'} = A_{\wt B'}$ and $\wt A := A_B = A_{\wt B}$. Since
$\wt \Pi$ is regular we have $\wt B' - \wt B \in \gotS_1(\wt \cH)$
(cf. Lemma \ref{def7.3a}). Moreover, we have
\bead
\gD^\Pi_{\wt A'/\wt A}(z) & = & \det(I + (B'-B)(B-M(z))^{-1}),
\quad z \in \rho(\wt A) \cap \rho(A_0),\\
\gD^{\wt \Pi}_{\wt A'/\wt A}(z) & = & \det(I + (\wt B'- \wt B)(\wt
B- \wt M(z))^{-1}), \quad z \in \rho(\wt A) \cap \rho(\wt A_0),
\eead
where $\wt A_0 = A^*\upharpoonright\ker(\wt \gG_0)$.

By \cite{GG91} (see also \cite[Proposition 1.7]{DM95}) there
exists a  $J$-unitary block-operator matrix $X \in [\cH \oplus
\cH]$ such that
\bed
\begin{pmatrix}
\wt \Gamma_1\\
\wt \Gamma_0
\end{pmatrix} =
\begin{pmatrix}
X_{11} & X_{12}\\
X_{21} & X_{22}
\end{pmatrix}
\begin{pmatrix}
\Gamma_1\\
\Gamma_0
\end{pmatrix}
\eed
where $J := \begin{pmatrix} 0 & -iI\\ iI & 0\end{pmatrix}$.  The
corresponding Weyl functions as well as boundary operators are related by (see \cite[Proposition
1.7]{DM95})
\be\label{5.19}
\begin{matrix}
\wt B  = \bigl(X_{11} B + X_{12}\bigr)\bigl(X_{21}B + X_{22}\bigr)^{-1},  \\[3mm]
\wt B' = \bigl(X_{11} B' + X_{12}\bigr)\bigl(X_{21}B' + X_{22}\bigr)^{-1}.
\end{matrix}
\ee
Since $\wt B'$ and $\wt B$ are bounded operators it follows from
\cite[Proposition 1.7]{DM95} that $0 \in \rho(X_{21}B + X_{22})$
and $0 \in \rho(X_{21}B' + X_{22})$. In particular, one has
that $\dom(X_{21}B + X_{22}) = \dom(B)$ and $\dom(B) = \ran((X_{21}B + X_{22})^{-1})$.
Similarly one gets that $0
\in \rho(X_{21}M(z) + X_{22})$ and
\be\label{5.18}
\wt M(z) = \bigl(X_{11}M(z) + X_{12}\bigr)
\bigl(X_{21}M(z) + X_{22}\bigr)^{-1}, \quad  z\in \C \setminus \R.
\ee
Taking this fact into account we have
\bead
\lefteqn{\wt B - \wt M(z)}\\
& = &
\bigl(X_{11} B + X_{12}\bigr)\bigl(X_{21}B + X_{22}\bigr)^{-1} -
\bigl(X_{11}M(z) + X_{12}\bigr)\bigl(X_{21}M(z) + X_{22}\bigr)^{-1} \\
& = &
X_{11}\bigl(B - M(z)\bigr)\bigl(X_{21}B + X_{22}\bigr)^{-1} +\\
& &
\bigl(X_{11}M(z) + X_{12}\bigr)
\left\{\bigl(X_{21}B + X_{22}\bigr)^{-1} - \bigl(X_{21}M(z) + X_{22}\bigr)^{-1}\right\}\\
& =  &
X_{11}(B - M(z))\bigl(X_{21}B + X_{22}\bigr)^{-1} - \\
& &
\bigl(X_{11}M(z) + X_{12}\bigr)
\bigl(X_{21}M(z) + X_{22}\bigr)^{-1} X_{21}(B - M(z))\bigl(X_{21}B + X_{22}\bigr)^{-1}.
\eead
Hence
\bed
\wt B - \wt M(z) = Q(z) (B - M(z))\bigl(X_{21}B +
X_{22}\bigr)^{-1}, \quad z \in \C_\pm,
\eed
 where
\bed
Q(z) := X_{11} -
\bigl(X_{11}M(z) + X_{12}\bigr)\bigl(X_{21}M(z) + X_{22}\bigr)^{-1}
X_{21}
\eed
which is a well defined operator-valued function. This yields the representation
\bed Q(z) = (\wt B - \wt M(z))\bigl(X_{21}B + X_{22}\bigr)(B - M(z))^{-1},
\quad z \in \rho(\wt A) \cap \C_\pm.
\eed
We set
\bed
\Xi(z) := (B - M(z))\bigl(X_{21}B + X_{22}\bigr)^{-1}(\wt B - \wt M(z))^{-1},
\quad z \in \rho(\wt A) \cap \rho(\wt A_0),
\eed
and note that due to \eqref{5.19} $\Xi(\cdot)$  is a well-defined
family of bounded operators. Clearly,  $Q(z)\Xi(z) = I_{\wt \cH}$
and $\Xi(z)Q(z)  = I_\cH$ for $z \in \rho(\wt A) \cap \C_\pm$.
Hence
\be\la{3.6}
(\wt B - \wt M(z))^{-1} = \bigl(X_{21}B + X_{22}\bigr)(B -
M(z))^{-1}\Xi(z), \quad z \in \rho(\wt A) \cap \C_\pm.
\ee
Similarly we find
\be\la{3.7}
\wt B' - \wt M(z) = Q(z) (B' - M(z))\bigl(X_{21}B' +
X_{22}\bigr)^{-1}, \quad z \in \C_\pm.
\ee
Combining  \eqref{3.6} with  \eqref{3.7} we arrive at the representation
\bead
\lefteqn{
I + (\wt B' - \wt B)(\wt B - \wt M(z))^{-1} }\\
  & &
= Q(z) (B' - M(z)) \bigl(X_{21}B' + X_{22}\bigr)^{-1}
\bigl(X_{21}B + X_{22}\bigr)(B - M(z))^{-1}\Xi(z)
    \eead
for $z \in \rho(\wt A) \cap \C_\pm$. It follows that
\bead
\lefteqn{
I + (\wt B' - \wt B)(\wt B - \wt M(z))^{-1}}\\
& = &
Q(z) (B' - M(z))\times\\
& &
\left\{I + \bigl(X_{21}B' + X_{22}\bigr)^{-1}X_{21}(B - B')\right\} (B - M(z))^{-1}\Xi(z) \\
& = &
Q(z) (B' - M(z))(B - M(z))^{-1}\Xi(z)  + \\
& &
Q(z) (B' - M(z))\bigl(X_{21}B' + X_{22}\bigr)^{-1}X_{21}(B - B')(B - M(z))^{-1}\Xi(z) \\
& = &
I + Q(z)(B' - B)(B - M(z))^{-1}Q(z)^{-1} +\\
& &
Q(z) (B' - M(z))\bigl(X_{21}B' + X_{22}\bigr)^{-1}X_{21}(B - B') (B - M(z))^{-1}Q(z)^{-1}.
\eead
Hence
\bead
\lefteqn{\det(I + (\wt B' - \wt B)(\wt B - \wt M(z))^{-1})} \\
& = &
\det\Big(I + (B' - B)(B - M(z))^{-1} \\
& &
+ (B' - M(z))\bigl(X_{21}B' + X_{22}\bigr)^{-1}X_{21}(B -B')(B - M(z))^{-1}\Big).
\eead
Further, it is easily seen that
\bead
\lefteqn{
I + (B' - B)(B - M(z))^{-1} +}\\
& &
(B' - M(z))\bigl(X_{21}B' + X_{22}\bigr)^{-1}X_{21}(B -B')(B - M(z))^{-1}\\
& = &
\left(I + (B' - B)(B - M(z))^{-1}\right)\times\\
& &
\times\left(I + (B - M(z))\bigl(X_{21}B' + X_{22}\bigr)^{-1}X_{21}(B - B')(B - M(z))^{-1}\right).
\eead
By the multiplicative property for determinants we get
\bead
\lefteqn{
\det(I + (\wt B' - \wt B)(\wt B - \wt M(z))^{-1})}\\
& & \hspace{-4mm}
= \det(I+(B' - B)(B - M(z))^{-1})
\det(I+X_{21}(B -B')\bigl(X_{21}B' +
X_{22}\bigr)^{-1}).
\eead
Note, that the second determinant in the last formula is well
defined. Indeed, it follows from \eqref{3.7}  that the operator
valued function
\bed
(B' - M(z))\bigl(X_{21}B' + X_{22})^{-1} = Q(z)^{-1}(\wt B'- M(z))
\eed
takes values in $\cH$ for $z\in \rho(\wt A')\cap \C_{\pm}$.
Taking into account Definition \ref{I.1}(iii) we get $(B -B')(B'-M(z))^{-1} \in
\gotS_1(\cH)$,\   $z\in \rho(\wt A')\cap \C_{\pm}.$ Therefore
\bead
\lefteqn{
X_{21}(B -B')\bigl(X_{21}B' + X_{22})^{-1} }\\
& &
= X_{21}(B -B')(B'-M(z))^{-1}(B'-M(z))\bigl(X_{21}B' +
X_{22})^{-1} \in \gotS_1(\cH)
\eead
for   $z\in \rho(\wt A')\cap \C_{\pm}$ and  the determinant is
well-defined. Setting $\wt c := \det\left(I + X_{21}(B
-B')\bigl(X_{21}B' + X_{22}\bigr)^{-1}\right)$ we get $\gD^{\wt
\Pi}_{\wt A'/\wt A}(z) = \wt c\;\gD^{\Pi}_{\wt A'/\wt A}(z)$ for
$z \in \rho(\wt A) \cap \C_\pm$.

Similarly,  there exists  a constant $\wt c' \in \C$ such that
$\gD^{\wt \Pi}_{\wt A'/\wt A}(z) = \wt c'\;\gD^{\Pi'}_{\wt A'/\wt
A}(z)$ for $z \in \rho(\wt A) \cap C_\pm$. Setting $c := \wt
c'/\wt c$ we find $\gD^{\Pi}_{\wt A'/\wt A}(z) = c
\;\gD^{\Pi'}_{\wt A'/\wt A}(z)$ for $z \in \rho(\wt A) \cap
\C_\pm$.

It remains to prove that $c= \overline c$ whenever $\wt A' =
(\wt A')^*$  and $\wt A = (\wt A)^*$. We have
\bed
c \;\gD^{\Pi'}_{A'/\wt A}(\overline{z}) = \gD^{\Pi}_{\wt
A'/\wt A}(\overline{z}) = \overline{\gD^{\Pi}_{\wt A'/\wt A}(z)} =
\overline{c}\;\gD^{\Pi'}_{\wt A'/\wt A}(\overline{z}), \quad z
\in \C_\pm,
   \eed
which yields $c = \overline{c}$.

Finally we note that
\eqref{4.108} was proved  for $z \in \rho(\wt A) \cap \C_\pm$.
Since both sides admit an analytic continuation to $\rho(\wt A)
\cap \rho(A_0) \cap \rho(A'_0)$, the relation \eqref{4.108}
extends to this set too.
\end{proof}

\subsection{Domain of holomorphy}

We note that relation \eqref{4.108} is not completely satisfactory because it should be valid for $z \in \rho(\wt A)$.
To this end we have to answer problem (c) of Remark \ref{prob}.
\begin{proposition}\la{III.7}
Let $\wt A', \wt A\in \Ext_A$ and let $\Pi$ be a boundary triplet for
$A^*$  such that $\{\wt A',\wt A\} \in \gotD^\Pi$.  Then
$\gD^\Pi_{\wt A'/\wt A}(\cdot)$ admits a holomorphic continuation
to $\rho(\wt A)$.
\end{proposition}
\begin{proof}
By Lemma \ref{def7.3a},  the perturbation determinant $\gD^\Pi_{\wt A'/\wt A}(\cdot)$ is
well defined and holomorphic on $\rho(\wt A) \cap \rho(A_0)$.
Since $A_0 = A_0^*$ the determinant $\gD^\Pi_{\wt A'/\wt A}(\cdot)$ is
holomorphic on $\rho(\wt A) \cap \C_\pm$.

Further, let us assume that $\gl \in \rho(\wt A) \cap \R$. By
Lemma \ref{def7.3a}(i), the set $\rho(\wt A') \cap \rho(\wt A)$
is not empty and \eqref{1.103} holds.
By Proposition \ref{II.14}, the pair $\{\wt A',\wt A\}$ is jointly
almost solvable with respect to $\gl$. It follows from Proposition
\ref{prop2.11}(ii) that there exists  a regular boundary triplet
$\Pi_\gl = \{\cH^\gl,\gG^\gl_0,\gG^\gl_1\}$  such that $\gl \in
\rho(A_{0,\gl})$, $A_{0,\gl} :=A^*\upharpoonright\ker(\gG^\gl_0).$
According to Lemma \ref{def7.3a}(ii) the perturbation determinant $\gD^{\Pi_\gl}_{\wt
A'/\wt A}(\cdot)$ is holomorphic on $\rho(\wt A) \cap
\rho(A_{0,\gl})$. If $\gl' \in \rho(\wt A) \cap \R$, we find a
regular boundary triplet $\Pi_{\gl'}$ such that  $\gl' \in
\rho(A_{0,\gl'})$. By Lemma \ref{def7.3a}(ii),  the perturbation determinant
$\gD^{\Pi_{\gl'}}_{\wt A'/\wt A}(z)$ is holomorphic on $\rho(\wt
A) \cap \rho(A_{0,\gl'})$. By Proposition \ref{III.5}, there is a
constant $c \in\C$ such that $\gD^{\Pi_\gl}_{\wt A'/\wt A}(z) =
c\;\gD^{\Pi_{\gl'}}_{\wt A'/\wt A}(z)$  for $z \in \rho(\wt A)
\cap \C_\pm$. Since the right-hand side of this identity is
holomorphic on $\rho(\wt A) \cap \rho(A_{0,\gl'})$ the left-hand
side admits a holomorphic continuation to $\rho(\wt A) \cap
\rho(A_{0,\gl'})$. Since $\gl' \in \rho(\wt A)$ is arbitrary, the proof is complete.
\end{proof}

Proposition \ref{III.7} enables us to consider the perturbation determinant as a holomorphic function on $\rho(\wt A)$.
Doing so we immediately obtain the following improvement of Proposition \ref{III.5}.
\bc\la{IV.8a}
Let the assumptions of Proposition \ref{III.5} be satisfied. If the extension $\wt A$ is
almost solvable, then there exists a constant $c \in \C$ such that
$\gD^\Pi_{\wt A'/\wt A}(z) = c\gD^{\Pi'}_{\wt A'/\wt A}(z)$ for $z \in \rho(\wt A)$.
\ec
\begin{proof}
From Proposition \ref{III.7} we immediately get that the relation \eqref{4.108} extends to $\rho(\wt A)$.
\end{proof}

\subsection{Properties}

Let us show that the perturbation determinant $\gD^\Pi_{\wt A'/\wt A}(\cdot)$ fulfills properties similar to those
of Krein's determinant mentioned in Section \ref{sec.II.1.4}.
\begin{proposition}\label{III.8}
Let $\Pi = \{\cH,\gG_0,\gG_1\}$ be a boundary triplet for $A^*$,
$M(\cdot)$ the corresponding  Weyl function. Assume  that the proper
extensions  $\wt A, \wt A', \wt A''(\in \Ext_A)$  are disjoint
with $A_0,$ i.e. $\wt A = A_B$,  $\wt A' = A_{B'}$, $\wt A'' =
A_{B''}$ with $B, B', B''\in \cC(\cH)$.

\item[\rm\;\;(i)]
If $B',B \in \gotS_1(\cH)$, then $\{\wt A',\wt A\} \in \gotD^\Pi$ and
\be\la{4.102}
\gD^\Pi_{\wt A'/\wt A}(z) = \frac{\det(I_\cH - B'M(z)^{-1})}{\det(I_\cH - BM(z)^{-1})},
\qquad z \in \rho(\wt A) \cap \rho(A_1),
\ee
where $A_1 := A^*\upharpoonright\ker(\gG_1)$.

\item[\rm\;\;(ii)]
Let $\rho(\wt A')\cap\rho(\wt A) \not= \emptyset$. If $\{\wt A'',\wt A'\} \in \gotD^\Pi$ and $\{\wt A',\wt A\} \in
\gotD^\Pi$, then $\{\wt A'',\wt A\} \in \gotD^\Pi$ and
\be\label{3.15}
\gD^\Pi_{\wt A''/\wt A'}(z)\;\gD^\Pi_{\wt A'/\wt A}(z) = \gD^\Pi_{\wt
A''/\wt A}(z), \qquad    z \in \rho(\wt A')\cap\rho(\wt A).
\ee

\item[\rm\;\;(iii)]
Let $\rho(\wt A')\cap\rho(\wt A) \not= \emptyset$.
If $\{\wt A',\wt A\} \in \gotD^\Pi$, then $\{\wt A,\wt A'\} \in \gotD^\Pi$
and
\be\la{4.104}
\gD^\Pi_{\wt A'/\wt A}(z) \gD^\Pi_{\wt A/\wt A'}(z) = 1,
\quad z \in \rho(\wt A') \cap \rho(\wt A).
\ee

\item[\rm\;\;(iv)]
Let $\{\wt A',\wt A\} \in \gotD^\Pi$.
If $z_0 \in \rho(A_0)$ is either a regular point
or a normal eigenvalue of $\wt A'$ and $\wt A$, then
   \be\la{3.11}
\ord\left(\Delta^\Pi_{\wt A'/\wt A}(z_0)\right) = m_{z_0}(\wt A')
- m_{z_0}(\wt A).
   \ee
In particular, if $z_0 \in \rho(A_0)\cap \rho(\wt A)$, then
   \be\la{3.11A}
\ord\left(\Delta^\Pi_{\wt A'/\wt A}(z_0)\right) = m_{z_0}(\wt A'),
   \ee
cf. Appendix \ref{B}.

\item[\rm\;\;(v)] If $\{\wt A',\wt A\} \in \gotD^\Pi$, then
\be\la{4.106}
\frac{1}{\gD_{\wt A'/\wt A}(z)}\frac{d}{dz}\gD_{\wt A'/\wt A}(z) =
\tr((\wt A - z)^{-1} - (\wt A' - z)^{-1})
\ee
for $z \in \rho(\wt A') \cap \rho(\wt A)$, where
the right hand side   has sense by Lemma \ref{def7.3a}(i).

\item[\rm\;\;(vi)]
If $\{\wt A',\wt A\} \in \gotD^\Pi$ and $\{\wt A'^*,\wt A^*\} \in \gotD^\Pi$,
then
\be\la{4.107}
\gD^\Pi_{\wt A'^*/\wt A^*}(z) = \overline{\gD^\Pi_{\wt A'/\wt A}(\overline{z})},
\quad z\in\rho(\wt A^*).
\ee

\item[\rm\;\;(vii)]
If $\{\wt A',\wt A\} \in \gotD^\Pi$, then
\bed
\frac{\gD^\Pi_{\wt A'/\wt A}(z)}{\gD^\Pi_{\wt A'/\wt A}(\zeta)} =
\det(I_\cH + (M(z)-M(\zeta))(B - M(z))^{-1}(B'-B)(B' - M(\zeta))^{-1}),
\eed
for $z \in \rho(\wt A)$ and $\zeta \in \rho(\wt A') \cap \rho(\wt A)$.
\end{proposition}
\begin{proof}
(i) Obviously we have $\{\wt A',\wt A\} \in \gotD^\Pi$. Further, let $z \in \rho(A_0)$. One has $0 \in \rho(B - M(z))$ if and only if $1 \in
\rho(BM(z)^{-1})$. Hence
\bead
\lefteqn{
I_\cH + (B'-B)(B - M(z))^{-1} = (B' - M(z))(B - M(z))^{-1}}\\
& &
= (I_\cH - B'M(z)^{-1})(I_\cH - BM(z)^{-1})^{-1}, \qquad z \in
\rho(\wt A) \cap \rho(A_0).
\eead
Since $z \in \rho(A_1)$ if and only if $0 \in \rho(M(z))$
the right-hand side of \eqref{4.102} admits a
holomorphic continuation  to $\rho(\wt A) \cap \rho(A_1)$.

(ii) If $z \in \rho(\wt A') \cap \rho(\wt A) \cap \rho(A_0) \not= \emptyset$,
then
\bed
\{z \in \rho(A_0): 0 \in \rho(B' - M(z)) \;\wedge \;0 \in \rho(B - M(z))\} \not= \emptyset.
\eed
Hence $\{z \in \rho(A_0): 0 \in \rho(B' - M(z))\} \not=\emptyset$.
To check  Definition \ref{I.1}(ii)
we note that $\dom(B'') = \dom(B') = \dom(B)$. Using
\bed
(B'' - B)(B - M(z))^{-1} = (B'' - B')(B - M(z))^{-1} + (B' - B)(B - M(z))^{-1}
\eed
for $z \in \rho(\wt A') \cap \rho(A_0)$. Notice that $(B' - B)(B -
M(z))^{-1} \in \gotS_1(\cH)$ by assumption.
It remains to verify that  $(B'' -
B')(\wt B - M(z))^{-1} \in \gotS_1(\cH)$ for $z \in \rho(\wt A)
\cap \rho(A_0)$. Obviously, we have
\bead
\lefteqn{
(B'' - B')(B - M(z))^{-1}
= (B'' - B')(B'-M(z))^{-1}} \\
& &
+\  (B''-B')(B'-M(z))^{-1}(B'-B)((B-M(z))^{-1}
\eead
for $z \in \rho(\wt A') \cap \rho(\wt A) \cap \rho(A_0)$. Since
$(B'' - B')(B' - M(z))^{-1} \in \gotS_1(\cH)$ for $z \in \rho(\wt
A') \cap \rho(A_0)$ and $(B'-B)(B - M(z))^{-1} \in \gotS_1(\cH)$
for $z \in \rho(\wt A) \cap \rho(A_0)$ by assumption we find $(B''
- B)(B - M(z))^{-1} \in \gotS_1(\cH)$ for $z \in \rho(\wt A') \cap
\rho(\wt A) \cap \rho(A_0)$. Now  \eqref{3.15} follows immediately from
the definition for $z \in \rho(\wt A') \cap \rho(\wt A) \cap \rho(A_0)$.
Finally, using Proposition \ref{III.7} we can omit $\rho(A_0)$.

(iii) If $\rho(\wt A') \cap \rho(\wt A) \not= \emptyset$, then
$\rho(\wt A') \cap \rho(\wt A) \cap \rho(A_0) \not=
\emptyset$. Hence one has
$\{z \in \rho(A_0): 0 \in \rho(B' - M(z)) \wedge 0 \in \rho(B - M(z))\} \not= \emptyset$.
Therefore the conditions (i) and (ii) of Definition \ref{I.1} are satisfied for the
ordered pair $\{\wt A,\wt A'\}$. To prove  condition (iii) of Definition \ref{I.1} we use the representation
\bed
(B - B')(B' - M(z))^{-1} = -(B' -B)(B - M(z))^{-1}(B - M(z))(B' - M(z))^{-1}
\eed
for $z \in \rho(\wt A') \cap \rho(\wt A) \cap \rho(A_0)$. Since
$(B - M(z))(B' - M(z))^{-1}$ is a bounded operator we get $(B -
B')(B' - M(z))^{-1} \in \gotS_1(\cH)$ for  $z \in \rho(\wt A')
\cap \rho(A_0)$. To prove \eqref{4.104} it suffices to set $\wt
A'' = \wt A$ in \eqref{3.15}.

(iv)
From Proposition \ref{t1.12} and the Krein-type formula
\eqref{2.30} we find that $\nu_{z_0}(\wt A') = \nu_{z_0}(B' - M(z_0))$
and $\nu_{z_0}(\wt A) = \nu_{z_0}(B - M(z_0))$. Following the reasoning
of\cite[Chapter IV, Sec. 3.4]{GK69} we get
$\ord(\gD^\Pi_{\wt A'/\wt A}(z_0)) = \nu_{z_0}(B' - M(z_0)) -
\nu_{z_0}(B - M(z_0))$ which proves \eqref{3.11}.

(v) From formula \cite[(1.7.10)]{Yaf92} we get
\bea \la{4.107AA}
\lefteqn{ \frac{1}{\gD^\Pi_{\wt A'/\wt
A}(z)}\frac{d}{dz}\gD^\Pi_{\wt A'/\wt A}(z)
 }\\
& &
=\tr\left(\left (I +
(B'-B)(B-M(z))^{-1}\right)^{-1}\frac{d}{dz}(B'-B)(B -
  M(z))^{-1}\right)\nonumber
\eea
for $z \in \rho(\wt A') \cap \rho(\wt A) \cap \rho(A_0)$. From
\eqref{mlambda} we find
\bed
\frac{d}{dz}(B - M(z))^{-1} = (B - M(z))^{-1}\gga(\overline{z})^*\gga(z)(B - M(z))^{-1}
\eed
for $z \in \rho(\wt A) \cap \rho(A_0)$. Combining two last
formulas we obtain
\bead
\lefteqn{\hspace{-1cm}
\tr\left(\left (I + (B'-B)(B-M(z))^{-1}\right)^{-1}\frac{d}{dz}(B'-B)(B -
  M(z))^{-1}\right) }\\
  & &
=\tr\left(\gga(z)(B' - M(z))^{-1}(B' - B)(B -
M(z))^{-1}\gga(\overline{z})^*\right) \eead
for $z \in \rho(\wt A') \cap \rho(\wt A) \cap \rho(A_0)$. On the
other hand,  the Krein-type formula \eqref{2.30} yields
\bed
(\wt A' - z)^{-1} - (\wt A - z)^{-1} = \gga(z)(B' - M(z))^{-1}(B -
B')(B - M(z))^{-1}\gga(\overline{z})^*
\eed
for $z \in \rho(\wt A') \cap \rho(\wt A) \cap \rho(A_0)$.
Combining this  formula with \eqref{4.107AA} we arrive at
\eqref{4.106}.

(vi)
If $\{\wt A',\wt A\} \in \gotD^\Pi$, then $C(B - M(z))^{-1} \in
\gotS_1(\cH)$ for $z \in \rho(\wt A) \cap \rho(A_0)$ where $C := B'
-B$, $\dom(C) := \dom(B') = \dom(B)$. Since $\{\wt A'^*,\wt A^*\} \in \gotD^{\Pi}$
the operators $B'$ and $B$ are densely defined. Hence $B'^*$, $B^*$ and $C^*$ exist and
$\overline{(B^* - M(z)^*)^{-1}C^*} \in \gotS_1(\cH)$.
Setting $C_* := B'^*-B^*$, $\dom(C_*) = \dom(B'^*) =
\dom(B^*)$ we get $\dom(C^*) \supseteq \dom(C_*)$. Moreover, we
have
\bed
C^*(B^* - M(z)^*)^{-1} = C_* (B^* - M(z)^*)^{-1} \in \gotS_1(\cH)
\eed
for $z \in \rho(\wt A) \cap \rho(A_0)$.
Applying Corollary \ref{A.2} we obtain
\bead
\lefteqn{
\det(I_\cH +  \overline{(B^* - M(z)^*)^{-1}C^*})  }\\
   & &
= \det(I_\cH + C^*(B^* - M(z)^*)^{-1}) = \det(I_\cH + C_*(B^*
- M(z)^*)^{-1}) \eead
for $z \in \rho(\wt A) \cap \rho(A_0)$.
Since $((B' - B)(B - M(z))^{-1})^* = \overline{(B^* - M(z)^*)^{-1}C^*}$ we find
\bead
\lefteqn{
\overline{\gD^\Pi_{\wt A'/\wt A}(z)}
= \det(I_\cH +  \overline{(B^* - M(z)^*)^{-1}C^*})}  \\
& &
= \det(I_\cH + C_*(B^* - M(z)^*)^{-1}) = \gD^\Pi_{\wt A'^*/\wt A^*}(\overline{z})
\eead
for $z \in \rho(\wt A) \cap \rho(A_0)$ where we have used $M(z)^* =
M(\overline{z})$ for $z \in \rho(A_0)$. Replacing $\overline{z}$ by
$z$ it follows \eqref{4.107} for $z \in \rho(\wt A^*) \cap \rho(A_0)$.

(vii) The proof follows from (ii) and (iii).
\end{proof}

Propositions  \ref{III.7} and  \ref{III.8}
show that the perturbation determinant  $\gD^\Pi_{\wt A'/\wt A}(\cdot)$
has the  properties similar to that  of the classical  perturbation determinant.
%
%
\begin{proposition}\la{IV.8}
Let $\wt A \in \Ext_A$ and  let $\Pi = \{\cH,\gG_0,\gG_1\}$ be a
boundary triplet for $A^*$ such that $\wt A = A_B$, $B \in
\cC(\cH)$, and $M(\cdot)$ the corresponding  Weyl function.
If for some $\zeta \in \rho(\wt A) \cap \rho(A_0)$
\be\la{4.113} (\wt A - \zeta)^{-1} - (A_0 - \zeta)^{-1} \in
\gotS_1(\gotH), \qquad A_0 := A^*\upharpoonright\ker(\gG_0),
       \ee
then the following holds:

\item[\;\;\rm (i)] $B$ is discrete and   $(B - \mu)^{-1} \in
\gotS_1(\cH)$   for any  $\mu \in \rho(B)$.

\item[\;\;\rm (ii)] There exists a regular  for $\{\wt A,A_0\}$ boundary triplet $\wt \Pi =
\{\wt \cH,\wt \gG_0,\wt\gG_1\}$ for $A^*$.  

\item[\;\;\rm (iii)] If $\wt \Pi$ is a boundary triplet for $A^*$ such
  that $\{\wt A,A_0\} \in \gotD^{\wt \Pi}$, then
there exist  $\mu = \bar\mu \in \rho(B)$ and  $c \in \C$ such that
\be\la{4.11}
\gD^{\wt \Pi}_{\wt A/A_0}(z) = c\det(I - (\mu - B)^{-1}(\mu - M(z))), \quad z \in \rho(A_0).
\ee
If $0\in \rho(B)$, then one can put $\mu = 0$ in \eqref{4.11}.
\end{proposition}
\begin{proof}
(i) If  \eqref{4.113} is satisfied, then, by Proposition
\ref{prop2.9}(i),  $B$ is unbounded, its spectrum is discrete and $(B - \mu)^{-1}
\in \gotS_1(\cH)$ for $\mu\in \rho(B).$

(ii) Since $A_0= A_0^*$,  Proposition \ref{II.9} yields  that
$A_0$ is almost solvable. By Proposition \ref{IV.30}, there exists
a boundary triplet $\wt \Pi$ for $A^*$ which is regular for the pair $\{\wt A,A_0\}$, in particular,   $\{\wt A,A_0\} \in \gotD^{\wt \Pi}$.

(iii) Since $B$ is discrete, $\rho(B)\cap\R\not =\emptyset.$
Choosing $\mu\in \rho(B)\cap\R$ we  set $\gG'_1 := \gG_0$ and
$\gG'_0 := -(\gG_1 + \mu\gG_0)$ and note that $\Pi' =
\{\cH, \gG'_0,\wt \gG'_1\}$ is a boundary triplet for $A^*$
too.   Moreover, we have $\wt A = A_{\wt B}$ where $\wt B = (\mu -
B)^{-1}$ and $A_0 = A_\bO = A^*\upharpoonright\ker(\wt \gG_1)$
where $\bO$ is the zero operator in $\cH$. Thus, the boundary
triplet $\wt \Pi$ is regular for $\{\wt A,A_0\}$. The  Weyl
function $\wt M(\cdot)$ associated with $\wt \Pi$ is  $\wt M(z) =
(\mu - M(z))^{-1}$, $z \in \C_\pm$. Hence, by Definition \ref{I.1},  the perturbation determinant is
\be\label{4.11New}
\gD^{\wt \Pi}_{\wt A/A_0}(z) = \det\left(I - \wt B\wt{M^{-1}}(z)\right) = \det\left(I - \wt B(\mu - M(z))\right),
\quad z \in \C_\pm.
\ee
Finally, applying Proposition \ref{III.5} we arrive at \eqref{4.11}.
\end{proof}
\bc\la{IV.9}
Let  $\Pi = \{\cH,\gG_0,\gG_1\}$ be a boundary triplet for $A^*$,
$M(\cdot)$ the corresponding  Weyl function  and let $\wt A',\wt A
\in \Ext_A$ be such that $\wt A' = A_{B'}$, $\wt A = A_B$, where
$B', B \in \cC(\cH)$. If $\rho(\wt A') \cap \rho(\wt A) \not= \emptyset$ and
\be\label{4.15}
(\wt A' - \zeta)^{-1} - (A_0 -\zeta)^{-1} \in \gotS_1(\gotH) \quad
\mbox{and} \quad (\wt A - \zeta)^{-1} - (A_0 -\zeta)^{-1} \in
\gotS_1(\gotH)
\ee
for some $\zeta \in \rho(\wt A') \cap \rho(\wt A) \cap \C_\pm$,
then the following holds:

\item[\;\;\rm (i)] $B'$ and $B$ are discrete and there exists
$\mu \in \rho(B') \cap \rho(B) \cap \R$ such that $(B'-\mu)^{-1}
\in \gotS_1(\cH)$ and $(B-\mu)^{-1} \in \gotS_1(\cH)$.

\item[\;\;\rm (ii)] There exists  a boundary triplet $\wt \Pi =
\{\wt \cH,\wt \gG_0,\wt \gG_1\}$ for $A^*$, which can be chosen to be regular
for the family $\{\wt A',\wt A,A_0\}$.  

\item[\;\;\rm (iii)] If $\wt \Pi = \{\wt \cH,\wt \gG_0,\wt
\gG_1\}$ is a (not necessarily regular for  $\{\wt A',\wt A,A_0\}$) boundary triplet for $A^*$ such that $\{\wt A',\wt A\} \in \gotD^{\wt \Pi}$,
$\{\wt A',A_0\} \in \gotD^{\wt \Pi}$ and $\{\wt A,A_0\} \in
\gotD^{\wt \Pi}$, then the perturbation determinant $\gD^{\wt \Pi}_{\wt A'/\wt A}(\cdot)$ admits the representation
    \be\label{4.16}
\gD^{\wt \Pi}_{\wt A'/\wt A}(z) = c\;\frac{\det(I - (\mu -
B')^{-1}(\mu - M(z)))}{\det(I - (\mu - B)^{-1}(\mu - M(z)))},
\quad z \in \rho(\wt A) \cap \C_\pm.
        \ee
If $0\in \rho(B')\cap \rho(B)$, then we can put  $\mu = 0$ in
\eqref{4.16}.
\ec
\begin{proof}
(i) This statement follows immediately from Proposition \ref{IV.8}(i).

(ii) Since $A_0=A_0^*$, it is almost solvable. Therefore, it
follows from  Proposition \ref{prop3.3} and \eqref{4.15} that  the
system  $\{\wt A',\wt A,A_0\}$ is jointly almost solvable. By
Proposition \ref{prop2.11}(i),  there  exists a regular boundary
triplet $\wt \Pi$ for $\{\wt A',\wt A,A_0\}$. Finally, Proposition
\ref{prop2.9} yields the inclusions  $\{\wt A',\wt A\} \in \gotD^{\wt \Pi}$,
$\{\wt A',A_0\} \in \gotD^{\wt \Pi}$ and $\{\wt A,A_0\} \in
\gotD^{\wt \Pi}$.

(iii) As in the proof of Proposition  \ref{IV.8}  we introduce the boundary triplet $\wt \Pi =
\{\cH,\wt \gG_0,\wt \gG_1\}$ where $\wt \gG_1 := \gG_0$ and $\wt
\gG_0 := -(\gG_1 + \mu\gG_0)$. Clearly,  $\wt A' = A_{\wt B'}$,
$\wt B'  := (\mu - B')^{-1}$ and $\wt A = A_{\wt B}$, $\wt B =
(\mu -B)^{-1}$ and $A_0 = A_\bO$ where $\bO$ is the zero operator.
Hence the boundary triplet $\wt \Pi$ is regular for the set of operators  $\{\wt A',\wt
A,A_0\}$. Combining the chain rule (cf. \eqref{3.15})
with  Proposition \ref{IV.8}(iii) we arrive at \eqref{4.16}.
To complete the proof it remains to  apply Proposition \ref{III.5}.
\end{proof}

\section{Perturbation determinants and trace formulas}\la{sec.V}

\subsection{Pairs of selfadjoint extensions}\la{sec.V.1}

A spectral shift function has originally been introduced by I.M.
Lifshitz in a special case and by M.G. Krein \cite{K53} in the
general case. Namely, Krein \cite{K53} proved that for a pair
$\{H' = H + V, H\}$ of selfadjoint operators with $V\in
\gotS_1(\gotH)$ there exists  a unique real-valued function
$\xi(\cdot) \in L^1({\R})$ such that the following trace formula
holds
\be\la{5.1A}
\tr\bigl((H'-z)^{-1}-(H-z)^{-1}\bigr) =
-\int_{\R}\frac{\xi(t)}{(t-z)^2}dt, \quad z\in\rho(H')\cap\rho(H).
\ee
Formula \eqref{5.1A} has been extended in \cite{K62a} to a pair of
selfadjoint operators $\{H',H\}$ which are only resolvent
comparable, that is, $(H - \zeta)^{-1} - (H_0 -\zeta)^{-1} \in
\gotS_1(\gotH)$ for some $\zeta \in \rho(H') \cap \rho(H)$. In this case formula
\eqref{5.1A} remains valid. However,  one has only $\xi(\cdot)\in
L^1\bigl({\mathbb R},\tfrac{1}{1+t^2}dt\bigr)$ which yields that the spectral shift
function $\xi(\cdot)$ is not uniquely defined by \eqref{5.1A}: alongside
with $\xi(\cdot)$ any function $\xi(\cdot)+c,\ c\in{\mathbb R}$,
satisfies \eqref{5.1A} too. First we show that the converse is
also true.
\bl\la{IV.1}
Let $H$ and $H_0$ be selfadjoint operators which are resolvent
comparable. Assume that  there exist  real-valued functions $\wt
\xi(\cdot), \xi(\cdot)  \in L^1(\R,\tfrac{1}{1+t^2}dt)$ such that
the trace formula \eqref{5.1A} holds with both $\xi(\cdot)$ and
$\wt \xi(\cdot)$. Then $\wt \xi(t)- \xi(t) = c$ for a.e. $t
  \in \R$  where $c$ is a real constant.
\el
\begin{proof}
Let $\eta(t) := \wt \xi(t)-\xi(t)$, $t \in \R$. Then
\be\label{5.2A}
\int^{\infty}_{-\infty}\frac{\eta(t) dt}{(t-z)^2}=0,   \qquad
z\in{\mathbb C}_+\cup{\mathbb C}_-,
\ee
and  $\eta(\cdot) \in L^1(\R,\tfrac{1}{1+t^2}dt)$. We set
\be\label{5.3A}
{\cP}_{\eta}(z,{\overline z}) :=
\frac{1}{\pi}\int\frac{y \;\eta(t)dt}{|t-z|^2}=\frac{1}{2i\pi}\int
\bigl(\frac{1}{t-z}-\frac{1}{t-{\overline z}}\bigr)\eta(t)dt,
\ee
where $z=x+iy \in \C_\pm$. Differentiating ${\cP}_{\eta}(z,{\overline
z})$ with respect to $z$ and $\overline z$ and taking \eqref{5.2A}
into account we get
\bed
\frac{\partial}{\partial z}{\mathcal P}_{\eta}(z,{\overline z})=
\frac{\partial}{\partial{\overline z}}{\mathcal
P}_{\eta}(z,{\overline z})=0.
\eed
Thus, ${\cP}_{\eta}(z,{\overline z})$ is holomorphic
and anti-holomorphic in ${\C}_+\cup{\C}_-$.  Hence
${\cP}_{\eta}(z,{\overline z})=a= {\rm const.}$, $z \in{\C}_+$.
Applying the Fatou theorem to \eqref{5.3A} we get
\bed
\eta(t) = {\mathcal P}(t+i0,t-i0)= a= \overline a ={\rm const.}
\eed
for a.e. $t \in \R$.
\end{proof}

In the sequel we need the following  technical  lemma.
\bl\la{IV.2}
Let $\Pi = \{\cH,\gG_0,\gG_1\}$ be a  boundary triplet for $A^*$
and  $M(\cdot)$  the corresponding Weyl function. Let also $B$ be
a maximal accumulative operator, in particular, a selfadjoint
operator. Then the following statements are true:

\item[\;\;\rm(i)] If $V_+ \in \gotS_1(\cH)$ and $V_+\ge 0$, then
there exist  a constant $c_+ > 0$ and a non-negative function
$\xi_+(\cdot) \in L^1(\R,\tfrac{1}{1+t^2}dt)$ such that the
following representation holds
\be\la{4.4}
\det\left(I + V_+(B - M(z))^{-1}\right)= c_+
\exp\left\{\frac{1}{\pi}\int_R \left(\frac{1}{t-z}
-\frac{t}{1+t^2}\right)\xi_+(t) dt\right\},
\ee
$z \in \C_+$.

\item[\;\;\rm(ii)] If $V = V^* \in \gotS_1(\cH)$, then there exist
a constant $c > 0$ and a real-valued  function $\xi(\cdot) \in
L^1(\R,\tfrac{1}{1+t^2}dt)$ such that the  representation
\be\la{4.4a}
\det\left(I + V(B - M(z))^{-1}\right)= c\;
\exp\left\{\frac{1}{\pi}\int_R \left(\frac{1}{t-z}
-\frac{t}{1+t^2}\right)\xi(t) dt\right\}
\ee
holds for $z \in \C_+$.
\el
\begin{proof}
(i) We introduce the operator-valued Nevanlinna function
\bed
\gO_+(z) := I + \sqrt{V_+}(B - M(z))^{-1}\sqrt{V_+}, \quad z \in \C_+.
\eed
Since $\gO_+(z)$ is $m$-dissipative for $z \in \C_+$ and $0 \in
\rho(\gO_+(z))$, $z \in \C_+$, the operator-valued function
$\log(\gO_+(z))$ is well-defined by \eqref{2.17} for $z \in \C_+$.
Since $\log(\gO_+(z)) \in \gotS_1(\cH)$  \cite[Theorem 2.8]{GM00b}
guarantees the existence of a measurable function $\Xi_+(\cdot):
\R \longrightarrow \gotS_1(\cH)$ such that $\Xi(t) \ge 0$ for a.e.
$t \in \R$ and $\tr(\Xi(\cdot)) \in L^1(\R,\tfrac{1}{1+t^2}dt)$.
Moreover,  the following representation holds
\bed
\log(\gO_+(z)) = \gO_+ +
\frac{1}{\pi}\int_\R \left(\frac{1}{t-z} - \frac{t}{1+t^2}\right)\Xi_+(t)dt,
\quad z \in \C_+,
\eed
where the integral is taken in the weak sense and $\gO_+ =
\gO_+^* \in \gotS_1(\cH)$.  Setting  $\xi_+(t) := \tr(\Xi_+(t))$,
$t \in \R$, we define a non-negative function $\xi_+(\cdot)$
satisfying $\xi_+(\cdot) \in L^1(\R,\tfrac{1}{1+t^2}dt)$ and such
that
\bed
\tr(\log(\gO_+(z))) = \tr(\gO_+) + \frac{1}{\pi}\int_\R
\left(\frac{1}{t-z} - \frac{t}{1+t^2}\right)\xi_+(t) dt, \quad z
\in \C_+.
\eed
Taking into account \eqref{2.20} we  verify \eqref{4.4} with  $c_+ :=
\exp\{\tr(\gO_+)\} > 0$.

(ii) Using the decomposition $V = V_+ - V_-$, $V_\pm \ge 0$, we
set $B_- := B - V_-$.  It follows from the identity
$$
\left(I + V(B - M(z))^{-1}\right) \left(I + V_-(B_- -
M(z))^{-1}\right) =  I + V_+(B_- - M(z))^{-1}
$$
that
\be\label{4.20}
\det\left(I + V(B - M(z))^{-1}\right) = \frac{\det\left(I +
V_+(B_- - M(z))^{-1}\right)} {\det\left(I + V_-(B_- -
M(z))^{-1}\right)},\qquad z \in \C_+.
\ee
Combining \eqref{4.20} with  the representation \eqref{4.4} we
arrive at  \eqref{4.4a}.
\end{proof}

Lemma \ref{IV.2} implies the following representation theorem  for
a perturbation determinant.
\bt\la{V.30}
Let $\wt A',\wt A \in \Ext_A$ be selfadjoint extensions of $A$ such
that the pair $\{\wt A',\wt A\}$ is resolvent comparable, that is,
condition \eqref{1.103} is satisfied. Then the following holds:

\item[\;\;\rm (i)] There exists  a boundary triplet $\Pi =
\{\cH,\gG_0,\gG_1\}$ for $A^*$, which can be chosen regular for $\{\wt A',\wt A\}$,
such that $\{\wt A',\wt A\} \in \gotD^\Pi$.

\item[\;\;\rm (ii)] If $\{\wt A',\wt A\} \in \gotD^\Pi$, then
there exist a real-valued function $\xi(\cdot) \in
L^1\bigl(\R,\tfrac{1}{1+t^2}dt\bigr)$ and a  constant $c=\bar c$
such that the following representation
\be\label{5.37}
\gD^\Pi_{\wt A'/\wt A}(z) =
c\;\exp\left\{\frac{1}{\pi}\int_\R\left(\frac{1}{t-z}-\frac{t}{1+t^2}\right)\xi(t)dt\right\},
\quad z \in \C_\pm.
\ee
holds. Moreover, there exists  an integer $n \in \Z$ such that
\be\label{5.38}
\xi(t) = \lim_{\varepsilon\to +0}\IM(\log(\gD^\Pi_{\wt A'/\wt
A}(t+i\varepsilon))) + n\pi, \quad \text{for a.e.}\ t \in \R.
\ee

\item[\;\;\rm (iii)]
The trace formula
\be\label{5.39}
\tr\left((\wt A' -z)^{-1} -(\wt A - z)^{-1}\right) =
-\frac{1}{\pi}\int^{\infty}_{-\infty}\frac{\xi(t)}{(t-z)^2}dt,
\qquad z \in \C_\pm.
\ee
is valid where $\xi(\cdot)$ is given by (ii). Any real-valued function $\wt \xi(\cdot)$ satisfying
\eqref{5.39} differs from $\xi(\cdot)$ by an additive real constant.
\et
\begin{proof}
(i) This statement follows immediately is from  Corollary \ref{IV.40}(i).

(ii) At first, let us assume that $\Pi$ is regular for $\{\wt
A',\wt A\}$. Further, we assume that $\wt A' = A_{B'}$ and $\wt A = A_B$,
where $B'$ and $B$ are bounded selfadjoint operators.
By  Lemma \ref{IV.2}(ii) there exist  a  real constant $c$
and a real-valued function $\xi(\cdot) \in
L^1(\R,\tfrac{1}{1+t^2}dt)$ such that the following representation
\bea\label{5.9}
\lefteqn{
\gD^\Pi_{\wt A'/\wt A}(z) =
\det\left(I + (B'-B)(B - M(z))^{-1}\right) } \\
      & & = c\;
\exp\left\{\frac{1}{\pi}\int_R \left(\frac{1}{t-z}
-\frac{t}{1+t^2}\right)\xi(t) dt\right\}
\nonumber
\eea
holds for $z \in \C_\pm$. This proves \eqref{5.37}. If $\Pi$ is not regular, one obtains
\eqref{5.37}  by combining \eqref{5.9} with Proposition
\ref{III.5}.

It follows from \eqref{5.37}  that there exists  an integer  $n
\in \Z$ such that
\be\la{4.9}
\log(\gD^\Pi_{\wt A'/\wt A}(z)) = \log(|c|) + \frac{1}{\pi}
\int_\R\left(\frac{1}{t - z} - \frac{t}{1 + t^2}\right) \xi(t)dt +
in\pi,
 \ee
for $z \in \C_+$. At first we find that this representation is
true in a neighborhood of a point $z \in \C_+$ such that the value
$\gD^\Pi_{\wt A'/\wt A}(z)$ does not belong to  the negative
imaginary semi-axis. By analytical continuation it holds for $z
\in \C_+$. Taking the imaginary part of both sides of \eqref{4.9} and applying
Fatou's theorem we arrive at \eqref{5.38}.

(iii) From (ii) and Proposition \ref{III.8}(v) we immediately
obtain that the trace formula \eqref{5.39} holds
choosing $\wt \xi(t) := \xi(t)$, $t \in \R$. Applying Lemma \ref{IV.1} we get that any real-valued function
$\wt \xi(\cdot)$ obeying \eqref{5.39} differs from $\xi(\cdot)$ by a real constant.
\end{proof}

Any function $\wt \xi(\cdot) \in
L^1\bigl(\R,\tfrac{1}{1+t^2}dt\bigr)$ obeying
\eqref{5.39} is called the spectral shift function of the
pair $\{\wt A',\wt A\}$ of extensions of $A$. Theorem \ref{V.30}(iii)
shows that each pair $\{\wt A',\wt A\}$ admits many spectral shift functions.
\bt\label{th5.1A}
Let  $\wt A_2,\wt A_1, H\in \Ext_A$  be selfadjoint extensions of $A$.
Assume that the pairs $\{\wt A_1,H\}$ and $\{\wt A_2,H\}$ are resolvent comparable, that is,
conditions of type  \eqref{1.103} are satisfied. If for some
$\gl_0\in \rho(\wt A_2)\cap \rho(\wt A_1)\cap\R$ the condition
\be\label{5.12}
(\wt A_2 - \gl_0)^{-1} \le (\wt A_1 - \gl_0)^{-1}
\ee
is valid, then the
real-valued spectral shift functions $\xi_1(\cdot)$ and $\xi_2(\cdot)$ of the pairs
$\{\wt A_1,H\}$ and $\{\wt A_2,H\}$, respectively, can be chosen such that
$\xi_1(t) \le \xi_2(t)$ holds for a.e. $t \in \R$.
\et
\begin{proof}
By Corollary \ref{IV.40}(ii) there is a boundary triplet $\Pi
= \{\cH,\gG_0,\gG_1\}$ for $A^*$, which is regular for
$\{\wt A',\wt A,H\}$ and satisfies $\gl \in \rho(A_0)$, such that
$\{\wt A',H\} \in \gotD^\Pi$, $\{\wt A,H\} \in \gotD^\Pi$ and $\{\wt
A',\wt A\} \in \gotD^\Pi$. By Theorem \ref{V.30}(ii) there are
real-valued functions $\xi_j(\cdot) \in L^2(\R,\frac{1}{1+t^2}dt)$, $j =1,2$,
such that the perturbations determinants $\gD^\Pi_{\wt A_1/H}(z)$
admits the representations
\be\la{5.11}
\gD^\Pi_{\wt A_1/H}(z) = c_1\exp\left\{\frac{1}{\pi}\int_\R
\left(\frac{1}{t-z} - \frac{t}{1+t^2}\right)\xi_1(t)\right\},
\quad z \in \C_\pm,
\ee
$c_1 \in \R$, are valid. Taking into account the chain rule, cf. Proposition
\ref{III.8}(ii), we get
\be\la{5.13}
\gD^\Pi_{\wt A_2/H}(z) = \gD^\Pi_{\wt A_2/\wt A_1}(z)
\gD^\Pi_{\wt A_1/H}(z), \quad z \in \C_\pm.
\ee
Assume that $\wt A_j = A_{B_j}$ where $B_j$ are bounded operators in
$\cH$. Since $\gl_0 \in \rho(\wt A_2) \cap
\rho(\wt A_1) \cap \rho(A_0)$,   Proposition \ref{t1.12}(i) yields
$0 \in \rho(B_2 - M(\gl_0))$ and $0 \in  \rho(B_1 - M(\gl_0))$. Hence
\bead
\lefteqn{
0 \le (\wt A_1 - \gl_0)^{-1} - (\wt A_2 - \gl_0)^{-1} }\\
& & = \gga(\gl_0)\left((B_1 - M(\gl_0))^{-1} - (B_2 -
M(\gl_0))^{-1}\right)\gga(\gl_0)^*.
\eead
It follows from \eqref{5.12}  that $(B_1 - M(\gl_0))^{-1}
- (B_2 - M(\gl_0))^{-1} \ge 0$. Next we
introduce a new  boundary triplet $\wt
\Pi := \{\cH,\wt \gG_0,\wt \gG_1\}$ for $A^*$ by setting
\be\la{4.10a}
\wt \gG_1 := - \gG_0, \quad \wt \gG_0 := \gG_1 -
M(\gl_0)\gG_0.
\ee
Clearly, $\wt A_j = A_{\wt B_j} = A^*\upharpoonright\ker(\wt \gG_1 -
\wt B_j \wt\gG_0)$, where
$\wt B_j := -(B_j - M(\gl_0))^{-1}$, $j =1,2$.
Notice that $\{\wt A_2,\wt A_1\} \in \gotD^{\wt \Pi}$.  By
Definition \ref{I.1} one has
\bed
\gD^{\wt \Pi}_{\wt A_2/\wt A_1}(z) = \det(I + (\wt B_2 - \wt B_1)(\wt B_1 - \wt M(z))^{-1}),
\quad z\in \C_\pm,
\eed
where $\wt M(\cdot)$ is the Weyl function corresponding to the
boundary triplet $\wt \Pi$.  Since $\wt B_2 - \wt B_1 \ge 0$,  Lemma
\ref{IV.2}(i) implies  existence of  a non-negative function $
\xi(\cdot) \in L^1(\R,\tfrac{1}{1 + t^2}dt)$ such that the
representation
\be\la{5.16}
\gD^{\wt \Pi}_{\wt A_2/\wt A_1}(z) = \wt c \;
\exp\left\{\frac{1}{\pi}\int_\R \left(\frac{1}{t-z} -
\frac{t}{1+t^2}\right)\xi(t)dt\right\}, \quad z \in \C_+,
\ee
$\wt c > 0$, holds. By Proposition \ref{III.5} there is a real
constant $\wt c_{21}$ such that $\gD^\Pi_{\wt A_2/\wt A_1}(z) =
\wt c_{21}\;\gD^{\wt \Pi}_{\wt  A_2/\wt A_1}(z)$, $z \in \C_\pm$, which yields
\be\la{5.14}
\gD^\Pi_{\wt A_2/\wt A_1}(z) = c_{21} \; \exp\left\{\frac{1}{\pi}\int_\R \left(\frac{1}{t-z} -
\frac{t}{1+t^2}\right)\xi(t)dt\right\}, \quad z \in \C_\pm,
\ee
$c_{21} := \wt c_{21}\wt c \in \R$.
Inserting \eqref{5.11} and \eqref{5.14} into \eqref{5.13} we find
\be\la{5.15}
\gD^\Pi_{\wt A_2/H} = c_2 \; \exp\left\{\frac{1}{\pi}\int_\R \left(\frac{1}{t-z} -
\frac{t}{1+t^2}\right)\xi_2(t) dt\right\}, \quad z \in \C_\pm,
\ee
$c_2 := c_{21}c_1 \in \R$, $\xi_2(t) := \xi_1(t) + \xi(t)$, $t \in
\R$. Using Proposition \ref{III.8}(v) one gets from \eqref{5.11} and
\eqref{5.15} that $\xi_1(t)$ and $\xi_2(t)$ are spectral shift
functions for the pairs $\{\wt A_1,H\}$ and $\{\wt
A_2,H\}$. respectively. Since $\xi(t) \ge 0$ for a.e. $t \in \R$
we obtain $\xi_1(t) \le \xi_2(t)$ for a.e. $t \in \R$.
\end{proof}
\bc
Let $\wt A_2,\wt A_1 \in \Ext_A$ be selfadjoint extensions of $A$.
If the pair $\{\wt A_2,\wt A_1\}$ is resolvent comparable and condition
\eqref{5.12} is satisfied for some $\gl_0 \in \rho(\wt A_2) \cap
\rho(\wt A_1) \cap \R$, then the real-valued spectral shift function $\xi(\cdot)$ of the pair
$\{\wt A_2,\wt A_1\}$ can be chosen such that
$\xi(t) \ge 0$ for a.e. $t \in \R$.
\ec
\begin{proof}
By Corollary \ref{IV.40}(ii) there is a boundary triplet $\Pi$ for $A^*$
which is regular for $\{\wt A_2,\wt A_1\}$ and $\gl_0 \in \rho(A_0)$.
We set $B := \tfrac{B_1 + B_2}{2}$
and $H := A_B$. From $B_2 - B_1 \in \gotS_1(\cH)$ we get $B_2 - B \in
\gotS_1(\cH)$  and $B_1 - B \in \gotS_1(\cH)$. Hence the pairs
$\{\wt A_2, H\}$ and $\{\wt A_1,H\}$ are resolvent comparable.
By Theorem \ref{th5.1A} there are real-valued spectral shift functions
$\xi_2(\cdot)$ and $\xi_1(\cdot)$ of the pairs $\{\wt A_2, H\}$ and
$\{\wt A_1,H\}$, respectively, satisfying $\xi_1(t) \le \xi_2(t)$ for
a.e. $t \in \R$. Setting $\xi(t) := \xi_2(t) - \xi_1(t) \ge 0$, $t \in \R$, we
define a real-valued spectral shift function of the pair $\{\wt A_2,\wt A_1\}$.
\end{proof}

\subsection{Pairs of accumulative extensions}\la{sec.V.2}

We are going to prove a technical lemma which will be important to
prove a new type of trace formula in the following.
\bl\la{IV.3a}
Let $\Pi = \{\cH,\gG_0,\gG_1\}$  be a boundary triplet for $A^*,$
 $M(\cdot)$ the corresponding  Weyl function and let $B$ be a bounded accumulative
operator in $\cH$.

\item[\;\;\rm (i)] If $0 \le V_+ \le |B_I| = -B_I$, $B_I
:=\IM(B)\le 0$ and $V_+ \in \gotS_1(\cH)$, then  the function
$w_+(\cdot):= \det(I + iV_+(B - M(\cdot))^{-1})$ is holomorphic
and contractive in $\C_+.$ Moreover, there exist  a
\emph{non-negative function} $\eta_+(\cdot) \in
L^1(\R,\tfrac{1}{1+t^2}dt)$ and  a  number $\varkappa_+ \in \T$
such that the following  representation holds
  \be\la{4.180}
w_+(z) = \varkappa_+ \exp\left\{\frac{i}{\pi}\int_\R
\left(\frac{1}{t-z} - \frac{t}{1+t^2}\right) \eta_+(t)dt\right\},
\qquad z \in \C_+,
  \ee
where $\eta_+(t) = -\ln(|\det(w_+(t+i0))|)$ for a.e. $t \in \R$, i.e.
$w_+(\cdot)$ is an outer function.

\item[\;\;\rm (ii)] If $V \le |B_I| = -B_I$ and $V \in
\gotS_1(\cH)$, then there exist  a real-valued function
$\eta(\cdot) \in L^1(\R,\tfrac{1}{1+t^2}dt)$ and  a complex number
$\varkappa \in \T$ such that the perturbation determinant
$w(\cdot):= \det(I + iV(B - M(\cdot))^{-1})$ admits the following
representation
\be\la{4.190}
w(z) = \varkappa
\exp\left\{\frac{i}{\pi}\int_\R \left(\frac{1}{t-z} -
\frac{t}{1+t^2}\right) \eta(t)dt\right\}, \qquad  z \in \C_+,
\ee
where $\eta(t) = -\ln(|\det(w(t+i0))|)$  for a.e. $t \in \R$, i.e. $w(\cdot)$ is an outer function.
\el
\begin{proof}
Since $V_+\in \gotS_1$,  it admits the spectral decomposition
$V_+ = \sum_{k\in\N}\mu_k (\cdot,\psi_k)\psi_k$ where $\mu_k \ge
0$, $\{\mu_k\}_{k\in\N}\in l_1$ and $\{\psi_k\}_{k\in\N}$ is an
orthonormal system. We set
\bed
B_0 : = B_R + i(B_I + V_+)
\quad \mbox{and} \quad
B_l := B_0 - i\sum^l_{k=1} \mu_k (\cdot,\psi_k)\psi_k, \quad l \in \N.
\eed
Notice that $B_l = B + i\sum^\infty_{k= l+1}\mu_k(\cdot,\psi_k)\psi_k$ and
$\lim_{l\to\infty}\|B_l - B\|_{\gotS_1} = 0$ where $\|\cdot\|_{\gotS_1}$ denotes the trace norm.

By assumption,  $B_I + V_+ \le 0$ the operator  $B_0$ is
$m$-accumulative. Further let us introduce the operator-valued
function
\bed
W_l(z) := I + i\mu_lP_l(B_l - M(z))^{-1}P_l, \quad P_l :=
(\cdot,\psi_l)\psi_l,
\quad z \in \C_+, \quad l \in \N.
\eed
We set $w_l(z) := \det(W_l(z))$, $z \in \C_+$, $l\in \N$. Clearly,
\be\la{4.22}
w_l(z) = 1 + i\mu_l((B_l - M(z))^{-1}\psi_l,\psi_l), \quad z \in
\C_+, \quad l \in \N.
\ee
Further, we set
\bed
\gD_{B_{l-1}/B_l}(z) := \det\left(I + (B_{l-1} - B_l)(B_l -
M(z))^{-1}\right), \quad z \in \C_+, \quad l \in \N.
\eed
Since $B_{l-1} - B_l = i\mu_l(\cdot,\psi_l)\psi_l$, $l \in \N$,
we get $\gD_{B_{l-1}/B_l}(z) = w_l(z)$, $z \in \C_+$, $l \in \N$.
Due to the chain rule we obtain
\be\la{4.19a}
\gD_{B_0/B_l}(z) = \prod^l_{k=1}\gD_{B_{k-1}/B_k}(z) =
\prod^l_{k=1}w_k(z), \quad z \in \C_+, \quad l  \in \N.
\ee
Since $B_0 - B =V_+$ we have
$\det(W_+(z)) = \gD_{B_0/B}(z)$. By $\lim_{l\to\infty}\|B_l-B\|_{\gotS_1} = 0$ we get
from \eqref{4.19a} that
\be\la{4.20b}
w_+(z) := \det(W_+(z)) = \gD_{B_0/B}(z) =
\lim_{l\to\infty}\gD_{B_0/B_l}(z) =
\lim_{l\to\infty}\prod^l_{k=1}w_k(z)
\ee
for $z \in \C_+$. Note, that alongside with $W_+(\cdot),$  the
operator function  $W_l(\cdot),$ $l \in \N,$ is holomorphic and
contractive in $\C_+$. Hence $w_l(z) = \det(W_l(z))$, $l \in \N$, is  holomorphic and
contractive in $\C_+$,  thus $w_l(\cdot) \in H^\infty(\C_+)$.
Next we set
\bed
\theta_l(z) := \gD_{B_l/B_{l-1}}(z) := 1 - i\mu_l\left((B_{l-1} -
M(z))^{-1}\psi_l,\psi_l\right), \quad z \in \C_+, \quad l \in \N.
\eed
Notice that $\theta_l(z)$ is well defined since $B_{l-1}$ is accumulative.
Moreover, one has
\be\label{5.25}
\theta_l(z)w_l(z) =  w_l(z)\theta_l(z) = 1, \quad z \in \C_+,
\quad l \in \N.
\ee
Since $B_{l-1}$ is accumulative,  $\IM((B_{l-1} - M(z))^{-1}>0$,
hence
\bed
\RE(\theta_l(z)) = 1 + \mu_l\IM\left(\left((B_{l-1} -
M(z))^{-1}\psi_l,\psi_l\right)\right)>1, \quad z \in \C_+, \quad l
\in \N.
\eed
Combining this inequality with \eqref{5.25} we get
\bed
\RE(w_l(z)) = \frac{1}{|\theta_l(z)|^2}\RE(\theta_l(z)) >
\frac{1}{|\theta_l(z)|^2}>0, \quad z \in \C_+, \quad l \in \N.
\eed
By  \cite[Corollary II.4.8 ]{Gar81},  for each $l \in \N$ the
function $w_l(z)$ is an outer function.  According to
\eqref{4.15a} it admits  the representation
\bed
w_l(z) = \varkappa_l\;
\exp\left\{\frac{i}{\pi}
\int_R\left(\frac{1}{t-z} - \frac{t}{1+t^2}\right)\eta_l(t)dt\right\},
\quad \varkappa_l \in \T,
\eed
for $z \in \C_+$,  $l \in \N$, where $\eta_l(t) :=
-\ln(|w_l(t+i0)|)$, $t\in \R$. Hence
\be\label{5.26}
\prod^l_{k=1}w_l(z) = \\
\prod^l_{k=1}\varkappa_k\;
\exp\left\{\sum^l_{k=1}\frac{i}{\pi} \int_\R\left(\frac{1}{t-z}-\frac{t}{1+t^2}\right)\eta_k(t)dt\right\}
\ee
for $z \in \C_+$ and $l \in \N$. Now  \eqref{4.19a} yields
\bed
0 \le \left|\gD_{B_0/B_l}(z)\right| =
\exp\left\{-\sum^l_{k=1}\frac{1}{\pi}\int_R\frac{y}{(t-x)^2+y^2}\eta_k(t)dt\right\}
\eed
where $z = x+iy$. Since $w_k(z)$, $z \in \C_+$, is contractive, we
get $\eta_k(t) \ge 0$ for a.e. $t \in \R$. Combining
Corollary \ref{A.4} with \eqref{4.22} we obtain
\bed
-\int_R \ln(|w_k(t+i0)|)\frac{1}{1+t^2}dt \le 2\pi |w_k(i)-1| \le
2\pi\mu_k\frac{1}{\|\IM(M(i))\|}, \quad k \in \N.
\eed
Since
$\{\mu_k\}_{k\in\N}\in l_1$,  the Beppo Levi theorem yields
\bed
0\le \eta_+(t):= \sum_{k\in\N} \eta_k(t) = - \sum_{k\in\N}\ln(|w_k(t+i0)|)
\in L^1(\R,\tfrac{1}{1+t^2}dt),
\eed
and
\be\label{5.27}
\frac{i}{\pi}
\int_R\left(\frac{1}{t-z} - \frac{t}{1+t^2}\right)\eta_+(t)dt
= \sum^\infty_{k=1}\frac{i}{\pi} \int_R\left(\frac{1}{t-z} -
\frac{t}{1+t^2}\right)\eta_k(t)dt.
\ee
It follows from  \eqref{5.26} and \eqref{5.27} that
\bead
\lefteqn{
w_+(z) = \lim_{l\to\infty}\prod^l_{k=1}w_k(z) = }\\
& &
\left(\lim_{l\to\infty}\prod^l_{k=1}\varkappa_k\right)
\exp\left\{\frac{i}{\pi} \int_R\left(\frac{1}{t-z} -
\frac{t}{1+t^2}\right)\eta_+(t)dt\right\}, \quad z \in \C_+,
\eead
where $w_+(\cdot) = \det(W_+(\cdot))$  and $W_+(\cdot)$ is given
by \eqref{4.18b}. Hence the limit
$\varkappa_+ := \lim_{l\to\infty}\prod^l_{k=1}\varkappa_k \in \T$
exists and we   arrive at the representation \eqref{4.180}. Thus,
$w_+(\cdot)$ is the outer function and  $\eta_+(t) =
-\ln(|\det(w_+(t+i0))|)$ for a.e. $t \in \R$, see Appendix \ref{App.III}.

(ii) Let $V = V_+ - V_-$, $V_\pm \ge 0$. We set $B_- := B - iV_-$.
Since $(B_-)_I = B_I - V_- \le 0$, the operator $B_-$ is
accumulative. According to \eqref{4.20}
\bed
\det(I + iV(B - M(z))^{-1}) =
\frac{\det(I + iV_+(B_- - M(z))^{-1})}{\det(I + iV_-(B_- - M(z))^{-1})},
\quad z \in \C_+.
\eed
The assumption $V \le -B_I$ yields  $0 \le V_+ \le -B_I + V_- =
-(B_-)_I$. Applying (i) we get the existence of a complex number
$\varkappa_+ \in \T$ and a non-negative function $\eta_+(\cdot)
\in L^1(\R,\tfrac{1}{1+t^2}dt)$ such that the representation
\bed
\det(I + iV_+(B_- - M(z))^{-1}) = \varkappa_+
\exp\left\{\frac{i}{\pi}\int_\R \left(\frac{1}{t-z} - \frac{t}{1+t^2}\right)\eta_+(t)dt\right\},
\eed
is valid for $z \in \C_+$. From $0 \le V_- \le -B_I + V_- =
(B_-)_I$ and (i) we get the existence of a complex number
$\varkappa_- \in \T$ and a non-negative function $\eta_-(\cdot)
\in L^1(\R,\tfrac{1}{1+t^2}dt)$ such that the representation
\bed
\det(I + iV_-(B_- - M(z))^{-1}) = \varkappa_-
\exp\left\{\frac{i}{\pi}\int_\R \left(\frac{1}{t-z} - \frac{t}{1+t^2}\right)\eta_-(t)dt\right\},
\eed
holds for $z \in \C_-$. Setting $\varkappa :=
\varkappa_+/\varkappa_- \in \T$ and $\eta(t) := \eta_+(t) -
\eta_-(t)$, $t \in \R$, we arrive at  the representation
\eqref{4.190}.
\end{proof}

Next we apply Lemma  \ref{IV.3a}  to a pair $\{\wt A, H\}$  of
extensions of $A$ where $\wt A$ is maximal accumulative  and $H$
is selfadjoint.
\bt\la{V.60}
Let
$\wt A, H \in \Ext_A$, $H = H^*$  and let  $\wt A$   be a
$m$-accumulative extension. If the condition
\be\la{4.14}
(\wt A - \zeta)^{-1} - (H - \zeta)^{-1} \in
\gotS_1(\gotH), \quad \zeta \in \rho(\wt A) \cap \rho(H).
\ee
is satisfied, then the following holds:

\item[\;\;\rm (i)]
There exists  a boundary triplet $\Pi =
\{\cH,\gG_0,\gG_1\}$ of $A^*$, which can be chosen regular for $\{\wt A,H\}$,
such that $\{\wt A,H\} \in \gotD^\Pi$.

\item[\;\;\rm (ii)]
If $\{\wt A,H\} \in \gotD^\Pi$, then there
exists a complex constant $c$ and a complex-valued  function
$\go(\cdot) \in L^1(\R,\tfrac{1}{1+t^2}dt)$ satisfying
$\IM(\go(t)) \le 0$ for a.e. $t \in \R$ and  such that the representation
\be\label{5.57}
\gD^\Pi_{\wt A/H}(z)= c\;
\exp\left\{\frac{1}{\pi}\int_\R\left(\frac{1}{t-z} -
\frac{t}{1+t^2}\right)\go(t)dt\right\},\quad z\in \C_+,
\ee
is valid.

\item[\;\;\rm(iii)] The trace formula
\be\la{4.16a}
\tr\left((\wt A -z)^{-1} - (H - z)^{-1}\right) =
- \frac{1}{\pi}\int_\R \frac{\go(t)}{(t-z)^2}dt,\qquad z \in \C_+,
\ee
holds where $\go(\cdot)$ is given by {\rm (ii)}.
\et
\begin{proof}
(i) Since  $H=H^*$ is almost solvable and  $\rho(\wt A)\cap
\rho(H)\supset \C_+$
the existence of a regular  boundary triplet $\Pi$ for $A^*$ is
implied by   Proposition \ref{IV.30} and condition \eqref{4.14}.

(ii) Let us assume that $\Pi$ is regular for $\{\wt A,H\}$. If
$\wt A = A_B$ and $H = A_C$, where $B,C \in [\cH]$, then, by
Proposition \ref{prop2.1}, the  operator $B$ is accumulative and
$C=C^*$.  From $B- C \in \gotS_1(\cH)$, cf. Proposition  \ref{prop2.9}(ii), we get $B_R - C \in
\gotS_1(\cH)$, $B_R := \RE(B)$ and $B_I \in \gotS_1(\cH)$, $B_I :=
\IM(B) \le 0$.   Consider the selfadjoint extension $\wt A_R :=
A_{B_R}$. Hence,  $\{\wt A, \wt A_R\} \in \gotD^\Pi$
and $\{\wt A_R,H\} \in \gotD^\Pi$ as well as the following relation
\be\la{4.18a}
\gD^\Pi_{\wt A/H}(z) = \gD^\Pi_{\wt A/\wt A_R}(z)\gD^\Pi_{\wt A_R/H}(z), \quad z \in \C_+.
\ee
holds. Applying  Theorem \ref{V.30}(ii)
we get  the existence of a  real number $c_R$ and a
real-valued function $\xi(\cdot) = \overline \xi(\cdot) \in
L^1(\R,\tfrac{1}{1+t^2}dt)$ such that
\be\la{4.21}
\gD^\Pi_{\wt A_R/H}(z) = c_R
\exp\left\{\frac{1}{\pi}\int_R \left(\frac{1}{t-z} -
\frac{t}{1+t^2}\right)\xi(t)dt\right\}, \qquad z \in \C_+,
       \ee
Let us consider $\gD^\Pi_{\wt A_R/\wt A}(z)$, $z \in \C_+$. We have
\bed
\gD^\Pi_{\wt A_R/\wt A}(z) = \det(I -iB_I(B - M(z))^{-1}), \quad z \in \C_+.
\eed
By Lemma \ref{IV.3a}(i) there exists  a complex number $\varkappa
\in \T$ and a non-negative function $\eta_+(\cdot) \in
L^1(\R,\tfrac{1}{1+t^2}dt)$ such that the representation
\be\la{4.300}
\gD^\Pi_{\wt A_R/\wt A}(z) = \varkappa\;
\exp\left\{\frac{i}{\pi}\int_\R \left(\frac{1}{t-z} -
\frac{t}{1+t^2}\right) \eta_+(t) dt\right\}, \quad z\in \C_+,
\ee
holds. Hence,
\be\la{4.300a}
\gD^\Pi_{\wt A/\wt A_R}(z) = \overline{\varkappa}\;
\exp\left\{-\frac{i}{\pi}\int_\R \left(\frac{1}{t-z} -
\frac{t}{1+t^2}\right) \eta_+(t) dt\right\}, \quad z\in \C_+.
\ee
Setting $\go(t) := \xi(t) - i\eta_+(t)$, $t \in \R$, and
inserting \eqref{4.21} and \eqref{4.300a} into \eqref{4.18a}  we
arrive at  the representation \eqref{5.57} provided that $\Pi$ is
regular for $\{H,\wt A\}$. Notice $\IM(\go(t)) \le 0$ for a.e. $t \in
\R$. Finally, applying Proposition \ref{III.5} we verify the representation \eqref{5.57} for any
boundary triplet  $\Pi$ such that $\{\wt A,H\} \in \gotD^\Pi$.
Moreover, the condition $\IM(\go(t)) \le 0$ for a.e. $t \in \R$
remains valid.

(iii) The trace formula \eqref{4.16a} follows immediately from
\eqref{5.57} and Proposition \ref{III.8}(v).
\end{proof}

In the following a complex-valued function $\go(\cdot) \in
L^1(\R,\tfrac{1}{1+t^2}dt)$ such that the trace formula \eqref{4.16a} takes place
is called a spectral shift function for the pair $\{\wt A, H\}$
consisting of an accumulative and a selfadjoint extension.
\bt\la{IV.3}
Let $\wt A_1,\wt A_2 \in \Ext_A$ be $m$-accumulative and let $H \in
\Ext_A$ be a selfadjoint extension such that both pairs $\{\wt A_1,H\}$ and $\{\wt A_2,H\}$ be
are resolvent comparable, that is, condition \eqref{4.14} is
satisfied for both pairs. Then the following holds:

\item[\;\;\rm(i)]
If $\gl_0 \in \rho(\wt A_1) \cap \rho(\wt A_2)\cap
\R$ and the inequality
\be\la{4.18}
\RE((\wt A_2 - \gl_0)^{-1}) \le \RE((\wt A_1 - \gl_0)^{-1})
\ee
holds, then there are complex-valued spectral shift functions $\go_1(\cdot)$
and $\go_2(\cdot)$ of the pairs $\{\wt A_1,H\}$ and $\{\wt A_2,H\}$, respectively,
such that $\RE(\go_1(t)) \le \RE(\go_2(t))$ for a.e. $t \in \R$.

\item[\;\;$\rm (ii)$]
If $\gl_0 \in \rho(\wt A_1) \cap \rho(\wt
A_2)\cap \R$ and the inequality
\be\la{4.19}
\IM((\wt A_2 - \gl_0)^{-1}) \le \IM((\wt A_1 - \gl_0)^{-1}),
\ee
then there are complex-valued spectral shift functions $\go_1(\cdot)$
and $\go_2(\cdot)$ of the pairs $\{\wt A_1,H\}$ and $\{\wt A_2,H\}$, respectively,
such that $\IM(\go_1(t)) \le \IM(\go_2(t)) \le 0$ for a.e. $t \in \R$.

\item[\;\;$\rm (iii)$]
If both  conditions \eqref{4.18} and
\eqref{4.19} are valid, then there are complex-valued spectral shift functions
$\go_1(\cdot)$ and $\go_2(\cdot)$
of the pairs $\{\wt A_1,H\}$ and $\{\wt A_2,H\}$, respectively,
such that both inequalities $\RE(\go_1(t)) \le
\RE(\go_2(t))$ and $\IM(\go_1(t)) \le \IM(\go_2(t))\le 0$ are satisfied
for a.e. $t \in \R$.
\et
\begin{proof}
(i) By Proposition \ref{II.14}, the family $\{\wt A_2,\wt
A_1,H\}$ is jointly almost solvable with respect to $\gl_0$. Taking
into account Proposition \ref{prop2.11} there is a boundary
triplet $\Pi$ for $A^*$ which is regular for $\{\wt A_2,\wt
A_1,H\}$ and $\gl_0 \in \rho(A_0)$. Hence there exist bounded
operators $C,B_2,B_1 \in [\cH]$ such that $H= A_C$, $\wt A_j =
A_{B_j}$,  $j = 1,2$. By Proposition \ref{prop2.1}, $C=C^*$
and the operators $B_2,B_2$ are accumulative. Moreover,
Proposition \ref{prop2.9}(ii) and condition \eqref{4.14} yield
$B_j-C \in \gotS_1(\cH)$, $j =1,2,$ and $B_2-B_1 \in
\gotS_1(\cH)$. Hence $\{\wt A_j,H\} \in \gotD^\Pi$, $j =1,2$, and
$\{\wt A_2,\wt A_1\} \in \gotD^\Pi$. Notice that $B_{j,I} := \IM(B_j)
\in \gotS_1(\cH)$, $j =1,2$, $j =1,2$.

It follows from the Krein-type formula  \eqref{2.30} that
\bea\label{5.40}
\lefteqn{
(\wt A_1 - \gl_0)^{-1} - (\wt A_2 - \gl_0)^{-1} =}\\
& &
\gga(\gl_0)\left((B_1 - M(\gl_0))^{-1} - (B_2 -
M(\gl_0))^{-1}\right)\gga(\gl_0)^*.
\nonumber
\eea
Hence
\bed
\RE((\wt A_1 - \gl_0)^{-1}) - \RE((\wt A_2 - \gl_0)^{-1}) =
\gga(\gl_0)\left(\RE(\wt B_2) - \RE(\wt B_1) \right) \gga(\gl_0)^*
\eed
where $\wt B_j := -(B_j - M(\gl_0))^{-1}$, $j = 1,2$.
This identity yields the equivalence
\be\la{4.330a}
\RE((\wt A_2 -\gl_0)^{-1} \le \RE((\wt A_1 -\gl_0)^{-1}
\Longleftrightarrow  \RE(\wt B_1) \le \RE(\wt B_2).
\ee
Notice that the operators  $\wt B_1$ and $\wt B_2$ are accumulative
simultaneously with $B_1$ and $B_2$. Furthermore it holds
$\wt B_{j,I} := \IM(\wt B_{j,I}) \in \gotS_1(\cH)$, $j =1,2$.

Introducing the modified
boundary triplet $\wt \Pi = \{\cH,\wt \gG_0,\wt \gG_1\}$ defined
by \eqref{4.10a} we find that $\wt A_1 = A_{\wt B_1}$ and $\wt A_2
= A_{\wt B_2}$. Next we set $\wt A_3:= A_{\wt B_3}$, where
\be\la{4.330}
\wt B_3 := \wt B_{2,R} + i\wt B_{1,I},
\ee
and $\wt B_{j,R} := \RE(\wt B_j)$ and $\wt B_{j,I} = \IM(\wt
B_j)\le 0$, $j = 1,2$. By Proposition  \ref{prop2.1}(iii),
$A_{\wt B_3}$ is $m$-accumulative since $\wt B_{3,I}\le 0$.
One easily  verifies that $\{\wt A_2,\wt A_3\} \in \gotD^{\wt \Pi}$ and
and $\{\wt A_3,\wt A_1\} \in \gotD^{\wt \Pi}$. Indeed we have $\wt B_2 -
\wt B_3 = \wt B_{2,I} - \wt B_{1,I} \in \gotS_1(\cH)$ and
$\wt B_3 - \wt B_1 = \wt B_{2,R} - \wt B_{1,R} \in \gotS_1(\cH)$.
It follows from the chain rule \eqref{3.15} that
\be\la{4.25}
\gD^{\wt \Pi}_{\wt A_2/H}(z) = \gD^{\wt \Pi}_{\wt A_2/\wt A_1}(z)\gD^{\wt \Pi}_{\wt A_1/H}(z) =
\gD^{\wt \Pi}_{\wt A_2/\wt A_3}(z) \gD^{\wt \Pi}_{\wt A_3/\wt A_1}(z)\gD^{\wt \Pi}_{\wt
  A_1/H}(z),
\ee
$z\in \C_+$.

By Theorem \ref{V.60}(ii),  the perturbation determinant $\gD^{\wt \Pi}_{\wt
A_1/H}(\cdot)$ admits the representation
\be\la{4.35}
\gD^{\wt \Pi}_{\wt A_1/H}(z) = c_1 \exp\left\{\frac{1}{\pi} \int_R
\left(\frac{1}{t-z} - \frac{t}{1+t^2}\right)\go_1(t)dt\right\},
\qquad z \in \C_+,
\ee
where $c_1 \in \C$, $\go_1(\cdot) \in L^1(\R,\tfrac{1}{1+t^2}dt)$
and $\IM(\go_1(t)) \le 0$ for a.e. $t \in \R$.

Further, denoting by $\wt M(\cdot)$  the Weyl function
corresponding to $\wt \Pi$, we have, by definition,
\bed
\gD^{\wt \Pi}_{\wt A_3/\wt A_1}(z) = \det(I + (\wt B_3 - \wt
B_1)(\wt B_1 - \wt M(z))^{-1}),\qquad z \in \C_+.
\eed
Since the operators $\wt A_1$ and $\wt A_2$ are resolvent
comparable, Proposition \ref{prop2.9}(ii) combining with
\eqref{4.330} yields $\wt B_3 - \wt B_1 = \wt B_{2,R} - \wt
B_{1,R} \in \gotS_1(\cH)$.
It follows from \eqref{4.18}, \eqref{4.330a} and \eqref{4.330}
that  $\wt B_3 - \wt B_1 = \wt B_{2,R} - \wt B_{1,R} \ge 0$.
Therefore Lemma \ref{IV.2}(i) implies the existence of a
non-negative real number $\wt c_R > 0$ and  a non-negative
function $\xi(\cdot) \in L^1(\R,\tfrac{1}{1+t^2}dt)$ such that
the representation
\be\la{4.25a}
\gD^{\wt \Pi}_{\wt A_3/\wt A_1}(z) = \wt c_R
\exp\left\{\frac{1}{\pi}\int_\R\left(\frac{1}{t-z} -
    \frac{1}{1+t^2}\right)\xi(t)dt\right\}, \quad z \in \C_+,
\ee
holds. Since
\bed
\gD^{\wt \Pi}_{\wt A_2/\wt A_3}(z) = \det(I + (\wt B_2 - \wt
B_3)(\wt B_2 - \wt M(z))^{-1}),  \qquad z \in \C_+,
\eed
and $\wt B_2 - \wt B_3 = i(B_{2,I} - B_{1,I})$ and $B_{2,I} -
B_{1,I} \le -B_{1,I}$ we get from Lemma \ref{IV.3a}(ii) the
existence of a unimodular constant $\wt \varkappa \in \T$ and a
real-valued function $\eta(\cdot) \in L^1(\R,\tfrac{1}{1+t^2}dt)$
such that the representation
\be\la{4.370}
\gD^{\wt \Pi}_{\wt A_2/\wt A_3}(z) = \wt \varkappa
\exp\left\{\frac{i}{\pi}\int_R \left(\frac{1}{t-z} -
\frac{t}{1+t^2}\right)\eta(t)dt\right\}, \qquad z \in \C_+,
\ee
is valid. Inserting \eqref{4.35}, \eqref{4.25a} and \eqref{4.370}
into \eqref{4.25} and setting $c_2 := c_1\wt c_R \wt \varkappa \in
\C$ and $\go_2(t) := \go_1(t) + \xi_+(t) + i\eta(t)$, $t \in \R$,
we arrive at  the representation
\be\la{4.380}
\gD^{\wt \Pi}_{\wt A_2/H}(z) = c_2 \exp\left\{\frac{1}{\pi}\int_R
\left(\frac{1}{t-z} - \frac{t}{1+t^2}\right)\go_2(t)dt\right\}, \quad z
\in \C_+,
\ee
where $\RE(\go_1(t)) \le \RE(\go_1(t) + \xi_+(t)) = \RE(\go_2(t))$
for a.e. $t \in \R$.

(ii)
To handle the second case we use again the factorization
\eqref{4.25} but in slightly different manner.
From Theorem \ref{V.60}(ii) it follows the representation \eqref{4.380}
where $c_2 \in \C$, $\go_2(\cdot) \in L^1(\R,\tfrac{1}{1+t^2}dt)$
and $\IM(\go_2(t)) \le 0$ for a.e. $t \in \R$.
Now representation \eqref{4.25a} follows  from
Lemma \ref{IV.2}(ii) where, however, the function $\xi(\cdot)$ is
not necessarily non-negative. It follows  from \eqref{5.40} and
the assumption \eqref{4.19} that
\bed
\IM((\wt A_2 -\gl_0)^{-1}) \le \IM((\wt A_1 -\gl_0)^{-1})
\Longleftrightarrow \IM(\wt B_1) \le \IM(\wt B_2)
\eed
which yields $\wt B_{2,I} - \wt B_{1,I} \ge 0$. By Lemma \ref{IV.3a}(i),
representation \eqref{4.370} holds with a non-negative function
$\eta(\cdot)\ge 0.$  Inserting \eqref{4.380},
\eqref{4.25a} and \eqref{4.370} into \eqref{4.25} and setting
$\go_1(t) := \go_2(t) - \xi(t) - i\eta(t)$, $t \in \R$, we obtain the
representation \eqref{4.35}. From \eqref{4.35} and Proposition
\ref{III.8}(v) we obtain that $\go_1(t)$ is a spectral shift function
for the pair $\{\wt A_1,H\}$. Obviously, we have
$\IM(\go_1(t)) \le \IM(\go_1(t) + \eta(t))  = \IM(\go_2(t)) \le 0$ for a.e. $t \in \R$.

(iii) This statement can be proved following the reasoning of
(ii). Since in addition the condition \eqref{4.18} is satisfied
we find that the representation \eqref{4.25a} holds and $\xi(t) \ge 0$ for a.e. $t
\in \R$. Since $\go_1(t) = \go_2(t) -\xi(t) -i\eta(t)$, $t \in \R$, we
easily verify $\RE(\go_1(t)) \le \RE(\go_2(t))$ and $\IM(\go_1(t))
\le \IM(\go_2(t))$ for a.e. $t \in \R$.
\end{proof}
\br\la{IV.4}
{\em
\item[\;\ (i)]
By Lemma \ref{IV.1} the trace formula \eqref{5.39}
for selfadjoint extensions determines the function $\xi(\cdot)
\in L^1(\R,\tfrac{1}{1+t^2}dt)$  uniquely  up to a real constant.
In contrast to that, the trace formula \eqref{4.16a}  does not determine
the spectral shift function $\go(\cdot) \in
L^1(\R,\tfrac{1}{1+t^2}dt)$ up to a constant.

\item[\;\ (ii)]
Trace formula \eqref{4.16a} differs from that one
of \eqref{5.39} because the spectral shift function $\go(\cdot)$
is not real-valued. Using the results of \cite{AN90} one gets a trace
formula of type \eqref{5.39} for the pair $\{\wt A,H\}$ if in
addition to the assumption \eqref{4.14} the condition
\be\la{4.38}
(\wt A^* + i)^{-1} - (\wt A - i)^{-1} + 2i(\wt A^* + i)^{-1}(\wt A - i)^{-1} \in \gotS^0_1(\gotH)
\ee
is satisfied. Here $\gotS^0_1(\gotH)$ stands for the ideal of all
compact operators $T$ satisfying
\bed
\sum^\infty_{k=1}s_k(T)\log^+\left(\frac{1}{s_k(T)}\right) < \infty,
\eed
where $s_k(T)$, $k \in \N$, are the singular numbers of $T$. Note
that $\gotS^0_1(\gotH)$ is a strict part of $\gotS_1(\gotH).$ In
this case there exists  {\emph{a real-valued}} function
$\vartheta(\cdot) \in L^1(\R,\tfrac{1}{1+t^2}dt)$ such that the
trace formula
\bed
\tr\left((\wt A - z)^{-1} - (H - z)^{-1}\right) =
-\frac{1}{\pi}\int_\R \frac{\vartheta(t)}{(t-z)^2} dt, \qquad z
\in \C_+,
\eed
holds. If $\Pi = \{\cH,\gG_0,\gG_1\}$ is a boundary triplet for
$A^*$, which is regular for $\{\wt A,H\}$. If $\wt A = A_B$ and $H = A_C$,
$B,C \in \cH]$, then condition \eqref{4.38} is equivalent to $B_I = \IM(B)
\in \gotS^0_1(\cH)$. Assumption \eqref{4.14} implies $B_R - C \in
\gotS_1(\cH)$.

\item[\;(iii)] For further results on trace formulas for
  non-selfadjoint operators we refer to papers of A.
Rybkin \cite{Ryb84,Ryb84b, Ryb87,Ryb89,Ryb94}.
}
\er
Let us extend the results to maximal accumulative extensions.
\bt\la{IV.7}
Let $A$ be as above and let $\wt A_1, \wt A_2 \in \Ext_A$ be
maximal accumulative extensions of $A$ such that
\be\la{4.41}
(\wt A_2 - \zeta)^{-1} - (\wt A_1 - \zeta)^{-1} \in
\gotS_1(\gotH), \qquad \zeta \in \rho(\wt A_2) \cap \rho(\wt A_1).
\ee
and $\rho(\wt A_1) \cap \C_- \not= \emptyset$.
Then the following assertions are valid:

\item[\;\;\rm (i)]
There exists a boundary triplet $\Pi = \{\cH,
\gG_0, \gG_1\}$ for $A^*$, which can be chosen regular for $\{\wt
A_2,\wt A_2\}$, such that $\{\wt A_2,\wt A_1\} \in
\gotD^\Pi$.

\item[\;\;\rm (ii)]
If $\{\wt A_2,\wt A_1\} \in \gotD^\Pi$, then
there exist a complex-valued function $\go(\cdot) \in
L^1(\R,\tfrac{1}{1+t^2}dt)$ and a  constant $c \in \C$ such that
the representation
\be\la{4.42}
\gD^\Pi_{\wt A_2/\wt A_1}(z) = c\;
\exp\left\{\frac{1}{\pi}\int_R\left(\frac{1}{t-z} - \frac{t}{1 +
t^2}\right)\go(t)dt\right\}, \quad z \in \C_+,
\ee
holds.

\item[\;\;\rm (iii)]
The trace formula
\be\la{4.43}
\tr\left((\wt A_2 - z)^{-1} - (\wt A_1 - z)^{-1}\right) = -
\frac{1}{\pi}\int_\R \frac{\go(t)}{(t-z)^2} dt, \quad z \in \C_+.
\ee
holds where $\go(\cdot)$ is given by (ii).
\et
\begin{proof}
(i)  Since $\rho(\wt A_1) \cap \rho(\wt A_2) \supset \C_+$
there exists a boundary triplet $\Pi$ by Corollary \ref{IV.40}(i)
which is regular for the pair $\{\wt A_2,\wt A_1\}$.
Hence there exist accumulative operators
$B_j\in [\cH]$  such that $\wt A_j = A_{B_j}$, $j =1,2$. By
Proposition \ref{prop2.9}(ii) condition \eqref{4.41} is equivalent
to  $B_2 - B_1 \in \gotS_1(\cH)$ which yields the inclusion $\{\wt
A_2,\wt A_1\} \in \gotD^\Pi$.

(ii)   First, let us assume that $\Pi = \{\cH, \gG_0, \gG_1\}$ is
regular for $\{\wt A_2,\wt A_1\}$.  Clearly, we have $B_{2,R} -
B_{1,R} \in \gotS_1(\cH)$ and $B_{2,I} - B_{1,I} \in
\gotS_1(\cH)$. We set
\be\label{5.53}
B_3 = B_{2,R} + iB_{1,I}
\ee
and $\wt A_3 := A_{B_3},\  \dom(A_3) = \ker(\Gamma_1 -
B_3\Gamma_0)$. Since $B_3$ is accumulative, $\wt A_3$ is
$m$-accumulative (cf. Proposition \ref{prop2.1}). Since $B_2 - B_3
= i(B_{2,I} - B_{1,I}) \in \gotS_1(\cH)$,   $\{\wt A_2,\wt A_3\}
\in \gotD^\Pi$. Furthermore,  $B_3 - B_1 = B_{2,R} - B_{1,R} \in
\gotS_1(\cH)$, hence $\{\wt A_3,\wt A_1\} \in \gotD^\Pi$.
Therefore the perturbation determinant $\gD^\Pi_{\wt A_2/\wt A_3}(\cdot)$ is well
defined,
\bead
\lefteqn{
\gD^\Pi_{\wt A_2/\wt A_3}(z) = \det(I + (B_2 - B_3)(B_3 - M(z))^{-1})   }\\
& & = \det(I + i(B_{2,I} - B_{1,I})(B_3 - M(z))^{-1}), \qquad z
\in \C_+. \eead
Since $B_{2,I} - B_{1,I} \le -B_{1,I} = -B_{3,I}$, we obtain from
Lemma \ref{IV.3a}(ii) that there exist a  complex number
$\varkappa \in \T$ and {\emph {a real-valued}} function
$\eta(\cdot) \in L^2(\R,\tfrac{1}{1+t^2}dt)$ such that the
representation
\be\la{4.44}
\gD^\Pi_{\wt A_2/\wt A_3}(z) = \varkappa
\exp\left\{\frac{i}{\pi}\int_R\left(\frac{1}{t-z} -
\frac{t}{1+t^2}\right)\eta(t)dt\right\}, \qquad z \in \C_+,
\ee
holds.  Further, it follows from \eqref{5.53} that
\bead
\lefteqn{
\gD^\Pi_{\wt A_3/\wt A_1}(z) =
\det(I + (B_3 - B_1)(B_1 - M(z))^{-1})  }\\
& & = \det(I + (B_{2,R} - B_{1,R})(B_1 - M(z))^{-1}), \quad z \in
\C_+. \eead
By Lemma \ref{IV.2}(ii), there exist  a constant $c_1 > 0$ and a
{\emph {real-valued function}} $\xi(\cdot) \in
L^2(\R,\frac{1}{1+t^2}dt)$ such that the representation
\be\la{4.45}
\gD^\Pi_{\wt A_3/\wt A_1}(z) = c_1
\exp\left\{\frac{1}{\pi}\int_R\left(\frac{1}{t-z} -
\frac{t}{1+t^2}\right)\xi(t)dt\right\}, \quad z \in \C_+,
\ee
holds. Combining \eqref{4.44} with \eqref{4.45}, applying the
identity $\gD^\Pi_{\wt A_2/\wt A_3}(z)\gD^\Pi_{\wt A_3/\wt A_1}(z)
= \gD^\Pi_{\wt A_2/\wt A_1}(z)$, $z \in \C_+$,  and setting $c_0
:= c_1\varkappa$ and $\go(t) := \xi(t) + i\eta(t)$, $t \in \R$,
we  arrive at  the representation \eqref{4.42} with $c_0$ in place
of $c$. Finally, if $\Pi$ is not regular we apply Proposition \ref{III.5}
to get \eqref{4.42}.

(iii) The trace formula \eqref{4.43} follows immediately from \eqref{4.42} and
Proposition \ref{III.8}(v).
\end{proof}
\bc
Let $\wt A_2,\wt A_1 \in \Ext_A$ be accumulative extensions of $A$
such that the pair $\{\wt A_2,\wt A_1\}$ is resolvent comparable, that
is condition \eqref{4.41} is satisfied.

\item[\;\;\rm (i)]
If \eqref{4.18} is satisfied, then there is a complex-valued spectral shift function
$\go(\cdot)$ of the pair $\{\wt A_2,\wt A_1\}$ such that $\RE(\go(t)) \ge 0$
for a.e. $t\in \R$.

\item[\;\;\rm (ii)]
If \eqref{4.19} is satisfied, then there is a complex-valued spectral shift function
$\go(\cdot)$  of the pair $\{\wt A_2,\wt A_1\}$ such that $\IM(\go(t)) \ge 0$
for a.e. $t\in \R$.

\item[\;\;\rm (iii)]
If \eqref{4.18} and \eqref{4.19} are satisfied, then there is a complex-valued spectral shift function
$\go(\cdot)$ of the pair $\{\wt A_2,\wt A_1\}$ such that $\RE(\go(t))
\ge 0$ and $\IM(\go(t)) \ge 0$ for a.e. $t\in \R$.
\ec
\begin{proof}
By Corollary \ref{IV.40}(ii) there is a boundary triplet $\Pi$ for $A^*$ which is regular
for $\{\wt A_2,\wt A_1\}$ such that $\gl \in \rho(A_0)$ if
where $\gl \in \rho(\wt A_2) \cap \rho(\wt A_1) \cap \R$.
Hence there are bounded accumulative operators $B_2$ and $B_1$ such
that $\wt A_j = A_{B_j}$, $j = 1,2$. Let us introduce the boundary
triplet \eqref{4.10a}. Setting $\wt B_j = -(B_j - M(\gl_0))^{-1}$, $j
=1,2$, we have $\wt A_j = A_{\wt B_j}$, $j =1,2$. In addition we
introduce $\wt B_3$ defined by \eqref{4.330} and the maximal
accumulative extension $\wt A_3 := A_{\wt B_3}$. Now we follow the proof
of Theorem \ref{IV.3}.

(i) In this case we get $\wt B_{1,R} \le \wt B_{2,R}$. As above, this
yields that $\xi(\cdot) \in L^1(\R,\tfrac{1}{1+t^2}dt)$ can be chosen
non-negative in the representation \eqref{4.25a}. Using
$\gD^{\wt \Pi}_{\wt A_2/\wt A_1}(z) =
\gD^{\wt \Pi}_{\wt A_2/\wt A_3}(z)\gD^{\wt \Pi}_{\wt A_3/\wt A_1}(z)$,
$z \in \C_+$,  taking into account
the representations \eqref{4.25a} and \eqref{4.370} and setting $\go(t) := \xi(t) +
i\eta(t)$, $\in \R$,  we get
\be\la{5.51}
\gD^{\wt \Pi}_{\wt A_2/\wt A_1}(z) =
\wt c \exp\{\frac{1}{\pi} \int_\R \left(\frac{1}{t-z} - \frac{1}{1+t^2}\right)\go(t)dt,
\quad z \in C_+,
\ee
$\wt c \in \C$ where $\RE(\go(t)) \ge 0$ for a.e. $t \in \R$. Hence
there is a spectral shift function satisfying $\RE(\go(t)) \ge 0$ for
a.e. $t \in \R$.

(ii) In this case we have $\wt B_{2,I} \le \wt B_{1,I}$. This
yields that $\eta(\cdot)$ in the representation \eqref{4.370} can be
chosen non-negative. This immediately yields that $\IM(\go(t)) \ge 0$
for a.e. $t \in R$ in the representation \eqref{5.51}.
Hence there is a spectral shift function
satisfying $\IM(\go(t)) \ge 0$ for a.e. $t \in \R$.

(iii) Finally, in this case $\xi(\cdot)) \ge 0$ and $\eta(\cdot)$
can be chosen non-negative in the representation  \eqref{4.25a}
and \eqref{4.370}, respectively.  Setting $\go(t) := \xi(t) + i\eta(t)$,
$t \in \R$, we verify \eqref{5.51}. Hence there is a spectral shift function
satisfying $\RE(\go(t)) \ge 0$ and $\IM(\go(t)) \ge 0$ for a.e. $t \in \R$.
\end{proof}

\subsection{Pairs of extensions with one $m$-accumulative
  operator}\la{sec.V.3}

Here  we consider trace formulas for pairs $\{\wt
A',\wt A\}$ of proper extensions of a closed symmetric operator
$A$ assuming that $\wt A'$ is  $m$-accumulative extension.
\bl\la{IV.3aB}
Let $\Pi = \{\cH,\gG_0,\gG_1\}$  be a boundary triplet for $A^*,$
$M(\cdot)$ the corresponding  Weyl function and let $B\in [\cH]$ be an accumulative
operator, i.e $B_I \le 0$.

\item[\;\;\rm (i)]
If $0 \le V_+ \le 2|B_I| = -2B_I$,  $V \in
\gotS_1(\cH)$, then  the holomorphic function $w_+(z):= \det(I + iV_+(B -
M(z))^{-1})$, $z \in \C_+$, is contractive.
In particular, there exist a non-negative Borel measure
$\mu_+(\cdot)$  satisfying $\int_\R \tfrac{1}{1+t^2}d\mu_+(t) <
\infty$ as well as  numbers $\varkappa_+ \in \T$ and  $\ga_+ \ge 0$  such
that the following representation
\be\la{4.180a}
w_+(z) = \varkappa_+  \cB_+(z) \exp\left\{\frac{i}{\pi}\int_\R
\left(\frac{1}{t-z} -
\frac{t}{1+t^2}\right)d\mu_+(t)\right\}e^{i\ga_+ z}
\ee
$z \in \C_+$, holds where $\cB_+(\cdot)$ is the  Blaschke product formed by the zeros
$\{z^+_k\}_{k\in\N}$  of $w_+(\cdot)$ in $\C_+$, cf. \eqref{7.17A}.

\item[\;\;\rm (ii)]
If $V \le 2|B_I| = -2B_I$ and $V \in
\gotS_1(\cH)$, then the perturbation determinant $w(\cdot):= \det(I + iV(B -
M(\cdot))^{-1})$ belongs to the Smirnov class $\cN^+(\C_+)$, see Appendix \ref{App.III}. In particular,
there exist  a non-negative Borel measure
$\mu_+(\cdot)\ge 0$ satisfying $\int_\R\tfrac{1}{1+t^2}d\mu_+(t) < \infty$,
a non-negative function $\eta(\cdot) \in L^1(\R,\tfrac{1}{1+t^2}dt)$
as well as numbers $\varkappa \in \T$,  $\ga \ge 0$  such that
\be\la{4.190A}
w(z) = \varkappa\; \cB_+(z)\exp \left\{\frac{i}{\pi}\int_\R
\left(\frac{1}{t-z} - \frac{t}{1+t^2}\right)d\mu(t)\right\}
e^{i\ga z}, \quad z \in \C_+,
\ee
where $\mu(\cdot) := \mu_+(\cdot) - \eta(\cdot)dt$ and $\cB_+(z)$ is the Blaschke product formed by the zeros of
$w(\cdot)$ lying in $\C_+$.
\el
\begin{proof}
(i) We introduce the holomorphic operator-valued function
\be\la{4.18b}
W_+(z) := I + i\sqrt{V_+}(B - M(z))^{-1}\sqrt{V_+}, \quad z \in \C_+.
\ee
Since $B_I \le 0$, $\IM(M(z)) > 0$  and $0 \in \rho(\IM(M(z)))$
for $z \in \C_+$, the operator $(B - M(z))^{-1}$ is well-defined
and bounded for $z \in \C_+$. Further, we have
\bea\la{5.60}
\lefteqn{ I-W_+(z)^*W_+(z)} \nonumber \\
& &
= i\sqrt{V_+}\bigl((B^* - M(z)^*)^{-1}-(B-M(z))^{-1}\bigr)\sqrt{V_+} \nonumber\\
& & - \sqrt{V_+}(B^* - M(z)^*)^{-1}\,V_+\,(B-M(z))^{-1}\sqrt{V_+}.
\eea
Noting that
\bea\la{5.61}
\lefteqn{
\bigl(B^*-M(z)^*\bigr)^{-1}-\bigl(B-M(z)\bigr)^{-1}}  \\
& & = -2i(B^* - M(z)^*)^{-1} \cdot (|B_I| + M_I(z))\cdot
(B-M(z))^{-1},\nonumber
\eea
we obtain from \eqref{5.60} that
\bead
\lefteqn{
I-W_+(z)^*W_+(z) }\\
     & &
= \sqrt{V_+}(B^* - M^*(z))^{-1} \cdot (2|B_I| - V_+ +
2\IM(M(z)))\cdot (B-M(z))^{-1}\sqrt{V_+}. \eead
Since $V_+ \le 2|B_I|$ and  $\IM(M(z))$ is positively definite, we
have  $I- W_+(z)^*W_+(z) \ge 0$ for $z \in \C_+$, i.e.
$W_+(\cdot)$ is contractive in $\C_+$. Hence
$w_+(\cdot) = \det W_+(\cdot)$ is contractive in  $\C_+$.
Now the representation \eqref{4.180a}  immediately follows from the
factorization  \eqref{4.9a}.

(ii) Let $V = V_+ - V_-$, $V_\pm \ge 0$. We set $B_- := B - iV_-$.
Since $(B_-)_I = B_I - V_- \le 0$, the operator $B_-$ is
accumulative too. According to \eqref{4.20} we get
\be\label{5.23}
\det(I + iV(B - M(z))^{-1}) =
\frac{\det(I + iV_+(B_- - M(z))^{-1})}{\det(I + iV_-(B_- - M(z))^{-1})},
\quad z \in \C_+.
\ee
The assumption $V \le -2B_I$ yields  $0 \le V_+ \le -2B_I + V_- \le
-2B_I + 2V_-  = -2(B_-)_I$. Applying statement (i) to the operators
$B_-$ and $V_+$ we obtain that $\det(I + iV_+(B_- - M(z))^{-1})$
is a contractive analytic function. Furthermore, from $0 \le V_- \le -B_I + V_- = -(B_-)_I$ and Lemma
\ref{IV.3a}(i) we get that $\det(I + iV_-(B_- - M(z))^{-1})$ is an outer function.
Applying Lemma \ref{D.I} we complete the proof.
\end{proof}
\bt\label{th5.2}
Let $\wt A', \wt A \in \Ext_A$ and let  $\wt A$ be an
$m$-accumulative extension with $\rho(\wt A) \cap \C_- \not=
\emptyset$. Assume in addition,  that  condition \eqref{1.103} is
satisfied for some  $\zeta \in \rho(\wt A') \cap \rho(\wt A).$
Then the following holds:

\item[\;\;\rm (i)] There exists a  boundary triplet $\Pi = \{\cH,
\gG_0, \gG_1\}$ for $A^*$, which is  regular  for $\{\wt A',\wt
A\}$  and  such that $\{\wt A',\wt A\} \in \gotD^\Pi$.

\item[\;\;\rm (ii)] If $\{\wt A',\wt A\} \in \gotD^\Pi$, then
there exist a non-negative Borel measure $\mu_+(\cdot)$
satisfying $\int \frac{1}{1+t^2}d\mu_+(t) < \infty$ and a complex-valued function
$\go(\cdot) \in L^1(\R,\tfrac{1}{1+t^2}dt)$
as well as constants $\ga_+ \ge 0$ and $c \in \C$
such that the representation
\be\la{4.45b}
\gD^\Pi_{\wt A'/\wt A}(z) = c \;\cB_+(z)
\exp\left\{\frac{1}{\pi}\int_\R \left(\frac{1}{t-z} -
\frac{t}{1+t^2}\right)d\nu(t)\right\} e^{i\ga_+ z},
\ee
$z\in\C_+$, holds where $d\nu(\cdot) := \go(\cdot)dt + id\mu_+(\cdot)$
and $\cB_+(\cdot)$ is the Blaschke product (cf. \eqref{7.17A})
formed by  the eigenvalues $z^+_k$ of $\wt A'$ in $\C_+$ and their
algebraic multiplicities $m_k^+$ satisfying condition \eqref{4.12}.

\item[\;\;\rm (iii)]
The following trace formula holds
\bea\label{4.48}
\lefteqn{
\tr\left((\wt A' -z )^{-1} - (\wt A - z)^{-1}\right) }\\
& &
= -2i\sum_k \frac{m^+_k\IM(z^+_k)}{(z- \overline{z}^+_k)(z - z^+_k)} -
\frac{1}{\pi}\int_\R \frac{1}{(t-z)^2}d\nu(t) - i\ga_+, \;\; z \in
\rho(\wt A') \cap \C_+.
\nonumber
\eea
\et
\begin{proof}
(i) Since $\wt A$ is $m$-accumulative,  $\C_+ \subset \rho(\wt
A)$.  Therefore and due to the assumption  $\rho(\wt A) \cap \C_-
\not= \emptyset$, the conditions of  Proposition \ref{II.9} are
satisfied, hence the extension $\wt A$ is almost solvable. Now the
existence of a regular boundary triplet $\Pi$ is implied by
Corollary \ref{IV.40}(i) and the assumption \eqref{1.103}.

(ii) Let $\Pi$ be regular for $\{\wt A',\wt A\}$. By definition,
there exist  bounded operators $B', B \in [\cH],$ such that $\wt
A' = A_{B'}$ and $\wt A = A_B$.  Since $\wt A$ is
$m$-accumulative, by Proposition \ref{prop2.1}(iii), the operator
$B$ is accumulative too, i.e.  $B_I = \IM(B) \le 0$.  By
Proposition \ref{prop2.9}(ii) condition  \eqref{1.103} is
equivalent to $B' - B \in \gotS_1(\cH)$.  Hence $B'_R - B_R \in
\gotS_1(\cH)$ and $V := B'_I - B_I \in \gotS_1(\cH)$.
We set
\be\label{5.58}
 C := B'_R + iB_I.
\ee
Since $B$ is accumulative, the operator $C$ is also accumulative,
$C_I = B_I \le 0$. Let $V = B'_I - B_I = V_+ - V_-$, $V_\pm \ge
0$, be the orthogonal decomposition of the operator $V=V^*$. We
set
\be\label{5.58A}
D := C - i(V_+ + V_-)
\ee
and note that $D$ is accumulative because so is $C$ and
$V_{\pm}\ge 0.$ Since
\be\label{5.58B}
B' - D = B'_R + iB'_I - B'_R - iB_I + i(V_+ + V_-) =  2iV_+ \in \gotS_1(\cH),
\ee
we get $\{\wt A',A_D\} \in \gotD^\Pi$. Notice that
  \bed
2V_+ \le -2B_I + 2V_+ + 2V_- = -2(B_I - V_+ - V_-) = -2 D_I.
  \eed
According to  \eqref{5.58B} one has $\gD^\Pi_{\wt A'/A_D}(\cdot) = \det(I + 2iV_+(D-M(\cdot))^{-1})$. By
Lemma \ref{IV.3aB}(i),  $\gD^\Pi_{\wt A'/A_D}(\cdot)$ is contractive in $\C_+$ and admits the representation
\be\la{4.46}
\gD^\Pi_{\wt A'/A_D}(z) = \varkappa_+\;\cB_+(z)
\exp\left\{\frac{i}{\pi}\int_\R\left(\frac{1}{t-z} - \frac{t}{1+t^2}\right)d\mu_+(t)\right\}
e^{i\ga_+z}
\ee
where $\varkappa_+ \in \T$, $\ga_+ \ge 0$,  $\mu_+$ is a
non-negative Borel measure on $\R$ satisfying
$\int_\R\tfrac{1}{1+t^2}d\mu_+(t) < \infty$ and  $\cB_+(\cdot )$ is the
Blaschke product, cf. \eqref{7.17A} formed by  zeros $z^+_k
\in \C_+$ of  $\gD^\Pi_{\wt A'/A_D}(\cdot)$. By Proposition
\ref{III.8}(iv),  (cf. formula \eqref{3.11A}), each zero $z^+_k$
of $\gD^\Pi_{\wt A'/A_D}(\cdot)$ of the multiplicity  $m^+_k$ is
just  the eigenvalue of $\wt A'$ lying in $\C_+$,  and $m^+_k$ is
its algebraic multiplicity.    In particular, this yields that
$\{z^+_k\}_{k\in\N} = \sigma_p(\wt A')\cap \C_+,$ and the
eigenvalues $\{z^+_k\}_{k\in\N}$
satisfy  condition \eqref{4.12}.

Furthermore, since $\{\wt A',\wt A\} \in \gotD^\Pi$ and $\{\wt
A',A_D\} \in \gotD^\Pi$, we have $\{A_D,\wt A\} \in \gotD^\Pi$.
Note that the extension $A_D$ is $m$-accumulative since $D$ is
$m$-accumulative. Thus, both  $A_D$ and $\wt A$ are
$m$-accumulative and, by Theorem \ref{IV.7}(ii), there exists a
\emph{complex-valued function} $\go(\cdot) \in
L^1(\R,\tfrac{1}{1+t^2}dt)$ and a  complex constant $c_D \in \C$
such that the following representation holds
\be\la{4.47} \gD^\Pi_{A_D/\wt A}(z) = c_D \;
\exp\left\{\frac{1}{\pi}\int_\R\left(\frac{1}{t-z} - \frac{t}{1+
t^2}\right)\go(t)dt\right\}, \quad z \in \C_+.
   \ee
Using the chain rule $\gD^\Pi_{\wt A'/\wt A}(z) = \gD^\Pi_{\wt
A'/A_D}(z)\gD^\Pi_{A_D/\wt A}(z)$, $z \in \rho(\wt A) \cap \C_+$
(cf. Proposition \ref{III.8}(ii)), and combining \eqref{4.46} with
\eqref{4.47} we arrive at  representation \eqref{4.45b} with $c
:= c_D\varkappa_+$ and $d\nu(t) = \go(t)dt + id\mu_+(t)$.

The case of a boundary triplet  $\Pi$ which is not regular for the
pair $\{\wt A',\wt A\}$  is reduced to the previous one by
applying Proposition \ref{III.5}.

(iii) Trace formula  \eqref{4.48} is implied now by combining
\eqref{4.45b} with  Proposition \ref{III.8}(v).
\end{proof}

Using the Riesz-Dunford functional calculus, cf. Appendix
\ref{App.V}, we extend trace formula \eqref{4.48} to the case of
analytic functions of the class $\cF(\wt A,\wt A')$.
\bc\la{V.14}
Let the assumptions of Theorem \ref{th5.2} be satisfied. Let
$\{z^+_k\}_{k\in\N} = \sigma_p(\wt A')\cap \C_+$ and let $m^+_k$ be the
algebraic multiplicity of $z^+_k,\ k\in \N$. If $\Phi \in \cF(\wt
A,\wt A')$, cf. Appendix \ref{App.V}, then $\Phi(\wt A') -
\Phi(\wt A) \in \gotS_1(\gotH)$ and
\begin{align}\la{5.66}
\tr(\Phi(\wt A') - &\Phi(\wt A)) = \\
&\sum_k m^+_k(\Phi(z^+_k) -
\Phi(\overline{z}^{\;+}_k)) + \frac{1}{\pi}\int_\R \Phi'(t)d\nu(t)
+ i\ga_+\res_\infty(\Phi),
\nonumber
\end{align}
where $z^+_k$ are the eigenvalues of $\wt A'$ in $\C_+$ and $m_k^+$ their algebraic
multiplicities.
\ec
\begin{proof}
From Lemma \ref{A.Va} it follows that $\Phi(\wt A') -
\Phi(\wt A) \in \gotS_1(\gotH)$. Multiplying \eqref{4.48} with
$\Phi(z)$ and integrating both sides with respect to $dz$ we obtain
immediately \eqref{5.66} using formulas \eqref{F0}, \eqref{F3},\eqref{4.14a}
and \eqref{F4} of Appendix \ref{App.V}.
\end{proof}

\subsection{Pairs of an extension and its adjoint}\la{sec.V.4}

Next we consider perturbation  determinants and trace formulas
for pairs $\{\wt A,\wt A^*\}$ of proper
extensions $\wt A,\wt A^* \in \Ext_A$ assuming that $\rho(\wt A)\cap \rho({\wt A}^*)\not =
\emptyset.$

\bt\label{th5.4}
Let $\wt A \in \Ext_A$ and $\rho(\wt A)\cap \rho({\wt A}^*)\not =
\emptyset.$ Assume  also that
\be\la{5.65}
(\wt A - \zeta)^{-1} - ({\wt A}^* - \zeta)^{-1}\in \gotS_1(\gotH), \quad \zeta
\in \rho(\wt A)\cap \rho({\wt A}^*).
\ee
Then the following holds:

\item[\;\;\rm (i)] There exists a boundary triplet  $\Pi = \{\cH,
\gG_0, \gG_1\}$ for $A^*$, which can be chosen regular for $\{\wt A,\wt
A^*\}$,  such that $\{\wt A,\wt A^*\} \in \gotD^\Pi$.

\item[\;\;\rm (ii)] If $\{\wt A,\wt A^*\} \in \gotD^\Pi$ and the
triplet $\Pi$ is regular for the pair $\{\wt A,\wt A^*\}$, then
\be\la{5.69}
\gD^\Pi_{\wt A/\wt A^*}(z) = \det(W^\Pi_{\wt A}(z)), \qquad z \in
\rho(\wt A^*) \cap \C_\pm,
\ee
where $W^\Pi_{\wt A}(\cdot)$ is the characteristic operator-valued
function of $\wt A$ defined by \eqref{2.33}, cf. Proposition \ref{II.15}.

\item[\;\;\rm (iii)] If $\{\wt A,\wt A^*\} \in \gotD^\Pi$, then
there exist a real-valued measure $\mu$ on $\R$ satisfying
$\int_\R\tfrac{1}{1+t^2}|d\mu|(t)| < \infty$ and constants $\ga \in \R$ and  $c \in
\C$ such that the representation
\be\la{4.52} \gD^\Pi_{\wt A/\wt A^*}(z) = c\;
\frac{\cB_+(z)}{\cB_-(z)}\exp \left\{\frac{i}{\pi}\int_\R
\left(\frac{1}{t-z} - \frac{t}{1+t^2}\right)d\mu(t)\right\}
e^{i\ga z}, \ee
holds  for $z \in \rho(\wt A^*) \cap \C_+$ where $\cB_+(\cdot)$
and $\cB_-(\cdot)$ are Blaschke products (cf. \eqref{7.17A})
formed by the zeros $\{z^+_k\}_{k\in\N}$ and $\{\overline{z}{\,^-_l}\}_{l\in\N}$, where
$\{z^+_k\}_{k\in\N}$ and $\{z{\,^-_l}\}_{l\in\N}$ are eigenvalues of $\wt A$ in $\C_+$ and $\C_-$, respectively,
and $m_k^+$ and $m_l^-$ their algebraic multiplicities, respectively.
Both sequences $\{z^+_k\}_{k\in\N}$ and
$\{\overline{z}{\,^-_l}\}_{k\in\N}$ satisfy condition
\eqref{4.12}.

\item[\;\;\rm (iv)] The following  trace formula holds
   \bea\la{4.53}
\lefteqn{
\tr\left((\wt A^* - z)^{-1} - (\wt A - z)^{-1}\right) }\\
& = &
2i\sum_n \frac{m_n^+\IM(z_n)}{(z - z_n)(z - \overline{z}_n)}
+ \frac{i}{\pi}\int_R\frac{1}{(t-z)^2}d\mu(t) + i\ga,
\nonumber
   \eea
$z \in \rho(\wt A^*) \cap \rho(\wt A) \cap \C_+$, where $z_n$ and
$m_n$ denote the eigenvalues of $\wt A$ in $\C \setminus \R$ and
their algebraic multiplicities, respectively.
\et
\begin{proof}
(i) Let $z_0\in \rho(\wt A)\cap\rho(\wt A^*)$. Then  $\overline
z_0\in \rho(\wt A)\cap\rho(\wt A^*)$ and, by Proposition
\ref{II.9}, the extension $\wt A$ is  almost solvable. By
Corollary \ref{IV.40}(i), there  exists a boundary triplet $\Pi$
regular for $\{\wt A,\wt A^*\}$ and  satisfying  $\{\wt A,\wt
A^*\} \in \gotD^\Pi$.

(ii) Assume that $\{\wt A,\wt A^*\} \in \gotD^\Pi$ and $\Pi$ is a
regular boundary triplet for $\{\wt A,\wt A^*\}$. By definition,
there exists  a bounded operator $B\in [\cH]$ such that $\wt A =
A_B$ and $\wt A^* = A_{B^*}$. Let  $B_I := J|B_I|$ be the polar
decomposition of $B_I= (B - B^*)/2i = B_I^*$ where $J=J^* =
J^{-1}$. From $\{\wt A,\wt A^*\} \in \gotD^\Pi$ and Proposition
\ref{prop2.9}(ii) we get $B - B^* = 2iB_I \in \gotS_1(\cH)$, $B_I
:= \IM(B)$. Taking into account the property \eqref{2.26} we get
\bead
\gD^\Pi_{\wt A/\wt A^*}(z) = \det\bigl(I + (B - B^*)(B^* - M(z))^{-1}\bigr) \quad \quad \quad \\
= \det(I + 2i\sqrt{|B_I|}(B^* -M(z))^{-1}\sqrt{|B_I|}J), \qquad z
\in \rho(\wt A^*) \cap \C_\pm. \eead
Applying  Proposition \ref{II.15} we arrive at \eqref{5.69}.

(iii) Let a boundary triplet $\Pi$ be regular for $\{\wt A,\wt
A^*\}$.  Consider the spectral  decomposition $B_I = B_I^+ -
B_I^-$ of $B_I$, where  $B_I^{\pm}$  are orthogonal, $B_I^{\pm}\ge
0$, and $B_I^{\pm} \in \gotS_1(\cH)$. Alongside $B$ consider the
dissipative operator $B_1= B_R + i |B_I| = B_R + i B^+_I + i
B^-_I$ and the corresponding extension $\wt {A'} := A_{B_1^*}$. By
Proposition \ref{prop2.1}(iii),  $\wt {A'}$ is maximal
accumulative. We note that $\{\wt A,\wt A'\} \in \gotD^\Pi$. To
the perturbation determinant $\gD^\Pi_{\wt A/\wt A'}(\cdot)$ we
can apply Lemma \ref{IV.3aB}(i) with $B_1^*$  in place of $B$ and
$V_+ := 2B_I^+\le -2\IM(B^*_1)$. This  yields the representation
\be\la{5.75a}
\gD^\Pi_{\wt A/\wt A'}(z) = \varkappa_+\;
\cB_+(z)
\exp\left\{\frac{i}{\pi}\int_\R \left(\frac{1}{t-z} -
\frac{t}{1+t^2}\right)d\mu_+(t)\right\} e^{i\ga_+ z},
\ee
$z \in \C_+$. Here $\varkappa_+ \in \T$, $\ga_+ \ge 0$,  $\nu_+$
is a \emph{non-negative Borel measure} satisfying $\int_R
({1+t^2})^{-1}d\nu_+(t)< \infty$, and  $\cB_+(\cdot)$ is  the
Blaschke product, cf. \eqref{7.17A}, with zeros
$\{z^+_k\}_{k\in\N}.$ By Proposition  \ref{III.8}(iv),
$\{z^+_k\}_{k\in\N} = \sigma_p(\wt A)\cap\C_+$ and the order
$m^+_k$ of zero $z^+_k$ equals to the algebraic multiplicity of
$z^+_k$ as the eigenvalue of $\wt A.$

Next, consider the perturbation determinant $\gD^\Pi_{\wt A^*/\wt
A'}(\cdot)$.  Again Lemma \ref{IV.3aB}(i) yields  the
representation
\be\la{5.76}
\gD^\Pi_{\wt A^*/\wt A'}(z) = \varkappa_-\; \cB_-(z)
\exp\left\{\frac{i}{\pi}\int_\R \left(\frac{1}{t-z} -
\frac{t}{1+t^2}\right)d\mu_-(t)\right\} e^{i\ga_- z}, \quad z \in
\C_+.
\ee
Here  $\cB_-(\cdot)$ is the  Blaschke product,  cf. \eqref{7.17A},
with zeros  $\{\zeta^+_k\}_{k\in\N}$ being the eigenvalues of $\wt
A^*$ in $\C_+$. Moreover,  the order $n^+_k$ of zero $\zeta^+_k$
is equal to the algebraic multiplicity of $\zeta^+_k$ as the
eigenvalue of $\wt A^*$.
 Note however, that $\zeta^+_k \in\sigma_p(\widetilde{A}^*)$ if and
only if $z^-_k := \overline{\zeta^+_k} \in\sigma_p(\widetilde A)$
and  the corresponding algebraic multiplicities $n^+_k$ and
 $m^-_k$ coincide, $n^+_k = m^-_k$.
Thus,  the Blaschke product $\cB_-(\cdot)$ is defined  by the
complex conjugated eigenvalues $z^-_k$ of $\wt A$ lying in $\C_-$
and taken with orders $m^-_k$ equal to their algebraic
multiplicities.

Combining \eqref{5.75a} and   \eqref{5.76} with the chain rule
\be\la{5.70}
\gD^\Pi_{\wt A/\wt A^*}(z)  =
\frac{\gD^\Pi_{\wt A/\wt A'}(z)}{\gD^\Pi_{\wt A^*/\wt A'}(z)}, \qquad z\in \rho(\wt A^*) \cap
\rho(\wt A) \cap \C_+,
\ee
and setting  $c := \tfrac{c_+}{c_-}$, $\mu := \mu_+ - \mu_-$ and
$\ga := \ga_+ - \ga_-$, we arrive at  \eqref{4.52}.

To prove \eqref{4.52} for any (not necessarily regular) boundary
triplet $\Pi$ satisfying $\{\wt A^*,\wt A\} \in \gotD^\Pi$ it
remains  to apply Proposition \ref{III.5}.

(iv)
Taking into account \eqref{4.48} we find
\bea\la{4.54}
\lefteqn{
\tr\left((\wt A - z)^{-1} - (\wt A' - z)^{-1}\right) }\nonumber \\
& & = -2i\sum_k \frac{m_k^+\IM(z^+_k)}{(z - z_k)(z -
\overline{z}{^{\,+}_k})} -
\frac{i}{\pi}\int_R\frac{1}{(t-z)^2}d\mu_+(t) - i\ga_+
   \eea
for $z \in \rho(\wt A) \cap \C_+$ and
    \bea\la{4.55}
\lefteqn{
\tr\left((\wt A^* - z)^{-1} - (\wt A' - z)^{-1}\right) }\nonumber  \\
& &
 = -2i\sum_l \frac{m_l^-\IM(\overline{z}^-_l)}{(z - z^-_l)(z -
\overline{z}{^{\,-}_l})}  - \frac{i}{\pi}\int_R\frac{1}{(t-z)^2}d\mu_-(t) - i\ga_-
\nonumber\\
& & = 2i\sum_l \frac{m_l^-\IM(z^-_l)}{(z - z^-_l)(z -
\overline{z}{^{\,-}_l})}  -
\frac{i}{\pi}\int_R\frac{1}{(t-z)^2}d\mu_-(t) - i\ga_-
     \eea
for $z\in \rho(\wt A^*) \cap \C_+$. Subtracting \eqref{4.54} from
\eqref{4.55}  we easily obtain \eqref{4.53}.
\end{proof}
\bc\la{V.18}
Let the assumptions of Theorem \ref{th5.4} be satisfied. If $\Phi \in
\cF(\wt A,\wt A^*)$, then $\Phi(\wt A) - \Phi(\wt A^*) \in \gotS_1(\gotH)$ and
\bea\label{5.75}
\lefteqn{
\tr(\Phi(\wt A) - \Phi(\wt A^*)) = }\\
& &
\sum_n m_n(\Phi(z_n) -
\Phi(\overline{z}_n)) + \frac{i}{\pi}\int_\R \Phi'(t)d\mu(t) +
i\ga\;\res_\infty(\Phi)
\nonumber
\eea
where $z_n$ are the eigenvalues of $\wt A$ in $\C \setminus \R$ and $m_n$
their algebraic multiplicities.
\ec
\begin{proof}
The inclusion $\Phi(\wt A) - \Phi(\wt A^*) \in \gotS_1(\gotH)$  immediately follows
from \eqref{5.65} and Lemma \ref{A.Va}.
Further, we multiply identity \eqref{4.53}  by $\Phi(\cdot)$ and
then integrate the result along a simple closed curve $\gG$
containing  the spectra $\sigma(\wt A)\cup \sigma({\wt A}^*)$.
Applying  formulas   \eqref{4.14a}, \eqref{F3} and \eqref{F4}  we
arrive at formula \eqref{5.75}.
\end{proof}
\begin{remark}
{\em
Corollary \ref{V.18}  generalize the known result of V. Adamyan
and B. Pavlov \cite{AdamPav79} and coincide with that in the case
of an $m$-dissipative operator $\wt A$ with $\rho(\wt A)\cap \C_+
\not = \emptyset$. The result has obtained  in
\cite{AdamPav79} by applying  a functional model of
$m$-dissipative operators \cite{FN70}.
}
\end{remark}

Next we complete and simplify Theorem \ref{th5.4}  assuming,  in
addition, that the resolvent of an extension is compact.
\bt\la{V.16}
Let the assumptions of Theorem \ref{th5.4} be satisfied. If, in
addition, $(\wt A -z)^{-1}\in \mathfrak S_{\infty}(\gotH)$, then
the following holds:

\item[\;\;\rm (i)] If $\{\wt A,\wt A^*\} \in \gotD^\Pi$, then the
perturbation determinant $\gD^\Pi_{\wt A/\wt A^*}(\cdot)$ 
is holomorphic in a neighborhood of the real line  $\R$ and
\be\la{5.79}
|\gD^\Pi_{\wt A/\wt A^*}(x)| = 1 \qquad \text {for}\qquad  x\in \R.
\ee

\item[\;\;\rm (ii)] If $\{\wt A,\wt A^*\} \in \gotD^\Pi$, then the
representation \eqref{4.52} is simplified to
   \be\la{4.52A}
\gD^\Pi_{\wt A/\wt A^*}(z) = c\; \frac{\cB_+(z)}{\cB_-(z)}
 e^{i\ga z}, \qquad \alpha \in \R, \qquad z \in \rho(\wt A) \cap \C_+.
  \ee

  \item[\;\;\rm (iii)]
The following trace formula holds
\be\la{4.53B}
\tr\left((\wt A^* - z)^{-1} - (\wt A - z)^{-1}\right) =
2i\sum_n \frac{m_n\IM(z_n)}{(z - z_n)(z - \overline{z}_n)}  + i\ga,
\ee
for $z \in \rho(\wt A^*) \cap \rho(\wt A).$ 
In particular, if $a=\overline a\in \rho(\wt A)$, then
\be \la{4.53A}
\ga/2 = \tr(\IM(\wt A^* - a)^{-1}) - \sum_n \IM  \left(\frac{m_n}{a - z_n}\right)
\ee
where $z_n$ and $m_n$ are the eigenvalues of $\wt A$ in $\C
\setminus \R$ and their multiplicities, respectively.
\et
\begin{proof}
(i) Let $\Pi$ be a boundary triplet for $A^*$ regular for $\{\wt
A,\wt A^*\}$ such that $\{\wt A,\wt A^*\} \in \gotD^\Pi$, cf.
Theorem \ref{th5.4}(i). Since $(\wt A^* - z)^{-1}$ is compact for
$z\in \rho(\wt A^*)$ the perturbation determinant $\gD^\Pi_{\wt
A/\wt A^*}(\cdot)$ is meromorphic in $\C$.

Since $\Pi$ is  regular for $\{\wt A,\wt A^*\}$, one has  $\wt A =
A_B = A^* \!\upharpoonright\ker(\gG_1 - B\gG_0)$  and $\wt A^* =
A_{B^*}$ where $B\in [\cH]$. Therefore the real part $\wt A_R$ of
$\wt A$ is well defined, $\wt A_R := A_{B_R}.$ Since $B_I\in
\mathfrak S_1,$  the perturbation determinants $\gD^\Pi_{\wt
A_R/\wt A}(\cdot)$ and $\gD^\Pi_{\wt A_R/\wt A^*}(\cdot)$  are
well defined and
\bed
\gD^\Pi_{\wt A_R/\wt A}(z)  = \det\bigl(I + (B_R - B)(B -
M(z))^{-1}\bigr) = \det(I - i B_I (B - M(z))^{-1}),
\eed
$z \in \rho(\wt A) \cap \rho(A_0)$, and
\bed \gD^\Pi_{\wt A_R/\wt A^*}(z)   = \det\bigl(I + (B_R -
B^*)(B^* - M(z))^{-1}\bigr) = \det(I + i B_I (B^* - M(z))^{-1}),
\eed
$z \in \rho(\wt A^*) \cap \rho(A_0)$. Moreover, by Proposition
\ref{III.7}, the determinants   $\gD^\Pi_{\wt A_R/\wt A}(\cdot)$
and $\gD^\Pi_{\wt A_R/\wt A^*}(\cdot)$ admit holomorphic
continuations from $\rho(\wt A) \cap \rho(\wt {A^*})\cap
\rho(A_0)$  to $\rho(\wt A)$ and $\rho(\wt{A^*})$, respectively.
Since the resolvents of $\wt A$ and $\wt A^*$ are compact
the determinants   $\gD^\Pi_{\wt A_R/\wt
A}(\cdot)$ and $\gD^\Pi_{\wt A_R/\wt A^*}(\cdot)$ are meromorphic. According to
\eqref{4.107} we get
\bed \gD^\Pi_{\wt A_R/\wt A^*}(z) = \overline{\gD^\Pi_{\wt A_R/\wt
A}(\overline{z})}, \qquad z \in \rho(\wt{A^*}). \eed
In particular, we have
\bed
\gD_{\wt A_R/\wt A^*}(x) = \overline{\gD_{\wt A_R/\wt A}(x)},
\qquad  x \in \rho(\wt A)\cap \rho(\wt {A^*})\cap \R = \rho(\wt
A^*)\cap \R = \rho(\wt A)\cap \R.
\eed
Using this identity and applying the chain rule we get
\be\label{5.82}
\gD_{\wt A/\wt A^*}(x) = \frac{\overline{\gD_{\wt A_R/\wt
A^*}(x)}}{\gD_{\wt A_R/\wt A}(x)}, \qquad x\in \rho(\wt A) \cap \rho(\wt A^*) \cap
\rho(\wt A_R) \cap \R.
\ee
It follows that $|\gD_{\wt A/\wt A^*}(x) | = 1$ for  $x\in
\rho(\wt A) \cap \rho(\wt A_R) \cap \R.$ Since $(\wt A_R -
z)^{-1}\in \gotS_\infty(\gotH)$, $z\in \rho(\wt A)$, the operator
$\wt A_R$ has also discrete spectrum, $\gs(\wt A_R) = \gs_d(\wt
A_R)$. Thus, $|\gD_{\wt A/\wt A^*}(x)| = 1$ for $x$ outside a
discrete set $(\gs(\wt A_R)\cup \gs(\wt A)) \cap \R$. Hence  any
possible real pole $x_0$ of the meromorphic function $\gD^\Pi_{\wt
A/\wt A*}(\cdot)$  is removable.  Thus, $|\gD_{\wt A/\wt A^*}(x)|
= 1$ for any $x\in \R$  which shows that $\gD_{\wt A/\wt
A^*}(\cdot)$ is holomorphic in a neighborhood of $\R$.

(ii) Clearly, the extension $\wt A' = A_{B^*_1}$, $B_1 := B_R +
i|B_I|$, is $m$-accumulative. Moreover, since $B - B^*_1 = 2i
|B_I|\in \gotS_1(\cH)$, it follows from Proposition
\ref{prop2.9}(ii) that
the  resolvent of $\wt A'$ is compact, i.e.,  the  spectrum of $\wt
A'$ is discrete. Hence, the perturbation determinant $F_+(\cdot)
:= \gD^\Pi_{\wt A/\wt A'}(\cdot)$ is holomorphic in $\C_+$ and
meromorphic in $\C$. In particular, $F_+(\cdot)$ admits a
holomorphic continuation through $\R \setminus \gs(\wt A') = \R
\setminus \gs_p(\wt A')$ where, of course, $\gs_p(\wt A') \cap \R$
is a discrete set. From \cite[Theorem II.6.3]{Gar81} we find that
the inner and outer factors $I_{F_+}(\cdot)$ and
$\cO_{F_+}(\cdot)$, respectively, of the contractive in $\C_+$
holomorphic function $F_+(\cdot) := \gD^\Pi_{\wt A/\wt
A'}(\cdot)$, cf. Appendix \ref{App.III}, admit also a holomorphic
continuation through $\R \setminus \gs(\wt A')$. Since the
Blaschke product $\cB_+(\cdot)$ admits a holomorphic continuation
through $\R \setminus \gs_p(\wt A')$ we get that the singular
factor  $S_{F_+}(\cdot)$  (cf. \eqref{4.10Inner})
has this property. By \cite[Theorem II.6.2]{Gar81} the singular
part $\mu^s_+$ of the measure $\mu_+$ is supported on $\R \cap
\gs(\wt A')$. Thus,  the singular continuous part $\mu^{sc}_+$ of
the measure $\mu_+$ is missing, i.e. $\mu^{sc}_+ \equiv 0$. Hence,
$\mu^s_+$ is atomic and supported on $\gs(\wt A')$, i.e.
\be\la{5.76a}
S_{F_+}(x) =
\exp\left\{\frac{i}{\pi}\sum_{t_k \in \gs(\wt A') \cap \R}\left(\frac{1}{t_k-x} -
    \frac{t_k}{1+t^2_k}\right)\mu^s_+(\{t_k\})
\right\},
\ee
$x \in \R \setminus \gs(\wt A')$. By a straightforward computations it follows from \eqref{4.9a}  that
\be\la{5.77}
\lim_{y\to+0}|\gD^{\Pi}_{\wt A/\wt A'}(x+iy)| = e^{-\mu'_+(x)}
\quad \text{for a.e.}\quad x \in \R.
\ee
By the same reasoning we get that the singular factor
$S_{F_-}(\cdot)$ of $F_-(\cdot) := \gD_{\wt A^*/\wt A'}(\cdot)$, 
(cf. \eqref{4.10Inner}) admits the representation
\be\la{5.78}
S_{F_-}(x) =
\exp\left\{\frac{i}{\pi}\sum_{t_k \in \gs(\wt A') \cap \R}\left(\frac{1}{t_k-x} -
    \frac{t_k}{1+t^2_k}\right)\mu^s_-(\{t_k\})
\right\},
\ee
$x \in \R \setminus \gs(\wt A')$, and $\mu^{sc}_- \equiv 0$. Moreover,  it follows from \eqref{5.76}
that
\be\la{5.79a} \lim_{y\to+0}|\gD^{\Pi}_{\wt A^*/\wt A'}(x+iy)| =
e^{-\mu'_-(x)} \qquad \text{for a.e.} \quad x \in \R.
    \ee
Combining  \eqref{5.82}  with  \eqref{5.77} and \eqref{5.79a} we
get
   \bed
|\gD^\Pi_{\wt A/\wt A'}(x)| = e^{-(\mu'_+(x) - \mu'_-(x))} \qquad
\text{for a.e.} \quad x \in \R.
  \eed
Since $|\gD^\Pi_{\wt A/\wt A'}(x)|=1$ for $x \in \R$ we get
$\mu'_+(x) = \mu'_-(x)$ for a.e. $x\in \R$, i.e. $\mu^{ac}_+ =
\mu^{ac}_-$. Hence $\cO_{F_+}(z) = \cO_{F_-}(z)$, $z \in \C_+$
which yields the representation
  \bed
\gD^\Pi_{\wt A/\wt A^*}(z) = \frac{F_+(z)}{F_-(z)} =
\frac{I_{F_+}(z)}{I_{F_-}(z)} =
\frac{\varkappa_+}{\varkappa_-}\frac{\cB_+(z)}{\cB_-(z)}\frac{S_{F_+}(z)}{S_{F_-}(z)}e^{i(\ga_+-\ga_-)},
\;\; z \in \rho(\wt A^*) \cap \C_+.
   \eed
Since the spectrum of $\wt A$ is discrete the eigenvalues $z^+_k$
of $\wt A$ in $\C_+$ cannot accumulate to the real axis. Hence the
limit $\cB_+(x) = \lim_{y\to+0}\cB_+(x+iy)$ exists for any $x \in
\R$ and the limit function $\cB_+(x)$ is continuous. Moreover,
$|\cB_+(x)| = 1$ for $x\in \R$. Similarly one shows that the limit
function $\cB_-(x)$ is defined everywhere, is continuous and
$|\cB_-(x)| = 1$ for $x \in \R$. Hence $\frac{\cB_+(x)}{\cB_-(x)}$
is continuous. Therefore the limit function $S(x) :=
\lim_{y\to+0}\frac{S_{F_+}(x+iy)}{S_{F_-}(x+iy)}$ exists everywhere and is
continuous. Combining  \eqref{5.76a} with  \eqref{5.78} we get the
representation
  \bed
S(x) = \exp\left\{\frac{i}{\pi}\sum_{t_k \in \gs(\wt A') \cap \R}\left(\frac{1}{t_k-x} -
    \frac{t_k}{1+t^2_k}\right)(\mu^s_+(\{t_k\}) - \mu^s_-(\{t_k\}))
\right\}
   \eed
for $x \in \R \setminus \gs(\wt A')$. It is easily  seen that the
function $S(\cdot)$ is continuous at $t_k$ if and only if
$\mu^s_+(\{t_k\}) = \mu^s_-(\{t_k\})$, $t_k \in \gs(\wt A') \cap
\R$ which yields $S(x) =1$ for $x \in \R$.   Thus, we arrive at
the representation \eqref{4.52A} with  $c := \tfrac{\varkappa_+}{\varkappa_-}$
and $\ga := \ga_+ - \ga_-$.

To prove \eqref{4.52A} for any (not necessarily regular) boundary
triplet $\Pi$ satisfying $\{\wt A^*,\wt A\} \in \gotD^\Pi$ it
remains  to apply Proposition \ref{III.5}.

(iii) Formula \eqref{4.53B} follows from \eqref{4.53} with the
measure $\mu = 0$. Formula \eqref{4.53A} follows immediately from
\eqref{4.53B}.
\end{proof}
\bc\la{V.19}
Let the assumptions of Theorem \ref{V.16} be satisfied. If
$\Phi \in \cF(\wt A,\wt A^*)$, then
$\Phi(\wt A^*) - \Phi(\wt A) \in \gotS_1(\gotH)$ and
\bed
\tr(\Phi(\wt A) - \Phi(\wt A^*)) = \sum_n m_n(\Phi(z_n) -
\Phi(\overline{z}_n))  + i\ga\;\res_\infty(\Phi)
\eed
where $\{z_n\}_{n\in\R}$ are the non-real
eigenvalues of $\wt A$ and  $m_n$   their algebraic multiplicities,
respectively.
\ec
\begin{proof}
Corollary \ref{V.19} follows from Corollary \ref{V.18} setting $\mu
\equiv 0$.
\end{proof}
\br\la{V.17}
{\em
\item[\;\;(i)]
If $\wt A$ is $m$-dissipative,
then the perturbation determinant $\gD_{\wt A/\wt A^*}(\cdot)$ is
holomorphic and contractive in $\C_+$ and due to \eqref{5.79} it
is an  inner function in $\C_+$.

In contrast to this fact, in the non-dissipative case the
perturbation determinant $\gD^\Pi_{\wt A/\wt A^*}(\cdot)$ admits
the representation $\gD^\Pi_{\wt A/\wt A^*}(z) =
\frac{F_+(z)}{F_-(z)}$, $z \in \C_+$, where the numerator and the
denominator  might really have outer factors despite of the
analyticity of both determinants  on the real line and the
necessary condition $|\gD^\Pi_{\wt A/\wt A^*}(x)| =1$ for $x \in
\R$ (cf. \eqref{5.79}).

\item[\;\;(ii)]  Notice that the non-dissipative  operator $\wt A$ might have
real eigenvalues even if it is completely non-selfadjoint. However
these eigenvalues do not appear in the representation \eqref{4.52A}
neither  in the trace formula \eqref{4.53B}.
This fact is not surprising since if  $\lambda_0 =
\overline{\lambda}_0 \in \sigma_p(\wt A)$ then $\lambda_0 \in
\sigma_p(\wt A^*)$ and $\dim\ker(\wt A - \lambda_0) =\dim\ker(\wt
A^* - \lambda_0)$ and these zeros cancel out in the representation
\eqref{4.52A}.
  Due to formula \eqref{5.69}  such eigenvalues do not appear
in the determinant of the characteristic function $W^\Pi_{\wt
A}(\cdot)$. In this connection  we mention the paper
\cite{VesNab86} where it is shown that even singular factors
cancel in a formula for the determinant of the  the characteristic
function.
}
\er

Theorem \ref{V.16} allows to indicate a condition which guarantees the
completeness of the root vector system, cf. Appendix \ref{App.IV}.
\bc\la{V.20}
Let the assumption of Theorem \ref{V.16} be satisfied and let us
assume in addition that $\wt A$ is a maximal dissipative operator. The
root vector system of $\wt A$ is complete if and only if $\ga = 0$.
\ec
\begin{proof}
Let $\wt A$ be a maximal dissipative extension of $A$ such the
$\rho(\wt A) \not= \emptyset$. Since $(\wt A - z)^{-1}$ is compact for some $z \in
\rho(\wt A)$ there is a real number $a \in
\rho(\wt A)$. Hence $a \in \rho(\wt A^*)$.  Let $R := (\wt A^* -
a)^{-1}$. A simple computation shows that $R$ is a bounded dissipative
operator. From \eqref{5.65} one gets that $(\wt A^* -a)^{-1} - (\wt A
-a)^{-1}$ is a trace class operator. Hence $R - R^*$ is trace class
operator which yields $R_I := \IM(R)$ is a trace class
operator. Since $R$ is dissipative we
obtain from  \eqref{4.53A} that
\bed
\tr(R_I) = \sum_n m_n\IM(\mu_n) + \frac{\ga}{2}, \quad \mu_n :=
\frac{1}{\overline{z}_n - a} ,
\eed
where $z_n$ and $m_n$ the eigenvalues of $\wt A$ in $\C_+$
and their multiplicities. Since $\wt A$ is maximal  dissipative it holds
$\ga \ge 0$. We note that $\mu_n$ are the eigenvalues
of $T$ and $m_n$ their multiplicities.
Applying Theorem V.2.1 of \cite{GK69} we prove
that the root vector system of $R$ is complete if and only if
$\ga = 0$. Using Lemma \ref{A.V} we get that root vector system of
$\wt A$ is complete if and only if $\ga = 0$.
\end{proof}

\subsection{Annihilation functions for dissipative
  extensions}\la{sec.V.5}

We are going to prove a Cayley-Hamiltonian-type theorem for maximal dissipative extensions
of a symmetric operator $A$ with finite deficiency indices.
{\vio{Let $A$ be a densely defined closed symmetric operator.
A point $z \in \C$ is called of regular type of $A$,
cf. \cite[Section VIII.100]{AG81}, if there is a constant $c > 0$
such that $\|(A - z)f\|^2 \ge c\|f\|^2$, $f \in \dom(A)$.
By $\wh\rho(A)$ we denote the set of all points of regular type of $A$.}}
{\vio{
\begin{proposition}\la{V.21}
Let $A$ be a \emph{simple}  closed symmetric operator in $\gotH$ with finite
deficiency indices $n:= n_+(A)=n_-(A)<\infty$ and
let $\Pi=\{\cH,\Gamma_0,\Gamma_1\}$ be a boundary triplet
for $A^*$. Further, let $\wt A \in \Ext_A$ be a maximal dissipative
extension of $A$ such that $\wt A = A_B$ with $B \in [\cH]$.
Assume that $\wh \rho(A) = \C$. Then the following holds:

\item[\;\;\rm (i)]
The resolvent of $\wt A$ is compact, i.e. the spectrum of $\wt A$ is discrete.

\item[\;\;\rm (ii)]
If $\ker(B_I) = \{0\}$, $B_I = \IM(B)$, then
$\wt A$ is completely non-selfadjoint.
In particular, $\R \subset \rho(\wt A)$.

\item[\;\;\rm (iii)]
If $\wt A$ is completely non-selfadjoint,
then $\wt A$ belongs to the class $C_0$.
Moreover, the perturbation determinant
$d(z) := \gD^\Pi_{\wt A/\wt A^*}(z)$, $z \in \C_+$,
is an annihilation function for $\wt A$.

\item[\;\;\rm (iv)]
If $\wt A$ is completely non-selfadjoint and complete, then
the  annihilation function $d(\cdot)$ is minimal for $\wt A$ if
and only if the geometric multiplicity of any eigenvalue $z$ of $\wt A$ is
one, i.e. $\dim(\ker(\wt A- z))=1$ or, equivalently, $\dim(\ker(B - M(z))) = 1$.
In particular, $d(\cdot)$ is minimal if $n_\pm(A) =1$.
\end{proposition}
}}
\begin{proof}
(i) Follow \cite[Section 105]{AG81} $z \in \C$ belongs to the continuous spectrum of $A$
if $z \in \C \setminus \wh \rho(A)$ and $\ran(A-z)$ is not closed. By
\cite[Theorem 100.1]{AG81} all selfadjoint extension of $A$ have the
same continuous spectrum. Since $\wh\rho(A) = \C$ we find that for all
selfadjoint extensions the continuous spectrum is empty. Hence, for
any selfadjoint extension the continuous spectrum is discrete which
shows that the resolvent of any selfadjoint extension of $A$ is
compact. Applying Krein formula \eqref{2.30} we find that the
resolvent of any other extension is compact, too.

(ii) Let us show  that $\wt A$ is completely
non-selfadjoint. From Proposition \ref{prop2.1} it follows
that the operator $B$ has to be dissipative, i.e. $\IM(B) \ge 0$.

Since the spectrum of $\wt A$ is discrete, it suffices to show that
$\wt A$ has no real eigenvalues. Let us assume that $x \in \gs(\wt A)
\cap \R$. It follows from the Green formula \eqref{2.0} that for any
$x\in \R$  the following identity holds
\bead
\IM\bigl( (A_B - x)f, f\bigr) = -i[(B\Gamma_0 f,\Gamma_0 f)_{\cH} -
(\Gamma_0 f, B\Gamma_0 f)_{\cH}] \nonumber  \\
 = 2(B_I\Gamma_0f, \Gamma_0f)_{\cH} =
2\left\|\sqrt{B_I}\Gamma_0 f\right\|^2_{\cH}, \qquad f\in \dom(A^*).
\eead
Let $f\in\ker(A_B-x)$. Since $\ker(B_I) = \{0\}$, the above identity yields  $\Gamma_0 f=0$
and $\Gamma_1 f=B\Gamma_0 f=0$. Thus, $f\in\dom(A)$  and
$f\in\ker(A-x)$.  This contradicts the simplicity of $A$. Thus, $\R
\subset \rho(\wt A))$.

(iii)
Since $A$ has finite deficiency indices,  $\{\wt A,\wt A^*\} \in
\gotD^\Pi$. Thus,  the perturbation determinant $d(\cdot) := \gD^\Pi_{\wt
  A/\wt A^*}(\cdot)$ exists on  $\rho(\wt A^*)$.  By Theorem
\ref{th5.4}(ii),  $d(z)=\det(W_{\wt A}(z))$, $z \in
\rho(\wt A^*) \cap \rho(A_0)$ where $W_{\wt A}(\cdot)$ is the characteristic
function of $\wt A$ defined by \eqref{2.33} with $J=I$.

Since $\wt A$ is $m$-dissipative, the characteristic function
$W_{\wt A}(\cdot)$ is contractive in ${\C}_+$. Hence, $d(\cdot)$ is
contractive in $\C_+$.  Further, since, by (i), the resolvent of $\wt A$ is
compact, it follows from Theorem \ref{V.16}(ii)  that
$d(\cdot)$ is an inner function. Hence
$W_{\wt A}(\cdot)$ is an inner operator-valued function (see \cite[Corollary
V.6.3]{FN70}). Applying
\cite[Proposition VI.3.5]{FN70} we obtain that $\wt A$
belongs to the class $C_{\cdot\, 0}$.

Similarly, since the operator $-\wt A^*$ is $m$-dissipative too,  its
characteristic function $W_{-\wt A^*}(\cdot)$  is also inner operator-valued function
in $\C_+$. Hence $-\wt A^*$ belongs to the class $C_{\cdot\, 0}$ which  is equivalent  to
the inclusion  $\wt A \in C_{0\,\cdot}$. Therefore $\wt A\in C_{\cdot\,0} \cap C_{0\,\cdot} = C_{00}$.

Since $n_{\pm}(A)=n<\infty$, the contraction $T_{\wt A} =
(\wt A - i)(\wt A + i)^{-1}$  (cf. Appendix \ref{App.VII}), has equal finite defect numbers, i.e.
$\dim\left(\ran(I - T^*_{\wt A} T_{\wt A})\right) =
\dim\left(\ran(I- T_{\wt A}T^*_{\wt A})\right) < \infty$.
By  \cite[Theorem VI.5.2]{FN70},  $T_{\wt A} \in
C_0$ and the determinant $d(\cdot) = \det(W^\Pi_{\wt A}(\cdot))$ defined on  $\C_+$, is an
annihilation function for  $\wt A$ , i.e. $d(\wt A) = \wt d(T_{\wt A}) = 0$, cf. Appendix \ref{App.VII}.

(iv)
Let $\{e_j\}^n_{j=1}$ be a fixed orthonormal basis in
$\cH$.   Denote by $\gT_{\wt A}(\cdot)$ the matrix representation of the
characteristic function $W^\Pi_{\wt A}(\cdot)$  with respect
to the basis $\{e_j\}^n_{j=1}$. By $\adj(\gT_{\wt A}(\cdot))$ we denote
the adjugate matrix of $\gT_{\wt A}(\cdot)$.  Note that alongside the matrix  $\gT_{\wt A}(\cdot)$
the adjugate matrix function $\adj(\gT_{\wt A}(z))$ is holomorphic and
contractive in $\C_+$ too (cf. the proof of \cite[Proposition V.6.1]{FN70}).
By \cite[Theorem VI.5.2]{FN70}, the determinant
$d(\cdot) := \det(\gT_{\wt A}(\cdot)) = \det(W^\Pi_{\wt A}(\cdot))$  of
$\gT_{\wt A}(\cdot)$  coincides with the minimal annihilation  function $m_{\wt A}(\cdot)$
of $\wt A$ if and only if the entries of $\adj(\gT_{\wt A}(\cdot))$ have no
common non-trivial inner divisor in the algebra $H^\infty(\C_+)$.

On the other hand,  by (iii), the operator $\wt A\in C_0$. Therefore
it  is complete if and only if the determinant $d(\cdot)$ is a
Blaschke product (see \cite[Section 4.5]{Nikol80}).
Therefore it follows from the identity $\adj(\gT_{\wt A}(z))\cdot\gT_{\wt A}(z) = d(z)I_n$
that each common divisor $\varphi(\cdot)$ of the entries of
$\adj(\gT_{\wt A}(\cdot))$ has to be a divisor of
$d(\cdot)$. Therefore $\varphi(\cdot)$ always contains a Blaschke factor, i.e.
it admits the representation
$\varphi(\cdot) =  \varphi_1(\cdot)b_{z_0}^{m_0}(\cdot)$ where $b_{z_0}^{m_0}(\cdot)$ is a Blaschke factor
$b_{z_0}^{m_0}(z)  :=   \left(e^{i\ga_0}({z - z_0})/({z - \overline{z_0}})\right)^{m_0}$, $m_0\ge 1$,
$z_0 \in \C_+$, cf. \eqref{7.17A}.  Clearly, the latter happens if and only if $\adj(\gT_{\wt A}(z_0)) = 0_n := 0\cdot I_n$.
However, $\adj(\gT_{\wt A}(z_0)) =0_n$ is valid if and only if
$\rank(\gT_{\wt A}(z_0)) \le n- 2$, that is,
$\dim(\ker(\gT_{\wt A}(z_0))) = \dim(\ker(W^\Pi_{\wt A}(z_0))) \ge 2$.

Further, by Proposition \ref{t1.12}(ii),  $\dim \ker(\wt A- z) = \dim\ker\bigl(B-M(z)\bigr)$ for any
$z\in \rho(A_0)$.  Let us show  that
\be\la{6.20}
\dim(\ker(W^\Pi_{\wt A}(z_0))) =  \dim(\ker(B - M(z_0))), \qquad z_0\in \C_+.
\ee
Indeed, setting $T_1 := 2i(B^* - M(z_0))^{-1}B_I^{1/2}$ and using
$\ker(B_I) = \{0\}$ one immediately verifies that $\ker(T_1) = \{0\}$.
Further, we note that
$W_{A_B}(z_0)h_0 = 0$ if and only if $h_0 \in \ker(I + B_I^{1/2}T_1)$.
If $h_0 \in \ker(I + B_I^{1/2}T_1)$, then $T_1h_0 \in \ker(I + T_1B^{1/2}_I)$
which yields $\dim(\ker(I + B_I^{1/2}T_1)) \le \dim(\ker(I + T_1B^{1/2}_I)$.
Conversely, if $h_1 \in \ker(I + T_1B^{1/2}_I)$, then $B^{1/2}_Ih_1
\in \ker(I + B_I^{1/2}T_1)$ which proves
$\dim(\ker(I + T_1B^{1/2}_I)) \le \dim(\ker(B_I^{1/2}T_1))$. Hence
\bed
\dim(\ker(W_{\wt A}(z_0)) = \dim(\ker(I + B_I^{1/2}T_1)) = \dim(\ker(I + T_1B_I^{1/2})).
\eed
Combining this relation with the identity  $I + T_1B_I^{1/2} = (B^* - M(z_0))^{-1}(B - M(z_0))$ we  arrive at \eqref{6.20}.
Thus,  $d(\cdot) = \gD^\Pi_{\wt A/\wt A^*}(\cdot)$ is a minimal
annihilation function if and only if $\dim(\ker(B - M(z))) = 1$ for
$z \in \gs(\wt A)$ {\vio{which yields $\dim(\ker(\wt A - z)) =1$.}}
\end{proof}

\section{The case of additive  perturbations}\la{sec.VI}

Here we extend some  results of Section \ref{sec.V} for extensions
to the case  of additive trace class perturbations of
$m$-accumulative operators. We emphasize that all results of this
section are new for the case additive perturbations.

\subsection{The pairs of $m$-accumulative operators}\la{sec.VIA}

We start with two technical statements.
\bl\label{lem6.1}
Assume that  $H$ and $H'$  are maximal accumulative operators in
$\gH$ and $V\in\gotS_1(\gH)$,  Then
\begin{equation}\label{7.10}
\lim_{y\to\infty}y^2\tr\bigl((H'-iy)^{-1}V(H-iy)^{-1}\bigr)
= -\tr(V).
\end{equation}
\el
\begin{proof}
The proof is based on the following statement: Let $\{Z(y)\}_{y \in \R_+}$ be a family of
bounded operators such that $\slim_{y\to\infty}Z(y) = Z$. If $V \in
\gotS_1$, then {\vio{$\wt Z(y) := Z(y)V \in \gotS_1(\gotH)$,
$y \in \R_+$, tends to $ZV$}} in the $\gotS_1$-norm, see \cite[Theorem III.6.3]{GK69}).

Let $Z(y) := y^2(H-iy)^{-1}(H'-iy)^{-1}$, $y \in \R_+$. Since $H$
and $H'$ are $m$-accumulative,
  \bed
\slim_{y\to\infty}y(H - iy)^{-1} = iI \quad \mbox{and} \quad
\slim_{y\to\infty}y(H' - iy)^{-1} = iI,
  \eed
which yields $\slim_{y\to\infty}Z(y) = -I$. Applying the statement
above we get  {\vio{$\lim_{y\to\infty}\|Z(y)V + V\|_{\gotS_1} =
0$.}} Hence
\bed
\lim_{y\to\infty}y^2\tr\bigl((H'-iy)^{-1}V(H-iy)^{-1}\bigr) =
{\vio{\lim_{y\to\infty}\tr\bigl(Z(y)V\bigr)}} = -\tr(V)
\eed
which proves \eqref{7.10}.
\end{proof}
\bc\label{cor6.2a}
Let $V \in \gotS_1(\gotH)$ and let $H$ be maximal accumulative in $\gotH$.
Let also $V_I := \IM(V) = V^+_I - V^-_I$ where $V^\pm_I
\ge 0$.  If $H' := H +V$, then $z = x+iy \in \rho(H')$ for $y >
\|V^+_I\|$ and
\eqref{7.10} holds.
\ec
\begin{proof}
We note that the operator $H' - i\|V^+_I\|$ is accumulative.
Using the representation $H' - iy = H' - i\|V^+_I\| - i(y - \|V^+_I\|)$
we find that $i(y - \|V^+_I\|) \in \rho(H' - \|V^+_I\|)$ provided that
$y - \|V^+_I\| > 0$. Since $i(y - \|V^+_I\|) \in \rho(H)$
it remains to apply Lemma \ref{lem6.1}.
\end{proof}

Next we present a counterpart to Lemma \ref{IV.2} for additive
perturbations.
\bl\label{lem6.7}
Let $H$ be a maximal accumulative operator.

\item[\;\;\rm (i)]
 If $0 \le V_+ = V^*_+\in\gotS_1(\gotH)$, then there exists a non-negative
function $\xi_+(\cdot) \in L^1(\R)$ such that the representation
\be\la{7.41}
\det\left(I + V_+(H - z)^{-1}\right) =
\exp\left\{\frac{1}{\pi}\int_\R
\frac{\xi_+(t)}{t-z}dt\right\},  \quad z\in \C_+,
\ee
holds and $\tr(V_+) = \frac{1}{\pi}\int_\R\xi_+(t)dt$ is satisfied.

\item[\;\;\rm(ii)]
If $V = V^* \in \gotS_1(\gotH)$, then there exists a real-valued function
$\xi(\cdot)\in L^1(\R)$ such that the representation \eqref{7.41} is
valid with $V$ and $\xi(\cdot)$ in place of $V_+$ and $\xi_+(\cdot)$,
in particular,
\be\la{6.3x}
\tr(V) = \int_R\xi(t)dt \qquad \mbox{and} \qquad \int_R|\xi(t)|dt \le
\|V\|_{\gotS_1}.
\ee
\el
\begin{proof}
(i) Let $V = V_+ \ge 0$.
We mimic the proof of Lemma \ref{IV.2}(i) replacing $B$ and $M(z)$  by $H$ and
and  $z$, respectively. {\vio{Doing so}} we find a non-negative function
$\xi_+(\cdot)$ satisfying $\int\tfrac{1}{1+t^2}\,\xi_+(t)dt < \infty$
and a positive constant $c_+$ such that the representation
\be\label{6.2A}
\det(I + V_+(H - z)^{-1}) = c_+ \exp\left\{\frac{1}{\pi}\int_\R
  \left(\frac{1}{t-z} - \frac{t}{1+t^2}\right)\xi_+(t) dt \right\},
\ee
$z \in \C_+$, is valid. Setting $T(z) := \sqrt{V_+}(H - z)^{-1}\sqrt{V_+}$, $z \in
\C_+$, we define a family of dissipative operators. Clearly,  $\lim_{y\to\infty}\|T(x+iy)\| = 0$, $x \in \R.$  Hence, $0 \in \rho(I +
T(x + iy))$ for any  fixed $x \in \R$ and  sufficiently large $y > 0$. Thus, for  $y$
large  enough we can take logarithm of both sides of \eqref{6.2A} using definition \eqref{2.19},
\be\la{6.3}
\log\det(I + V_+(H - z)^{-1}) = \log(c_+) + \frac{1}{\pi}\int_\R\left(\frac{1}{t-z} -
\frac{t}{1+t^2}\right)\xi_+(t) dt,
\ee
$z \in \C_+$. Hence
\bed
\IM(\log\det(I + T(z))) =
\frac{1}{\pi}\int_\R\frac{y}{(t-x)^2+y^2}\xi_+(t)dt, \quad z = x+iy
\in \C_+.
\eed
Using \eqref{2.20} we obtain
\be\la{6.3a}
\IM(\tr(\log(I + T(z)))) = \frac{1}{\pi}\int_\R\frac{y}{(t-x)^2+y^2}\xi_+(t)dt, \quad z = x+iy
\in \C_+.
\ee
It is easily seen that  $s-\lim_{y\uparrow\infty}(H-x-iy)^{-1}=0$ and $s-\lim_{y\uparrow\infty}(-iy)(H-x-iy)^{-1} = I$.
Since  $V_+\in{\gotS_1}(\gH)$, it follows with account of \cite[Theorem 3.6.3]{GK69} that
\bed
\lim_{y\uparrow \infty}\|T(x+iy)\|_{\gotS_1} = 0\quad \text{and}\quad \lim_{y\uparrow\infty}\|(-iy)T(x+iy)-V_+\|_{\gotS_1}=0.
\eed
Combining  these relations with definition \eqref{2.17} we obtain
\bed
\begin{split}
\lim_{y\uparrow\infty} y&\log\bigl(I+T(x+iy)\bigr)\\
& = -\lim_{y\uparrow\infty}(-iy)T(x+iy)\lim_{y\uparrow\infty}\int_{\R_+}\bigl(I+T(x+iy)+i\lambda\bigr)^{-1}(1+i\lambda)^{-1}d\lambda\\
& = -V_+\int_{\R_+}(1+i\lambda)^{-2}d\gl = i V_+
\end{split}
\eed
{\vio{for any fixed $x \in \R$}}. Notice that the
  convergence takes place in the $\gotS_1$-norm. Hence for any fixed $x \in \R$
\be\la{6.4a}
\lim_{y\to\infty}y\IM(\tr(\log(I + T(x+iy)))) = \tr(V_+).
\ee
On the other hand, multiplying the identity  \eqref{6.3a} by $y$ and
tending $y$ to $+\infty$ we arrive at the equality \eqref{6.4a} with
$\tfrac{1}{\pi}\int_\R\xi_+(t) dt$ in place of  $\tr(V_+)$.
Thus, these two quantities equal, $\tfrac{1}{\pi}\int_\R\xi_+(t) dt = \tr(V_+)$, in particular, $\xi_+(\cdot) \in L^1(\R)$.
Combining the later  inclusion  with representation \eqref{6.3} yields the representation
\be\label{6.8A}
\det(I + V_+(H - z)^{-1}) = c'_+ \exp\left\{\frac{1}{\pi}\int_\R \frac{\xi_+(t)}{t-z} dt \right\},
\quad z \in \C_+,
\ee
where
\bed
c'_+ :=
c_+\exp\left\{-\frac{1}{\pi}\int_\R\frac{t}{1+t^2}\xi_+(t)dt\right\}.
\eed
Setting in \eqref{6.8A} $z=iy$ and  tending $y$ to $+\infty$ and
noting that  $\lim_{y\to\infty}\det\left(I + V_+(H - iy)^{-1}\right) = 1$ we
find $c'_+ = 1$ which proves \eqref{7.41}.

(ii) Setting  $K := H- V_-$ where $V = V_+ - V_-$,
$V_\pm \ge 0$, and using the chain rule for perturbation determinants we get
\be\label{6.9}
\det(I + V(H-z)^{-1}) = \frac{\det(I +V_+(K-z)^{-1})}{\det(I + V_-(K-z)^{-1})}, \qquad z \in \C_+.
\ee
Applying (i) we get the representation  \eqref{7.41} with non-negative
$\xi_+(\cdot) \in L^1(\R)$ and {\vio{a}} similar  representation
 with non-negative $\xi_-(\cdot) \in L^1(\R)$ for ${\det(I + V_-(K-z)^{-1})}$. Setting $\xi
:= \xi_+(t) - \xi_-(t)$, $t \in \R$ and using \eqref{6.9}  we arrive at  the representation
\eqref{7.41} for $\det(I + V(H - z)^{-1})$, $z \in \C_+$. Further, since  $\tfrac{1}{\pi}\int_\R \xi_\pm(t) = \tr(V_\pm)$,
we get $\tfrac{1}{\pi}\int_\R\xi(t) dt = \tr(V)$. Moreover,
\bed
\frac{1}{\pi}\int_\R|\xi(t)|dt \le \frac{1}{\pi}\int_\R \xi_+(t) dt + \frac{1}{\pi}\int_\R \xi_-(t) dt =
\tr(V_+) + \tr(V_-) = \|V\|_{\gotS_1},
\eed
which proves \eqref{6.3x}.
\end{proof}
\br\la{remVI.4}
{\rm
Lemma \ref{lem6.7} can be proved in  a quite different way using
the classical results of Krein in \cite{K53b,K62a}, see also
\cite{K64} and \cite{BirPush98}. Indeed, since $H$ is a maximal
accumulative operator it admits a selfadjoint dilation,
that is, there is a selfadjoint operator $K$ in a larger
Hilbert space $\gotK \supseteq \gotH$  such that
\bed
(H-z)^{-1} = P^\gotK_\gotH (K - z)^{-1}\upharpoonright\gotH, \qquad z
\in \C_+.
\eed
cf.\cite{FN70}. Notice that
\bed
\gD_{H'/H}(z) = \det(I_\gotH + V(H-z)^{-1}) = \det(I_\gotK + V(K-z)^{-1}) = \gD_{K'/K}(z),
\eed
$z \in \C_+$, where $H' := H + V$ and $K' := K+ V$. By Theorem 1 of \cite{K62a} we
immediately find a real-valued function $\xi(\cdot) \in
L^1(\R)$ such that statement (ii) of Lemma \ref{lem6.7} is valid, in particular,
the relations \eqref{6.3x} hold. Moreover, If $V \ge 0$, then the same theorem
guarantees the existence of a non-negative function $\xi(\cdot) \in
L^1(\R)$ such that (i) of Lemma \ref{lem6.7} is valid.
}
\er
Using Lemma \ref{lem6.7} one readily  verifies trace formula \eqref{5.1A}.
{\vio{Passing to $m$-accumulative operators $H$ we firstly prove
an additive  counterpart of   Lemma \ref{IV.3a}.}}
\bl\label{cor6.4}
Let $H$ be a maximal accumulative operator in $\gotH$.

\item[\;\;\rm (i)]
Let $0 \le V_+ = V^*_+ \in \gotS_1(\gotH)$. If the condition
\be\la{6.10b}
(V_+f,f) \le -\IM(Hf,f), \quad f \in \dom(H).
\ee
is satisfied, then the function $w_+(z) := \det(I + iV_+(H-z)^{-1})$ admits the representation
\be\label{7.1B}
w_+(z)  =
 \exp\left\{\frac{i}{\pi}\int_\R
\frac{1}{t-z}\eta_+(t)dt\right\},  \qquad z\in \C_+,
\ee
with non-negative $\eta_+(\cdot) \in L^1(\R)$. Moreover, the representation
$\eta_+(t) = -\ln(|w_+(t + i0)|)$  holds for a.e. $t \in \R$ where
$w_+(t + i0) := \lim_{y\downarrow 0}w_+(t+iy)$.

In addition, the function $w^+(z) := \det(I - iV_+(H-z)^{-1})$, $z \in \C_+$,
admits the representation
\be\label{7.1B-}
w^+(z)  =
 \exp\left\{-\frac{i}{\pi}\int_\R
\frac{1}{t-z}\eta^+(t)dt\right\},  \qquad z\in \C_+,
\ee
with non-negative $\eta^+(\cdot) \in L^1(\R)$ such that
$\eta^+(t) = \ln(|w^+(t + i0)|) := \lim_{y\downarrow 0}w^+(t+iy)$
holds for a.e. $t \in \R$.

\item[\;\;\rm (ii)]
If  $V = V^*\in\gotS_1(\gH)$ and the condition
\be\la{6.14}
(Vf,f) \le -\IM(Hf,f), \qquad f \in \dom(H),
\ee
is satisfied, then the perturbation determinant  $w(z) := \det(I + iV(H-z)^{-1})$, $z \in
\C_+$, admits the representation
\be\label{7.1C}
w(z)  =
 \exp\left\{\frac{i}{\pi}\int_\R
\left(\frac{1}{t-z} \right)\eta(t)dt\right\},  \qquad z\in \C_+,
\ee
where $\eta(\cdot) \in L^1(\R)$  is \emph{real}-valued  and
$\eta(t) = -\ln\left(\left|w(t+i0)\right|\right)$ for a.e.
$t \in \R$.
\el
\begin{proof}
(i) In fact, the proof of Lemma \ref{IV.3a}(i) remains true if we
replace $B$ and $M(z)$ by $H$ and $z$, respectively. Hence there
exists a non-negative function $\eta_+(\cdot) \in
L^1(\R,\tfrac{1}{1+t^2}dt)$ such that the representation \eqref{4.180} holds, i.e.
  \be\label{6.15}
w_+(z) = \varkappa_+ \exp\left\{\frac{i}{\pi}\int_\R
\left(\frac{1}{t-z} - \frac{t}{1+t^2}\right) \eta_+(t)dt\right\},
\qquad z \in \C_+,
  \ee
where $\eta_+(t) = -\ln(|\det(w_+(t+i0))|)$ for a.e. $t \in \R$.
It follows that
\bed
|w_+(z)| = \exp\left\{-\frac{1}{\pi}\int_\R \frac{y}{(t-x)^2 +
    y^2}\eta_+(t)dt\right\},
\qquad z = x + iy \in \C_+,
\eed
or
\be\label{6.14A}
\exp\left\{\frac{1}{\pi}\int_\R \frac{y}{(t-x)^2 + y^2}\eta_+(t)dt\right\}
= \left|\frac{1}{w_+(z)}\right|, \qquad z = x + iy \in \C_+.
\ee
Notice that  $\frac{1}{w_+(z)} = \det(I - iV_+(H'-z)^{-1})$, where  $H' := H +
iV_+$. {\vio{Applying}} the known estimate for the perturbation determinant (see \cite[Section IV.1]{GK69}) we
arrive at the inequality
\bed
\frac{1}{|w_+(iy)|} \le \exp\left\{\left\|V_+(H' -iy)^{-1}\right\|_{\gotS_1}\right\} \le
 \exp\left\{\frac{\|V_+\|_{\gotS_1}}{y - \|V_+\|}\right\}, \quad y > \|V_+\|.
\eed
Combining this estimate  with relation \eqref{6.14A}  {\vio{this}} yields
\bed
\frac{1}{\pi}\int_\R \frac{y^2}{t^2 + y^2}\eta_+(t)dt \le
\|V_+\|_{\gotS_1}\frac{y}{y - \|V_+\|},
\quad y > \|V_+\|.
\eed
In turn, tending $y$ to $+\infty$ and applying  the
monotone convergence theorem yields
\bed
\|\eta_+\|_{L^1(\R)} = \int_\R\eta_+(t)\,dt \le \pi\|V_+\|_{\gotS_1}.
\eed
Moreover, setting
$\varkappa'_+ = \varkappa_+\exp\left\{-\frac{i}{\pi}\int_\R\frac{t}{1+t^2}\eta_+(t)dt\right\}$
and using \eqref{6.15} we arrive at the representation
\bed
w_+(z) = \varkappa'_+
\exp\left\{\frac{i}{\pi}\int_\R\frac{1}{t-z}\eta_+(t)dt\right\}, \quad z \in
  \C_+.
\eed
Finally, since  $\lim_{y\to\infty}w_+(iy) = 0$, we get $\varkappa'_+ = 1$ which
proves \eqref{7.1B}.

Notice that $H - iV_+$ is a maximal dissipative
operator. Obviously we have
\bed
w^+(z) = \det(I - iV_+(H-z)^{-1}) = \frac{1}{\det(I + iV_+(H - iV_+ -z)^{-1})}, \quad z \in \C_+.
\eed
By the result above there is a non-negative function $\eta^+(\cdot)
\in L^1(\R)$ such that the representation
\bed
\det(I + iV_+(H - iV_+ -z)^{-1}) =
\exp\left\{\frac{i}{\pi}\int_\R\frac{\eta^+(t)}{t-z)}dt\right\}, \quad z \in
  \C_+,
\eed
holds which immediately proves \eqref{7.1B-}.

(ii) Let $V = V_+ - V_-$, $V_\pm \ge 0$. We set $H_- := H -
iV_-$. Notice that $H_-$ is also a maximal accumulative operator.
Setting $w_\pm(z) := \det(I + iV_\pm(H_- -z)^{-1})$, $z \in \C_+$ and
using the chain rule (see Appendix \ref{B}, formula \eqref{1.1AB}) we get

\be\label{6.16A}
w(z) = \frac{w_+(z)}{w_-(z)}, \qquad z \in \C_+.
\ee
Next, we rewrite condition  \eqref{6.14} as
\bed
(V_+f,f) \le -\IM(Hf,f) + (V_-f,f) = -\IM(H_-f,f),
\quad f \in \dom(H_-).
\eed
Applying (i) we find that $w_+(\cdot)$ admits  {\vio{the}} representation \eqref{7.1B} with $\eta_+(t)\ge 0$.
Similarly, since  $(V_-f,f) \le -\IM(H_-f,f)$, $f \in
\dom(H_-)$, we obtain by applying (i) that  $w_-(\cdot)$ also admits a representation of type
\eqref{7.1B} with $\eta_-(t)\ge 0$ in place  of $\eta_+(t)$.
Combining \eqref{6.16A} with these representations  and setting $\eta(t) := \eta_+(t) - \eta_-(t)$, $t \in \R$,
we arrive at  \eqref{7.1C}.
\end{proof}

The counterpart of Theorem \ref{V.60} reads now as follows.
\bt\la{VI.5}
Let $H$ be a maximal accumulative operator,  $V \in
\gotS_1(\gotH)$ and let $H' = H+V$ be accumulative, i.e.
\be\la{6.7a}
\IM(Vf,f) \le -\IM(Hf,f), \quad f \in \dom(H).
\ee
Then $H'$ is maximal accumulative and there exists \emph{a
complex-valued function} $\go(\cdot) \in L^1(\R)$ such that the following holds:

\item[\;\;\rm (i)]
The perturbation determinant $\gD_{H'/H}(z)$, $z \in \C_+$, admits the
representation
\be\label{7.28}
\gD_{H'/H}(z) =
 \exp\left\{\frac{1}{\pi}\int_\R
\frac{\go(t)}{t-z}dt \right\},  \qquad z\in{\mathbb C}_+.
\ee

\item[\;\;\rm (ii)]
The trace formulas
\be\la{6.7b}
\tr\left((H'-z)^{-1} - (H-z)^{-1}\right) = -\frac{1}{\pi}\int_\R
\frac{\go(t)}{(t-z)^2}dt, \quad z \in \C_+,
\ee
and
\be\label{6.45}
\tr(V) = \frac{1}{\pi}\int_{\R}\go(t)dt.
\ee
hold.
\et
\begin{proof}
Clearly, $H'$ is $m$-accumulative because {\vio{ $H$ is so}} and $V$ is bounded.

(i) Let $V = V_R + iV_I$ where $V_R := \RE(V)$ and $V_I := \IM(V)$. We set $K := H + V_R$
and note that $K$ is $m$-accumulative. Using \eqref{6.7a} we find
\bed
(V_If,f) \le -\IM(Kf,f) = -\IM(Hf,f), \qquad f\in \dom(K) = \dom(H).
\eed
By Lemma \ref{cor6.4}(ii) there exists  a real-valued
function $\eta(\cdot) \in L^1(\R)$ such that the perturbation
determinant $\gD_{H'/K}(z) = \det(I + iV_I(K - z)^{-1})$, $z \in \C_+$,
admits the representation
\be\la{6.12b}
\gD_{H'/K}(z) = \exp\left\{\frac{i}{\pi}\int_\R\frac{\eta(t)}{t-z}
dt\right\}, \qquad z \in \C_+.
\ee
Furthermore, by Lemma \ref{lem6.7}{\vio{(ii)}}, there exists a real-valued function
$\xi(\cdot) \in L^1(\R)$ such that the perturbation determinant
$\gD_{K/H}(z) = \det(I + V_R(H-z)^{-1})$, $z \in \C_+$, admits the
representation
\be\la{6.12a}
\gD_{K/H}(z) = \exp\left\{\frac{1}{\pi}\int_\R\frac{\xi(t)}{t-z}dt\right\},
\quad z \in \C_+.
\ee
{\vio{By}} the chain rule  $\gD_{H'/H}(z) = \gD_{H'/K}(z)\gD_{K/H}(z)$, $z \in
\C_+$ (see formula  \eqref{1.1AB}), and setting $\go(t) := \xi(t) + i\eta(t)$, $t \in \R$,
we arrive at \eqref{7.28} with  a complex-valued function $\go(\cdot) \in L^1(\R)$.

(ii) Taking logarithmic derivative from both sides of \eqref{7.28} and using the property \eqref{1.1} we arrive at  the trace formula
\eqref{6.7b}.

Next, to prove \eqref{6.45}  we rewrite \eqref{6.7b} in the form
\bed
\tr\bigl((H' - z)^{-1}V(H - z)^{-1}\bigr) = 
 \frac{1}{\pi}\int_\R\frac{\go(t)}{(t-z)^2}dt,\qquad z\in \C_+,
\eed
and put here  $z=iy$. Then multiplying both sides by $y^2$  and tending $y$ to $+\infty$
with  account of Lemma \ref{lem6.1} we arrive at  \eqref{6.45}.
\end{proof}
\br\la{VI.6}
{\em
Let us compare Theorem \ref{VI.5}  with {\vio{Krein's results
of  \cite{K87}.}}

\item[\;\;\rm (i)]
Krein \cite{K87} considered the maximal accumulative operator
$H' :=H-iV_+$, with  $H = H^*$ and $ V_+ = V^*_+ \ge 0$. According to \cite[Theorem 9.1]{K87}
there exists  a non-decreasing
function $\tau(\cdot): \R \longrightarrow \R$ such that the
perturbation determinant $\gD_{H/H'}(z)$, $z \in \C_+$, admits the
representation
\be\la{6.16b}
\gD_{H/H'}(z) = \exp\left\{i\int_\R\frac{d\,\tau(t)}{t-z}\right\}, \quad
z \in \C_+.
\ee
{\vio{Lemma  \ref{cor6.4}(i)}} improves {\vio{Krein's result. Indeed, it is
shown}}  that the measure $d\tau(\cdot)$ is absolutely continuous, i.e. $d\tau(t)= \eta_+(t)dt$ where $\eta_+(\cdot)\ge 0$ and
$\eta_+(\cdot) \in L^1(\R)$. {\vio{Notice that in distinction to \eqref{7.28} Krein considers the perturbation determinant
$\gD_{H/H'}(z)$.}}

\item[\;\;\rm (ii)]
Theorem \ref{VI.5} generalizes Theorem 9.1 of \cite{K87} in two
directions.  Firstly, {\vio{$H$ can be $m$-accumulative
and, secondly, condition $\IM(V) \le 0$ in \cite{K87}  is relaxed to \eqref{6.7a}.}}

\item[\;\;\rm (iii)]
Let $H = H^*$ and  $H' = H+V$ where
$V$ is accumulative, $V\in \gotS_1(\gotH)$ {\vio{and the}}
condition $V_I\log(-V_I) \in \gotS_1(\gotH)$, $V_I = \IM(V) \le 0$, is
valid. By  \cite[Corollary 4.3]{AN90}, there exists \emph{a real-valued function}
$\xi(\cdot) \in L^1(\R,\tfrac{1}{1+t^2}dt)$ such that the trace formula
\be\la{6.17x}
\tr((H' - z)^{-1} - (H-z)^{-1}) = -\frac{1}{\pi}\int_\R \frac{\xi(t)}{(t-z)^2}dt,
\qquad z \in \C_+,
\ee
{\vio{is valid}}. Thus,  {\vio{alongside}}  representation
\eqref{6.7b} with a complex-valued function $\go(\cdot) \in L^1(\R)$
there is {\vio{the}} representation  \eqref{6.17x} with a real-valued
function $\xi(\cdot) \in L^1(\R,\tfrac{1}{1+t^2}dt)$.
}
\er

The trace formula \eqref{6.16b} can be extended to
a class of holomorphic in $\C_-$ functions $\Phi(\cdot)$ admitting the
representation
\be\label{6.55}
\Phi(z) = {\vio{\int_{[0,\infty)} \Psi(z,t) dp(t)}}, \quad z\in {\vio{\overline{\C_-}}},
\ee
where $p(\cdot)$ is a {\vio{complex-valued Borel measure on $[0,\infty)$ of finite variation, i.e
\bed
\int_{[0,\infty)}|dp((t)| < \infty,
\eed
and}}
\bed
{\vio{\Psi(z,t)}} :=
\begin{cases}
\frac{e^{-i t z}-1}{-it}, & t>0, \\
z, & t=0.
\end{cases},         \qquad z \in \overline{\C_-}.
\eed
It is well known that  any $m$-accumulative (in particular selfadjoint) operator $H$ in $\gotH$
generates a strongly continuous semigroup of contractions $e^{-itH}$, $t\ge 0$.
This fact allows one to define the  operator $\Phi(T)$ by
\be\la{6.180}
\Phi(H)h = \int_{[0,\infty)}{\vio{\Psi(H,t)}} h dp(t), \qquad h\in\dom(H).
\ee
{\vio{In general,}} $\Phi(H)$ is unbounded {\vio{ and closable such that}}
$\dom(\Phi(H)) \supseteq \dom(H)$. However, {\vio{if $\supp(p)\subset (0,\infty)$}},
then $\Phi(H)$ is bounded.

In \cite[Theorem 9.2]{K87} Krein {\vio{has}} shown that {\vio{for}} a selfadjoint operator $H =H^*$ and a maximal
accumulative operator $H' = H - iV_+$, $V_+ = V^*_+\ge 0$,
the trace formula
\bed
\tr\bigl(\Phi(H')-\Phi(H)\bigr) = -i\int_\R\Phi'(t)d\tau(t)
\eed
\newline
holds where $\Phi(\cdot)$ is given by \eqref{6.55} {\vio{and
$\tau(\cdot)$ is a non-decreasing function of finite variation}} such
that the representation \eqref{6.16b} is valid. {\vio{Notice that Krein's result becomes
comparable with those below if one changes the sign of the right-hand side, see Remark \ref{VI.6}(i).}}

We generalize  \cite[Theorem 9.2]{K87} as follows.
\bt\label{prop6.11}
Let the assumptions of Theorem \ref{VI.5} be satisfied and let
$\Phi(\cdot)$  be a complex function in $\C_+$ of {\vio{the}} form
\eqref{6.55}. Then both operators $\Phi(H')$ and  $\Phi(H)$ are well defined,  $\Phi(H') - \Phi(H)
\in \gotS_1(\gotH)$ and the following trace formula
\be\la{6.180a}
\tr(\Phi(H') - \Phi(H)) = \frac{1}{\pi}\int_\R \Phi'(t)\go(t) dt\,.
\ee
{\vio{holds}}
\et
\begin{proof}
We set
\be\label{6.27}
H_\ga = H(I + i\ga H)^{-1} \qquad \mbox{and} \qquad H'_\ga = H'(I +
i\ga H')^{-1},
\quad \ga > 0.
\ee
One easily verifies that $H'_\ga$ and $H_\ga$ are bounded accumulative
operators. Moreover, it is easily seen that
\be\label{6.29}
\slim_{\ga\to+0}(H'_\ga - z)^{-1} = (H'-z)^{-1} \quad \text{and}\quad
\slim_{\ga\to+0}(H_\ga - z)^{-1} = (H-z)^{-1}
\ee
{\vio{for $z \in \C_+$.}} Further, it  readily follows from definition \eqref{6.27}  that
\bed
(H_{\alpha}-z)^{-1} = \frac{i\alpha}{1-i\alpha z}I + \frac{1}{(1-i\alpha z)^2}\left(H - \frac{z}{1-i\alpha z}\right)^{-1}.
\eed
Combining this identity with a similar identity for $(H'_{\alpha}-z)^{-1}$ and applying
\eqref{6.7b} we get that for any $z \in  \C_+$
\be\la{6.18}
\begin{split}
&\tr\left((H'_\ga - z)^{-1} - (H_\ga - z)^{-1}\right)
= -\frac{1}{(1-i\alpha z)^2}\int_{\R}\frac{\go(t)}{(t-\frac{z}{1-i\alpha z})^2}dt\\
& = - \int_{\R}\frac{\go(t)}{\bigl(t-z(1+i\alpha t)\bigr)^2}dt
 = -\frac{1}{\pi}\int_\R \frac{\go(t)}{(1 + i\ga t)^2}\frac{1}{\left(z -\frac{t}{1+i\ga t}\right)^2} dt.
\end{split}
\ee
Let $\gG$ be a simple closed curve such that its interior contains
$\gs(H'_\ga) \cup \gs(H_\ga)$. Since $H_\ga$ and $H'_\ga$ are bounded, the Riesz-Dunford  functional calculus yields
\bed
e^{-isH'_\ga} = -\frac{1}{2\pi i}\oint_\gG e^{-sz}(H'_\ga-z)^{-1} dz,  \qquad s
\ge 0.
\eed
and
\bed
e^{-isH_\ga}  = -\frac{1}{2\pi i}\oint_\gG e^{-sz}(H_\ga-z)^{-1}{\vio{dz}},  \qquad s
\ge 0.
\eed
Hence
\be\label{6.31A}
\begin{split}
e^{-isH'_\ga} - e^{-isH_\ga} &=
-\frac{1}{2\pi i}\oint_\gG e^{-isz}\left((H'_\ga - z)^{-1} - (H_\ga - z)^{-1}\right) dz\\
&= \frac{1}{2\pi i}\oint_\gG e^{-isz}(H'_\ga - z)^{-1}V_\ga(H_\ga - z)^{-1}dz,
\end{split}
\ee
where
\be\la{6.19a}
\begin{split}
V_\ga := H'_\ga - H_\ga  &= H'(I + i\ga H')^{-1} - H(I + i\ga H)^{-1}\\
&= (I + i\ga H')^{-1}V(I + i\ga H)^{-1}, \quad \ga > 0.
\end{split}
\ee
Since  $V \in \gotS_1(\gotH)$, the last identity implies  $V_\ga \in \gotS_1(\gotH)$.
Combining this fact with \eqref{6.31A} {\vio{this}} yields  $e^{-isH'_\ga} - e^{-isH_\ga}
\in \gotS_1(\gotH)$.  Moreover,  we get from \eqref{6.31A}
\bed
\tr\left(e^{-isH'_\ga} - e^{-isH_\ga}\right) =
 -\frac{1}{2\pi i}\oint_\gG e^{-isz}\tr\left((H'_\ga - z)^{-1} -
   (H_\ga - z)^{-1}\right) dz.
\eed
Combining this formula with  \eqref{6.18} we obtain
\begin{align}\la{6.200}
\tr&\left(e^{-isH'_\ga} - e^{-isH_\ga}\right) = \\
&
\frac{1}{\pi}\int_\R \;dt\; \frac{\go(t)}{(1+i\ga t)^2}\frac{1}{2\pi i}
\oint_\gG\frac{e^{-isz}}{\left(z - \frac{t}{1+i\ga t}\right)^2}dz
= \frac{-is}{\pi}\int_\R  e^{-is\tfrac{t}{1+i\ga
    t}}\frac{\go(t)}{(1+i\ga t)^2}\;dt
\nonumber
\end{align}
Further, it follows from \cite[formula (IX.2.22)]{Ka76} and \eqref{6.19a} that for $s > 0$
\bea\label{6.33}
\lefteqn{
e^{-isH'_\ga} - e^{-isH_\ga} =}\\
& &
-i\int^s_0e^{-i(s-x)H'_\ga}(I + i\ga H')^{-1}V(I + i\ga H)^{-1} e^{-ixH_\ga}dx,
\nonumber
\eea
and
\be\label{6.34}
e^{-isH'} - e^{-isH} =
-i\int^s_0e^{-i(s-x)H'}Ve^{-ixH}dx, \qquad s > 0.
\ee
Since  $V \in \gotS_1(\gotH)$ {\vio{we find}}
$e^{-isH'_\ga} - e^{-isH_\ga} \in \gotS_1(\gotH)$ and  $e^{-isH'} - e^{-isH}
\in \gotS_1(\gotH)$, $s > 0$,  {\vio{see above}}.  Moreover, they imply
the following important estimates
\be\label{6.35}
\left\|e^{-isH'_\ga} - e^{-isH_\ga}\right\|_{\gotS_1} \le
  s\|V_\ga\|_{\gotS_1} \;\text{and} \; \left\|e^{-isH'} - e^{-isH}\right\|_{\gotS_1} \le  s\|V\|_{\gotS_1}.
\ee
Since $V \in \gotS_1(\gotH)$, it follows from \eqref{6.29} and   \eqref{6.19a} that
$\lim_{\ga\to 0}\|V_\ga - V\|_{\gotS_1} =0.$   Combining this relation with integral representations
\eqref{6.33} {\vio{and}} \eqref{6.34} we obtain
\be\la{6.210}
\lim_{\ga\to+0}\tr\left(e^{-isH'_\ga} - e^{-isH_\ga}\right)  =
\tr\left(e^{-isH'} - e^{-isH}\right), \quad s>0.
\ee
Since $\go(\cdot) \in L^1(\R)$ and $\ga>0$ the dominated convergence
theorem {\vio{(with majorizing function $|\go|$)}}
yields
\be\la{6.220a}
\lim_{\ga\to+0} \frac{-is}{\pi}\int_\R e^{-is\tfrac{t}{1+i\ga
    t}}\frac{\go(t)}{(1+i\ga t)^2}dt =
 \frac{-is}{\pi}\int_\R e^{-ist}\go(t)dt, \quad s > 0.
\ee
Taking into account \eqref{6.200},  \eqref{6.210} and \eqref{6.220a} we obtain
\be\la{6.240}
\tr\left(e^{-isH'} - e^{-isH}\right) = \frac{-is}{\pi}\int_\R e^{-ist}\go(t)dt, \quad s > 0.
\ee
On the other hand, since both $H$ and $H'$ are $m$-accumulative \eqref{6.180} yields the representation
\be\label{6.40}
\Phi(H') - \Phi(H) = \int_{{\vio{[0,\infty)}}}\frac{e^{-itH'} - e^{-itH}}{-it}dp(t).
\ee
Since the measure $p$ is finite, we obtain from \eqref{6.40} and  \eqref{6.35}   that $\Phi(H') - \Phi(H) \in
\gotS_1(\gotH)$ and
\bed
{\vio{\|\Phi(H') - \Phi(H)\|_{\gotS_1} \le \|V\|_{\gotS_1}\int_{[0,\infty)}|dp(t)|}}.
\eed
Combining  \eqref{6.240} with \eqref{6.40}  we {\vio{finally}} obtain
\bead
\lefteqn{
\tr(\Phi(H') - \Phi(H)) = \int_{{\vio{[0,\infty)}}}\frac{\tr\left(e^{-itH'} - e^{-itH}\right)}{-it}dp(t)} \\
& &
= \frac{1}{\pi}\int_{{\vio{[0,\infty)}}}dp(t)\int_\R e^{-itx}\go(x)\;dx\; =
\frac{1}{\pi}\int_\R dx\;\go(x)\int_{{\vio{[0,\infty)}}}e^{-itx}dp(t).
\eead
{\vio{By}} $\Phi'(x) = \int_{{\vio{[0,\infty)}}}e^{-itx}dp(t)$, $x \in \R$, we complete the proof.
\end{proof}

{\vio{
In fact the class of functions $\Phi$ introduced above is not
optimal even for selfadjoint operators. A more optimal class of
    functions in case of selfadjoint operators was introduced in
    \cite{AlekPel11,Pel05}.
}}

\subsection{Pairs $\{H,H'\}$ with one  $m$-accumulative operator}\la{sec.VIB}

Our next goal is to prove  trace formulas for pairs $\{H,H'\}$ with $m$-accumulative operator $H$,
i.e. to remove the condition \eqref{6.7a}. For this
purpose we need an analog of Lemma \ref{IV.3aB}. To this end we
recall a simple statement on the behavior at infinity of a
Blaschke products ${\cB}(\cdot)$ with non-real zeros
$\gL:=\{\lambda_j\}_{j\in\N}$ satisfying an additional
assumption
\be\label{6.1A}
\sum^{\infty}_{j=1}|\IM\lambda_j|<\infty
\ee
\bl[{\cite[Lemma 8.1]{K87}}]\label{lem6.0}
Let $\gL:=\{\lambda_j\}^\infty_{j=1}\subset{\C}\setminus{\mathbb R}$.
If the condition  \eqref{6.1A} is satisfied, then
\bed
\lim_{y\uparrow\infty}y^2\sum^{\infty}_{j=1}
\frac{\IM(\lambda_j)}{(iy- \lambda_j)(iy-\overline{\lambda}_j)} =
-\sum^{\infty}_{j=1}\IM\lambda_j.
\eed
and the (regularized) Blaschke product
\bed
\widetilde{\cB}(z)=
\prod^{\infty}_{j=1}\frac{z-\lambda_j}{z-\overline{\lambda}_j}
\eed
converges uniformly on any compact subset $\mathcal K\subset{\mathbb
C}$ satisfying $\dist(\mathcal K,\gL)>0$. Moreover,
the following relations hold
\bed
\lim_{y\uparrow\infty}\widetilde{\cB}(z)=1 \quad\text{and}\quad
\lim_{y\uparrow\infty}y \ln|\widetilde{\cB}(z)|=   -
2\sum^{\infty}_{j=1}\IM\lambda_j, \quad z=x+iy.
\eed
\el

Now we are ready to state a counterpart of Lemma \ref{IV.3aB} for  additive perturbations.
\bl\label{th6.3}
Let $H$ be a maximal accumulative operator in
$\mathfrak H$.

\item[\;\;\rm (i)]
If $0 \le V_+ = V^*_+ \in \gotS_1(\gotH)$ and the condition
\bed
(V_+f,f) \le -2\IM(Hf,f), \quad f\in \dom(H),
\eed
is satisfied, then the function $w_+(z) :=
\det\bigl(I + iV_+(H-z)^{-1}\bigr)$, $z \in \C_+$,
admits the representation
\be\label{7.1}
w_+(z) =  \prod^\infty_{j=1} \left(\frac{z - z^+_j}{z -
\overline{z^+_j}}\right)^{m^+_j}
 \exp\left\{\frac{i}{\pi}\int_\R
\frac{1}{t-z}d\mu_+(t)\right\},  \quad z\in \C_+,
\ee
where $\mu_+(\cdot)$ is a \emph{non-negative finite} Borel
measure,  $\{z^+_j\}^\infty_{j=1}$
is the set of zeros of  $w_+(\cdot)$  in $\C_+$,
and $\{m^+_j\}^\infty_{j=1}$  is the sequence of the corresponding  multiplicities.

\item[\;\;\rm (ii)]
If $V = V^* \in \gotS_1(\gotH)$ and the condition
\be\la{6.17a}
(Vf,f) \le -2\IM(Hf,f), \quad f\in \dom(H),
\ee
is satisfied, then the function $w(z) = \det(I + V(H-z)^{-1})$, $z \in
\C_+$, admits the representation
\eqref{7.1} where the measure $\mu_+(\cdot)$ is replaced by
a real-valued measure $\mu(\cdot)$ satisfying $\int_\R
|d\mu(t)| < \infty$, $\{z^+_j\}^\infty_{j=1}\subset{\C}_+$
are the zeros of the function $w(z)$ in $\C_+$
and $\{m^+_j\}^\infty_{j=1}$  their corresponding multiplicities.
\el
\begin{proof}
(i) We set  $H' = H+iV_+$.  Following  the proof of Lemma \ref{IV.3aB} with
$M(z)$ replaced by $z$,  we arrive at  the representation \eqref{4.180a}.
Obviously, we have
\bed
|w_+(z)| = |\cB_+(z)|\;\exp\left\{-\frac{1}{\pi}\int_\R
  \frac{y}{(t-x)^2 + y^2} d\mu_+(t)\right\},
\quad z = x+iy \in \C_+,
\eed
which yields
\bed
\exp\left\{\frac{1}{\pi}\int_\R \frac{y}{(t-x)^2 +
    y^2} d\mu_+(t)\right\} = |\cB_+(z)|\left|\frac{1}{w_+(z)}\right|,
\quad z \in \C_+ \setminus \bigcup^\infty_{j=1}\{z^+_j\},
\eed
where $\gs(H') \cap \C_+ = \bigcup^\infty_{j=1}\{z^+_j\}$. One easily gets
\bed
\frac{1}{w_+(z)} = \det\left(I - i\sqrt{V_+}(H'-z)^{-1}\sqrt{V_+}\right),
\quad z \in \C_+ \setminus \gs(H').
\eed
Combining this identity with the previous one and noting that $|\wt\cB(z)| \le 1$, $z \in \C_+$, we obtain
\bed
\exp\left\{\frac{1}{\pi}\int_\R \frac{y}{(t-x)^2 +
    y^2} d\mu_+(t)\right\} \le \left|\det\left(I - i\sqrt{V_+}(H'-z)^{-1}\sqrt{V_+}\right)\right|,
\eed
$z \in \C_+ \setminus \gs(H')$. In turn, combing this inequality  with a simple
known estimate  (see \cite[Section IV.1]{GK69}) we arrive at the estimate
\bed
\exp\left\{\frac{1}{\pi}\int_\R \frac{y}{(t-x)^2 +
    y^2} d\mu_+(t)\right\} \le
\exp\left\{\left\|V_+(H'-z)^{-1}\right\|_{\gotS_1}\right\},
\; z \in \C_+ \setminus \gs(H'),
\eed
which is equivalent to
\bed
\frac{1}{\pi}\int_\R \frac{y}{(t-x)^2 + y^2} d\mu_+(t)
\le \left\|V_+(H'-z)^{-1}\right\|_{\gotS_1},
\quad z = x + iy \in \C_+ \setminus \gs(H').
\eed
Further, combining  this estimate with the following  one
\bed
\left\|V_+(H'-z)^{-1}\right\|_{\gotS_1} \le
\|V_+\|_{\gotS_1}\frac{1}{y - \|V_+\|}, \qquad y > \|V_+\|,
\eed
we obtain
\bed
\frac{1}{\pi}\int_\R \frac{y}{(t-x)^2 + y^2} d\mu_+(t)
\le \|V_+\|_{\gotS_1}\frac{1}{y - \|V_+\|},
\qquad y > \|V_+\|.
\eed
Multiplying both sides by $y$ and tending $y$ to infinity we derive
\bed
\int_\R d\mu_+(t) \le \|V_+\|_{\gotS_1} = \tr(V_+).
\eed
The zeros $z^+_j$ of $w_+(z)$ lying in $C_+$ and their
multiplicities $m^+_j$ coincide with the eigenvalues of $H' := H + V_+$
and their algebraic multiplicities.  Since $\IM(Hf,f) \le 0$, $f \in \dom(H)$, we have
\be\label{7.22}
\IM(H'f, f) = \IM(Hf, f) + (V_+f, f) \le (V_+f,f), \quad f\in \dom(H),
\ee
Denoting by $\gH^+_p$ the (closed) invariant subspace
of $H'$ spanned  by the  (finite-dimensional) root subspaces
$\gotL_{z^+_j} := \ker(H'- z_j^+)^{m_j}$, $j\in \N$, and choose a Schur orthonormal basis
$\{f_k\}_{k\in\N}$ in  $\gH^+_p$ such that in this basis
the matrix of the operator $H'\upharpoonright\gH^+_p$ is triangular.
Taking into account \eqref{7.22} we get
\bead
\lefteqn{
0 \le \sum_j m^+_j \IM(\lambda^+_j) = \sum_k \IM(H'f_k,f_k)}\\
& &
= \sum_k\IM\bigl(Hf_k,f_k\bigr) + \sum_k(V_+f_k,f_k) \le \tr(V_+) <\infty. \nonumber
\eead
By Lemma \ref{lem6.0} the product $\wt B_+(z)= \prod_k
\left(\frac{z - z^+_j}{z - \overline{z^+_j}}\right)^{m^+_j}$
converges uniformly on compact subsets of ${\C}_+$. It is easily
seen that $\cB_+(z) = \wt \kappa\; {\wt \cB_+}(z)$ where $|\wt \gk| = 1$.
Again by Lemma \ref{lem6.0} we have $\lim_{y\uparrow\infty}{\wt \cB}_+(iy)=1$.
It follows from \eqref{4.180a} that
\be\la{6.48}
w_+(z) = \varkappa'\;\wt \cB(z) \exp\left\{\frac{i}{\pi}\int_\R
  \frac{1}{t-z}d\mu_+(t)\right\},
\quad z \in \C_+,
\ee
where
\bed
\varkappa' = \varkappa  \;\wt{\varkappa} \;\exp\left\{i\left(\ga_+ - \frac{1}{\pi}\int_\R
    \frac{t}{1+t^2}d\mu_+(t)\right)\right\}
\eed
Since $\lim_{y \uparrow \infty}w_+(iy) = 1$ we immediately obtain $\varkappa'
= 1$ which proves \eqref{7.1}.

(ii) We set $K := H - iV_-$ where $V := V_+ - V_-$, $V_{\pm} \ge
0$. We note $K$ is maximal accumulative. Using the chain rule for perturbation determinants
(see formula \eqref{1.1AB}) we easily get
\be\la{6.19}
w(z) = \frac{\det(I + iV_+(K - z)^{-1})}{\det(I + iV_-(K-z)^{-1})} =: \frac{w_+(z)}{w_-(z)}, \quad z \in \C_+.
\ee
Rewriting  \eqref{6.17a} in the form
\bed
(V_+f,f) - (V_-f,f) \le -2\IM(Hf,f), \quad f\in \dom(H),
\eed
we get for any $f \in \dom(K) = \dom(H)$
\bed
(V_+f,f) \le -2\IM(Hf,f) + (V_-f,f) \le -2\IM(Hf,f) + 2(V_-f,f) =
-\IM(Kf,f).
\eed
According to statement  (i)  the  perturbation determinant $w_+(z) := \det(I
+ iV_+(K-z)^{-1})$, admits  the representation \eqref{7.1}  $z \in \C_+$. Since
\bed
(V_-f,f) \le -\IM(Hf,f) + (V_-f,f) = -\IM(Kf,f), \quad f \in \dom(K) =
\dom(H),
\eed
Lemma \ref{cor6.4}(i) yields the following representation
\be\la{6.20a}
\det(I + iV_-(K-z)^{-1}) =
\exp\left\{\frac{i}{\pi}\int_\R\frac{\eta_+(t)}{t-z}dt\right\}, \quad
z \in \C_+.
\ee
Inserting \eqref{7.1} and \eqref{6.20a} into \eqref{6.19} we arrive at the representation
\bed
w(z) = \prod^\infty_{j=1} \left(\frac{z - z^+_j}{z -
\overline{z^+_j}}\right)^{m^+_j} \exp\left\{\frac{i}{\pi}\int_\R\frac{1}{t-z}d\mu(t)\right\},
\quad z\in \C_+.
\eed
Here $d\mu(t) = d\mu_+(t) - \eta_+dt$,  $\{z^+_j\}^\infty_{j=1}$  is the set zeros of
$w_+(\cdot)$ lying in $\C_+$ and $\{m^+_j\}^\infty_{j=1}$ the set
of the corresponding multiplicities.
Since the operator $H$ is $m$-accumulative, the function $w_-(z) = \det(I +
iV_-(K-z)^{-1}) = \Delta_{H/K}(z)$ has no zeros in $\C_+$.
Combining this fact with representation  \eqref{7.1} for $w_+(\cdot)$
we conclude  that the set  $\{z^+_j\}^\infty_{j=1}$  is  the set of zeros
of $w(\cdot)$ in $\C_+$ with the corresponding multiplicities  $\{m^+_j\}^\infty_{j=1}$.
\end{proof}

Now we are ready to obtain the trace formulas for a pair $\{H, H+V\}$
with  $m$-accumulative operator $H$. A counterpart of Theorem
\ref{th5.2} takes the following form for additive perturbations.
\bt\la{VI.8}
Let $H$ be a maximal accumulative operator in $\gotH$,  $V \in
\gotS_1(\gotH)$ and  let $H' := H+V$. Then {\vio{the following holds:}}

\item[\;\;\rm (i)]
There exists  a complex-valued  Borel measure
$d\nu(t) := id\mu_+(t) + \go(t)dt$ on $\R$ such that $d\mu_+(\cdot)$
is \emph{a non-negative finite Borel measure} on $\R$,  $\go(\cdot)\in L^1(\R)$
and  the perturbation determinant $\gD_{H'/H}(\cdot)$  admits the
representation
\be\la{6.21a}
\gD_{H'/H}(z) = {\vio{\prod^\infty_{j=1}}}\left(\frac{z - z^+_j}{z -
\overline{z^+_j}}\right)^{m^+_j}\exp\left\{\frac{1}{\pi}\int_\R\frac{1}{t-z}d\nu(t)\right\},
\quad z\in \C_+.
\ee
where  $d\nu(t) := id\mu_+(t) + \go(t)dt$, $\{z^+_j\}^\infty_{j=1}$  is
the set of  eigenvalues of $H'$ in $\C_+$, and $\{m^+_j\}^\infty_{j=1}$ is
the set of corresponding algebraic multiplicities.

\item[\;\;\rm (ii)]
The trace formula holds
\bea\la{6.4}
\lefteqn{
\tr\left((H'-z)^{-1} - (H-z)^{-1}\right) =}\\
& &
-{\vio{\sum^\infty_{j=1}}}\frac{2i \;m^+_j\cdot
\IM(z_j^+)}{(z-z^+_j)(z-\overline{z^+_j})} -
\frac{1}{\pi}\int_{\R}\frac{d\nu(t)}{(t-z)^2}, \qquad z \in \C_+.
\nonumber
\eea
In particular, one has
\be\label{7.2}
\begin{split}
\tr(V)
& = 2i\sum^{\infty}_{j=1}m^+_j\IM(z^+_j) +
\frac{1}{\pi}\int_{\mathbb R}d\nu(t) \\
& = 2i\sum_j \;m^+_j \IM(z^+_j)  +\frac{i}{\pi}\int_{\R}d\mu_+(t)
+\frac{1}{\pi}\int_{\R}\go(t)dt.
\end{split}
\ee
{\vio{and
\be\la{6.56}
\tr(V_I)
= 2\sum_j \;m^+_j \IM(z^+_j)  +\frac{1}{\pi}\int_{\R}d\mu_+(t) +\frac{1}{\pi}\int_{\R}\go_I(t)dt.
\ee
where $V_I := \IM(V)$ and $\go_I(t) := \IM(\go(t)) \le 0$, $t \in \R$.}}
\et
\begin{proof}
(i) Let $V_R := \RE(V)$ and $V_I = \IM(V)$. Further, let $V_I = V^+_I
- V^-_I$ be the spectral decomposition of $V_I$, i.e.  $V^\pm_I \ge 0$
and $V^+_IV^-_I = V^-_IV^+_I =0$. We set $K := H + \wt V$  and  $\wt V := V_R - i|V_I|$, where
$|V_I| = V^+_I + V^-_I$. Clearly, the
operator $K$ is  $m$-accumulative because so are $H$ and  $\wt V(\in [\gotH])$.

We put  $w_+(z) := \gD_{H'/K}(z) = \det(I + 2V^+_I(K - z)^{-1})$, $z
\in \C_+$. It is easily seen that  $(2V^+_If,f) \le -2\IM(Kf,f)$,
$f\in \dom(K)$. Therefore applying Lemma \ref{th6.3}(i) we arrive at
the representation \eqref{7.1} for   $w_+(\cdot)$ with the
\emph{non-negative finite Borel measure} $d\mu_+(\cdot)$,
the set $\{z^+_j\}^\infty_{j=1}$  of zeros of $w_+(\cdot)$ lying in $\C_+$
and the set $\{m^+_j\}^\infty_{j=1}$ of corresponding
multiplicities. However,  the zeros $\{z^+_j\}^\infty_{j=1}$ and their
multiplicities $\{m^+_j\}^\infty_{j=1}$ coincide with the eigenvalues of
$H'$ lying in $\C_+$ and their algebraic multiplicities,
respectively (see Appendix \ref{B}, property 4).

Further, since $H$ is accumulative,  we have
\bed
\IM(\wt Vf,f) = -(|V|f,f) \le -\IM(Hf,f),\qquad  f\in \dom(H).
\eed
Therefore, by Theorem \ref{VI.5}(i),  there exists  \emph{a complex-valued function}
$\go(\cdot) \in L^1(\R)$ such that the
following representation  holds
\be\la{6.24}
\gD_{K/H}(z) =  \det(I + (V_R - i|V_I|)(H - z)^{-1}) = \exp\left\{\frac{1}{\pi}\int_\R
\frac{\go(t)}{t-z}dt\right\},
\ee
$z\in \C_+$.
Setting $d\nu(t) := id\mu_+(t) + \go(t)dt$ we define a complex-valued  Borel
measure on $\R$ satisfying $\int_\R |d\nu(t)| < \infty$. Finally, combining
representation \eqref{7.1} for $w_+(\cdot) := \gD_{H'/K}(\cdot)$ with
representation  \eqref{6.24} for $\gD_{K/H}(\cdot)$ and  using the
identity $\gD_{H'/H}(\cdot) =\gD_{H'/K}(\cdot)\gD_{K/H}(\cdot)$ (see
\eqref{1.1AB}), we arrive at representation \eqref{6.21a}.

(ii) Clearly,   $\{z \in \C:
\IM(z) > \|V^+_I\|\} \subset \rho(H')$. Therefore
 formula \eqref{1.1} can be applied  to the determinant $\gD_{H'/H}(z)$
for $\IM(z) > \|V^+_I\|$. Taking the logarithmic  derivative of
both sides of \eqref{7.1} and applying  \eqref{1.1} we obtain
\bed
\tr\left((H - z)^{-1} - (H' - z)^{-1}\right) = \sum_j
 m^+_j\left(\frac{1}{z-z^+_j}-\frac{1}{z - \overline{z^+_j}}\right) +
 \frac{1}{\pi}\int_{\R}\frac{d\nu(t)}{(t-z)^2},
\eed
for $\IM(z) > \|V^+_I\|$ which proves \eqref{6.4}. Since $V$ is
bounded  one  rewrites this identity  as
\bed
\tr\bigl((H'-z)^{-1}V(H-z)^{-1}\bigr) =
\sum_j\frac{2i \;m^+_j\cdot
\IM(z_j^+)}{(z-z+_j)(z-\overline{z^+_j})} +
\frac{1}{\pi}\int_{\mathbb R}\frac{d\nu(t)}{(t-z)^2}.
\eed
Setting here $z=iy$, $y > \|V^+_I\|$,  multiplying both sides by
$y^2$ and passing to the limit as $y\uparrow\infty$ we obtain by combining
Lemma \ref{lem6.1} with  Lemma \ref{lem6.0} and applying
the dominated convergence theorem {\vio{the relation}}
\bed
\begin{split}
-\tr(V)
& =
-2i\sum_j\; m^+_j \IM(z^+_j) + \lim_{y\uparrow
\infty}\frac{y^2}{\pi}\int_{\R}\frac{d\nu(t)}{(t - iy)^2}\\
& = -2i\sum_j \;m^+_j \IM(z^+_j) - \frac{1}{\pi}\int_{\R}d\nu(t)  \\
& = -2i\sum_j \;m^+_j \IM(z^+_j)  - \frac{i}{\pi}\int_{\R}d\mu_+(t) -\frac{1}{\pi}\int_{\R}\go(t)dt.
\end{split}
\eed
{\vio{follows which implies \eqref{7.2}. In turn,  \eqref{7.2} yields  \eqref{6.56}}}.
\end{proof}

\section{Examples}\la{sec.VII}

\subsection{Matrix Sturm-Liouville operators on {$\R_+$}}\la{sec.VII.1}

 Let us consider the matrix Sturm-Liouville differential expression in $L^2(\R_+,\C^n)$
 \be\label{6.1}
(\mathcal A f)(x)  :=  -\frac{d^2}{dx^2}f(x) + Q(x)f(x), \quad
f=\col\{f_1,\ldots,f_n\},
\ee
with $n\times n$ selfadjoint matrix potential
$Q(\cdot)=Q(\cdot)^*\in L^1_{\rm loc}(\R_+,\C^{n \times n})$.

Denote  by $A = A_{\min}$ and $A_{\max}$ the minimal and the
maximal operators, respectively associated on $L^2(\R_+,\C^n)$ with the differential expression
\eqref{6.1}. Clearly, $A$ is symmetric. Assume also that $\mathcal
A$ is limit point at infinity, i.e. the deficiency
indices are minimal, $n_{\pm}(A)=n$.  It is known (see, for
instance, \cite[Section 5.17.4]{Nai69}) that  $A^*= A_{\max}$. The
latter means that the domain $\dom(A^*)$ is locally regular, i.e.
\be\label{6.2}
\dom(A^*)\subset W^{2,2}_{\rm loc}({\R}_+,\C^n) \quad
\text{and}\quad \chi_{[0,b]}\dom(A^*) = W^{2,2}([0,b],\C^n) \
\ee
for any $b>0$, where $\chi_\delta(\cdot)$ stands for the indicator of
a Borel subset $\delta$. $A^*$ is given by the differential expression \eqref{6.1} on
the domain $\dom(A^*)$. Therefore the trace operators $\Gamma_0,\Gamma_1:\dom(A^*)\to{\C}^n$,
\bed
\quad \gG_0f = f(0), \quad \gG_1f = f'(0),\quad f=\col\{f_1,\ldots,f_n\},
\eed
are well defined and the Green identity \eqref{2.0}  holds.
Moreover,  one easily proves that $\Pi
=\{\C^n,\gG_0,\gG_1\}$ forms a boundary triplet for $A^*.$ Hence
the minimal operator $A = A_{\min}$ is a restriction of $A^*$ to
the domain
\bead
\dom(A)  = \ker\gG_0\cap\ker\gG_1  = \{f \in \dom(A^*):\  f'(0) =
f(0) = 0\},
\eead
and due to \eqref{6.2} the regularity property $\dom(A) \subset
W^{2,2}_{0,\rm loc}(\R_+,\C^n)$ holds.

Notice that $\dom(A^*) = W^{2,2}(\R,\C^n)$ whenever
$Q\in L^{\infty}(\R_+,\C^{n\times n})$. Let
\bed
\Psi_j(z,x)=\col\{\Psi_{j_1}(z,x),\ldots,\Psi_{jn}(z,x)\},\qquad
j\in\{1,\ldots,n\},
\eed
be a basis in $\gN_z(A)=\ker(A^*-z)$ and let
$\Psi(z,x):=\bigl(\Psi_1(z,x),\ldots,\Psi_n(z,x)\bigr)$
$=\bigl(\Psi_{kj}(z,x)^n_{k,j=1}$ be the Weyl $n\times n$-matrix
solution of the equation $A^*f = zf$. Then the corresponding Weyl
function is
\bed
M(z) = {\Psi'(z,0)}{\Psi(z,0)}^{-1}.
\eed
We note that the assumption $n_{\pm}(A) = n$ is satisfied whenever
$Q(\cdot) = Q(\cdot)^* \in L^1(\R_+,\C^{n\times n}) \cap
L^\infty(\R_+,\C^{n \times n})$. In this case the Weyl matrix solution
$\Psi(z,x)$ is proportional to the so-called Jost function
$F(z,\cdot)$ which solves the equation
\be\la{6.0}
F(z,x) = e^{i\sqrt{z}x}I_n - \int^\infty_x
\;\frac{1}{\sqrt{z}}\sin\left(\sqrt{z}(x - t)\right)Q(t)F(z,t)dt,
\ z\in\C \setminus \{0\},
\ee
where $\IM(\sqrt{z})\ge 0$. Clearly,  $\Psi(z,x)=F(z,x)C(z)$
where $C(\cdot)$ is {\vio{a}} non-singular $n\times n$-matrix function,
$\det C(z)\not =0,\  z\in{\C}_{\pm}$. In this case the Weyl
function is $M(z)=F'(z,0)F(z,0)^{-1}$.

Let $\wt A'$ and $\wt A$ be proper extensions of $A$. Since
$n_\pm(A) = n<\infty$ condition \eqref{1.103} is always
satisfied. If both $\wt A'$ and $\wt A$ are disjoint from $A_0:=
A^*\upharpoonright\ker(\gG_0)$, then, by Proposition
\ref{prop2.1}, there exist bounded operators $B', B \in [\C^n]$
such that $\wt A' = A_{B'}$ and $\wt A = A_B$, i.e. $\dom(\wt A')
=  \{f \in \dom(A^*):\ f'(0) = B'f(0)\}$  and
$\dom(\wt A)  =  \{f \in \dom(A^*): f'(0) = Bf(0)\}$.
Thus, the boundary triplet $\Pi$ is regular for the pair $\{\wt
A',\wt A\}$ which yields $\{\wt A',\wt A\} \in \gotD^\Pi$. By
Lemma \ref{def7.3a}(iv) we get
\bed
\gD^\Pi_{\wt A'/\wt A}(z) =
\frac{\det\bigl(B'- M(z)\bigr)}{\det\bigl(B-M(z)\bigr)}=
\frac{\det\bigl(B'F(z,0)-F'(z,0)\bigr)}{\det\bigl(B F(z,0)-F'(z,0)\bigr)},\quad
z \in \rho(\wt A) \cap \rho(A_0).
\eed
In particular, if $Q\equiv 0$, then $M_0(z) = i\sqrt{z}I_n$ and
\bed
\gD^\Pi_{\wt A'/\wt A}(z) =
\frac{\det\bigl(B'- M_0(z)\bigr)}{\det\bigl(B - M_0(z)\bigr)} =
\frac{\det\bigl(B' -i\sqrt{z}\bigr)}{\det\bigl(B - i\sqrt{z}\bigr)},\quad z \in
\rho(\wt A) \cap \rho(A_0).
\eed
If $\wt A = A_0$, then the boundary triplet $\Pi$ is not regular
for $\{\wt A',\wt A\}$. Nevertheless, by Corollary \ref{IV.40}(i),
there exists a boundary triplet $\wt \Pi = \{\cH,\wt \gG_0,\wt
\gG_1\}$ for $A^*$ such that $\{\wt A',\wt A\} \in \gotD^{\wt
\Pi}$. If $\wt A'$ is disjoint from $A_1$, then $0\in \rho(B')$.
Thus, by Proposition \ref{IV.8}(iii),  there exists $\mu =0$ and a
constant $c \in \C$ such that
\bed
\gD^{\wt \Pi}_{\wt A',\wt A}(z) = c\det\left(I_n -
({B'})^{-1}M(z)\right) =
\frac{c}{\det(B')}\frac{\det(B'F(z,0) - F'(z,0))}{\det(F(z,0))},
\eed
for $z \in \rho(\wt A)$.  If $\wt A' = A_1$, then $B' = 0$ and
applying Proposition \ref{IV.8}(iii) with $\mu =1$ we find  a
constant $c \in \C$ such that
\be\la{6.6}
\gD^{\wt \Pi}_{\wt A'/\wt A}(z) = c \det(M(z)) =
 c\frac{\det(F'(z,0))}{\det(F(z,0))}, \quad z \in \rho(\wt A).
\ee

\subsection{Sturm-Liouville operators on $(0,b)$}\la{sec.VII.2}

Next we consider the differential expression \eqref{6.1} in
$L^2([0,b],\C^n)$ with $n\times
n$-matrix potential $Q(\cdot)=Q(\cdot)^*\in L^2([0,b],{\C}^{n\times n})$.
It is well known that the maximal operator
$A_{\max}$  associated in $L^2([0,b],{\C}^n)$ with the
differential expression $\cA$
\bed
(\cA[f])(x) := -\frac{d^2}{dx^2}f(x) + Q(x)f(x)
\eed
is given by
\be\label{6.7}
(A^*f)(x) := (\cA[f])(x), \quad f\in \dom(A^*) = W^{2,2}((0,b),\C^n).
\ee
where $f =\col\{f_1,\ldots,f_n\}$. The minimal operator $A=A_{\min}$ is a closed symmetric operator
given by
\be\la{7.6}
\begin{split}
(Af)(x) := & (\cA[f])(x),\\
f \in \dom(A) := &
\left\{W^{2,2}((0,b),\C^n):
\begin{matrix}
f(0) = f'(0) = 0\\
f(b) = f'(b) = 0
\end{matrix}.
\right\}
\end{split}
\ee
Notice that $A_{\max} = A^*$.
Due to the regularity property $\dom(A^*) = W^{2,2}((0,b),\C^n)$ the mappings
\be\label{6.8}
\gG_0f := \left(
\begin{matrix}
f(b)\\
-f(0)
\end{matrix}
\right),
\quad
\gG_1f := \left(
\begin{matrix}
-f'(b)\\[2mm]
-f'(0)
\end{matrix}
\right),\quad f \in  W^{2,2}((0,b),\C^n),
\ee
are  well defined. Moreover,  one easily checks that
$\Pi = \{\C^{2n},\gG_0,\gG_1\}$ forms a boundary triplet for
$A^*$.  Notice that $A_0 :=
A^*\upharpoonright\ker(\gG_0)$ and $A_1 :=
A^*\upharpoonright\ker(\gG_1)$ correspond to the Dirichlet and
Neumann extensions, respectively.

Let us introduce the $n\times n$ matrix solutions $C(z,x)$ and
$S(z,x)$
   \bead
\cA[C(z,x)] & = & zC(z,x), \quad C(z,0) = I_n, \quad C'(z,0) = 0_n\\
\cA[S(z,x)] & = & zS(z,x), \quad S(z,0) = 0_n, \quad S'(z,0)= I_n.
     \eead
Any $f_z \in \ker(A^* - z)$ admits the representation $f_z(x) =
C(z,x)\xi + S(z,x)\eta$  for some  $\xi,\eta  \in \C^n$. Hence
\bed
\gG_0f_z =
\begin{pmatrix}
C(z,b) & S(z,b)\\
-I_n & 0_n
  \end{pmatrix}
\begin{pmatrix}
\xi\\
\eta
\end{pmatrix}
\quad  \text{and}\quad
\gG_1f_z =
\begin{pmatrix}
-C'(z,b) & - S'(z,b)\\
0_n & -I_n
\end{pmatrix}
\begin{pmatrix}
\xi\\
\eta
\end{pmatrix}.
\eed
Combining these relations with  Definition \ref{Weylfunc}
we find that the Weyl function $M(\cdot)$ corresponding to the
triplet $\Pi$ is
\be\la{6.16}
\begin{split}
M(z) =&
\begin{pmatrix}
-C'(z,b)&-S'(z,b)\\
0_n& -I_n
\end{pmatrix}
    \begin{pmatrix}
0_n&-I_n\\
S(z,b)^{-1}&S(z,b)^{-1}C(z,b)
    \end{pmatrix}\\
=&  \begin{pmatrix}
-S'(z,b)S(z,b)^{-1}&C'(z,b)-S'(z,b)S(z,b)^{-1}C(z,b)\\
-S(z,b)^{-1}&-S(z,b)^{-1}C(z,b)
  \end{pmatrix}\\
=& - \begin{pmatrix}
S'(z,b)S(z,b)^{-1}& S^*(\bar z,b)^{-1} \\
S(z,b)^{-1}&-S(z,b)^{-1}C(z,b)
  \end{pmatrix},
 \qquad z \in \C_\pm.
\end{split}
 \ee
If $Q \equiv 0$, then the Weyl function $M_0(z)$ is
\be\la{6.17}
M_0(z) = -\frac{1}{\sin(\sqrt{z}b)}
\begin{pmatrix}
\sqrt{z}\cos(\sqrt{z}b)I_n & I_n\\
I_n & -\cos(\sqrt{z}b)I_n
\end{pmatrix}.
\ee
Let $B' =
\begin{pmatrix}
B'_{11} & B'_{12}\\
B'_{21} & B'_{22}
\end{pmatrix}
\quad \mbox{and} \quad B =
\begin{pmatrix}
B_{11} & B_{12}\\
B_{21} & B_{22}
\end{pmatrix}
$
where $B'_{ij}$, $B_{ij} \in  \C^{n\times n}$, $i,j\in \{1,2\}$.
Due to \eqref{6.8}  the extensions $\wt
A' = A_{B'}$ and $\wt A = A_B$ are given by
\bead
\wt A' & = & (\cA[f])(x), \quad f \in \dom(\wt A'),\\
\dom(\wt A') & = & \left\{f \in W^{2,2}((0,b),\C^n):
\begin{matrix}
f'(b) = -B'_{11}f(b) + B'_{12}f(0)\\
f'(0) = -B'_{21}f(b) + B'_{22}f(0)
\end{matrix}
\right\}
\eead
and
\bea\label{6.10}
\wt A & = & (\cA[f])(x), \quad f \in \dom(\wt A),\nonumber  \\
\dom(\wt A) & = & \left\{f \in W^{2,2}(0,b):
\begin{matrix}
f'(b) = -B_{11}f(b) + B_{12}f(0)\\
f'(0) = -B_{21}f(b) + B_{22}f(0)
\end{matrix}
\right\}.
\eea
Note that $\{\wt A',\wt A\} \in \gotD^\Pi$. By  Lemma
\ref{def7.3a}(iv),
\be\label{6.12}
\gD^\Pi_{\wt A'/\wt A}(z) = \frac{\det(B' - M(z))}{\det(B -
M(z))}, \qquad z \in \rho(\wt A) \cap \rho(A_0).
\ee
If for some boundary triplet $\wt\Pi$ the pair  $\{\wt A',\wt A\}
\in \gotD^{\wt\Pi}$, then, by Proposition \ref{III.5}, there exist
a constant $c \in \C$ such that
\be\label{6.13}
\gD^{\wt\Pi}_{\wt A'/\wt A}(z) = c\;\frac{\det(\wt B' - M(z))}{\det(\wt B -
M(z))}, \qquad z \in \rho(\wt A) \cap \rho(A_0) \cap \rho(A'_0).
     \ee
Slightly more complicated is the case when $\wt A' = A_{B'}$ and
$\wt A = A_0$. From Corollary \ref{IV.40}(i) we get the existence
of a boundary triplet $\wt \Pi$ such that $\{\wt A',\wt A\} \in
\gotD^{\wt \Pi}$ which is regular. By Proposition \ref{IV.8}(iii),
there  exist  a constant $c \in \C$ and a real number $\mu \in
\rho(B')$ such that
\bed
\gD^{\wt \Pi}_{\wt A'/\wt A}(z) = c\;\frac{\det(B' - M(z))}{\det(B' - \mu)}, \qquad z \in \rho(\wt A).
\eed
In particular, if $B' = 0$, then $\wt A' = A_1$. Therefore  chosen  $\mu = 1$
we get
\bed
\gD^{\wt \Pi}_{\wt A'/\wt A}(z) = c\;\det(M(z)), \quad z \in
\rho(\wt A).
\eed
This yields
\bed
\gD^{\wt \Pi}_{\wt A'/ \wt A}(z) = c\;\frac{\det C'(z,b)}{\det
S(z,b)}, \qquad z \in \rho(\wt A),\eed
which generalizes \eqref{6.6} to the case of a bounded interval.
\begin{proposition}\la{VII.1a}
Let $A$ be the minimal Sturm-Liouville operator on $[0,b]$ defined by \eqref{7.6} and
let $B\in{\C}^{2n\times 2n}$,   $B_I := \IM(B) \ge 0$, and $\ker
B_I=\{0\}$. Let also  $A_B = A^*\upharpoonright \dom(A_B)$, $\dom(A_B) := \ker(\gG_1 - B\gG_0)$.
Then:

\item[\;\;\rm (i)]
$A$ is simple.

\item[\;\;\rm (ii)]
$A_B$ is {\vio{a}} $m$-dissipative and completely
non-selfadjoint operator with discrete spectrum {\vio{such that
$\R \subseteq \rho(A_B)$}}. Additionally, $A_B$ is complete.

\item[\;\;\rm (iii)]
$A_B$ belongs to the class $C_{0}$. Moreover, the perturbation determinant
$d(\cdot) = \gD^{\Pi}_{A_B/A^*_B}(\cdot)$  is an annihilation function
for $A_B$, that is, $d(A_B) = 0$.

\item[\;\;\rm (iv)]
The annihilation function $d(\cdot)$ is minimal if
and only if $\dim\left(\ker\left(B-M(z)\right)\right) = 1$ for any
$z\in \sigma(A_B)\cap\C_+ = \sigma_p(A_B)$.
\end{proposition}
\begin{proof}
(i) It follows immediately from the Cauchy uniqueness theorem.

(ii) The first claim follows from Proposition \ref{V.21}(i)
{\vio{and}} (ii). Further, it is well known  that the resolvent of $A_0$ is of  trace class.
It follows from  Proposition  \ref{prop2.9} that the resolvent of $A_B$ is also
of trace class. Since, in addition, $A_B$ is $m$-dissipative, it
follows from \cite[Theorem V.6.1]{GK69}) that $A_B$
is complete.

(iii) and (iv) These statements follow from Proposition \ref{V.21}(iii) and (iv).
\end{proof}

\subsection{Second order elliptic operators in domains with compact boundary}\la{sec.VII.3}

\subsubsection{Elliptic background}

Here we present some known facts on second order elliptic operators, cf. \cite{Gru68} and \cite{Ma2010},  which we are need in the following .
Consider the second-order formally symmetric elliptic operator with  smooth
real coefficients in a domain $\Omega \subset \R^n$ with smooth compact boundary,
\be\label{7.14A}
{\mathcal A}:= - \sum_{j,k=1}^n \frac{\partial}{\partial x_j}
a_{jk}(x) \frac{\partial}{\partial x_j} + q(x), \quad a_{jk} = \overline{a}_{jk},\ \  q=\overline q  \in C^{\infty}({\overline\Omega}).
\ee
Recall  that ellipticity of $\cA$ means that its (principle)
symbol $a_0(x,\xi)$ satisfies
\bed
a_0(x,\xi) := \sum_{|\alpha|=|\beta| = 1}
a_{\alpha\beta}(x)\xi^{\alpha + \beta} \not =0, \qquad (x, \xi)
\in{\overline{\Omega}}\times ({\mathbb R}^n\setminus \{0\}).
\eed
Let  $A = A_{\min}$  be the minimal elliptic operator associated
in $L^2(\Omega)$ with the expression \eqref{7.14A}. Due to Green's
identity
the minimal operator $A = A_{\min}$  is symmetric in
$L^2(\Omega)$. Any proper extension $\widetilde A\in \Ext_A$ of
$A$ is called a realization of $\cA$. Clearly, any realization
$\widetilde A$ of $\cA$ is closable. We equip $\dom(A_{\max})$
with the corresponding graph norm. It is known (cf.\ \cite{Ber68,
LIoMag72}) for bounded domains that $\dom(A_{\min})
= H^{2}_0(\Omega)$ where the graph  norm of
$\dom(A_{\min})$ and the norm of the Sobolev space
$H^{2}_0(\Omega)$ are equivalent.
However, {\vio{in contrast}} to the case of ordinary differential
operators {\vio{one has instead of $\dom(A_{\max}) = H^{2}(\Omega)$}} always
\bed
H^{2}(\Omega)\subset \dom(A_{\max})\subset H^{2}_{\loc}(\Omega).
\eed
Since the symbol
$a_0(x,\xi)$ is real  the  operator $\cA$ is properly elliptic (see \cite{LIoMag72}).
\begin{Hypothesis} \label{h2.5}
Let $\cA$  be a formally symmetric  second order elliptic differential
expression of the form \eqref{7.14A} given on the bounded subset $\gO
\subseteq \R^n$ with smooth boundary $\partial\gO$ which is
uniformly elliptic. In addition, let $a_{\alpha\beta}(\cdot)\in C^{2}_b(\Omega)$
for $|\alpha| + |\beta| \le 2$ and $a_{\alpha\beta}(\cdot)\in C_{ub}^{\infty}(\Omega)$, cf. notation at the end of the introduction,
for $|\alpha| + |\beta| = 2$.
\end{Hypothesis}
In particular, assuming Hypothesis \ref{h2.5} we have
$\dom(A_{\min}) = H^{2}_0(\Omega)$.
Notice that for  bounded  $\Omega$  any elliptic differential
expression $\cA$ with $C({\overline{\Omega}})$-coefficients is
automatically uniformly elliptic in ${\overline{\Omega}}$.

Denote by $\frac{\partial}{\partial\nu}$  the conormal derivative:
\begin{equation}\label{7.3}
\frac{\partial}{\partial\nu} = \sum_{j,k=1}^n a_{jk}(x) \cos(n,
x_j) \frac{\partial}{\partial x_k}
\end{equation}
and  set
\bed
 G_0u := \gamma_0u := u|_{\partial \Omega}, \quad  G_1u :=
\gamma_0 \bigg(\frac{\partial u}{\partial\nu} \bigg) =
\bigg(\frac{\partial u} {\partial\nu}\bigg)\bigg|_{\partial
\Omega}, \qquad  u\in \dom(A_{\max}).
\eed
We define the Dirichlet  and Neumann realizations ${\wh A}_{G_0}$ and  ${\wh A}_{G_1}$ by setting
\be\label{7.16}
\begin{split}
 {\wh A}_{G_j} & := A_{\max}\upharpoonright \dom({\wh
A}_{G_j}),\\
\dom({\wh A}_{{G_j}}) & := \{u\in
H^{2}(\Omega) \,|\, G_ju =0\},\quad j\in \{0,1\}.
\end{split}
\ee
It is well known that under Hypothesis \ref{h2.5} the  realization ${\wh A}_{G_j}$
is  selfadjoint in $H^0(\gO) := L^2(\Omega)$, ${\wh A}_{G_j}= {\wh
  A}_{G_j}^*$, $j\in \{0,1\}$.

To apply Proposition \ref{IV.8} and Corollary \ref{IV.9} we need a
boundary triplet for $A^*.$ Note, that the classical Green's
formula reeds now as follows
\bea\label{6.5MFAT}
\lefteqn{
({\mathcal A}u,v) - (u, {\mathcal A}v) =  \int_{\partial
\Omega} \bigg( \frac{\partial u}{\partial\nu} \cdot {\overline v}
- u\cdot \overline{\frac{\partial v}{\partial\nu}}\bigg) ds}  \\
& &
= \int_{\partial \Omega} \Big( G_1u \cdot {\overline {G_0 v}} -
G_0u\cdot \overline{G_1 v}\Big)\;ds,  \qquad
u,v \in H^2(\Omega).
\nonumber
\eea
\begin{proposition}[{\cite{Gru68}}]\label{prop4.2Det}
Let the Hypothesis \ref{h2.5} be satisfied and let
$0 \in \rho\big({\widehat A}_{G_0}\big)$.
Then for any $s\in\R$ the operator ${G_0}$  isomorphically maps the set
\bed
Z^s_\cA(\Omega) := \{u\in H^s(\Omega): A_{\max}u=0\}
\eed
onto $H^{s-1/2}(\partial\Omega)$.
\end{proposition}
\bd[{\cite{Gru68, Vis63}}]\label{def3.2A}
{\rm
Assume Hypothesis \ref{h2.5}.
\item[\;\;(i)]
Let $z\in \rho({\widehat A}_{ G_0})$ and
$\varphi\in H^{s-1/2}(\partial\Omega)$, $s\in\R$. Then one defines
$P(z)\varphi$ to be the unique $u\in Z^s_{\cA-zI_{L^2(\Omega)}}(\Omega)$
satisfying $G_0 u=\varphi$.

\item[\;\;(ii)]
The Poincare-Steklov  operator $\gL (z)$ is defined  by
\be\label{3.5A}
\gL(z): H^{s-1/2}(\partial\gO) \to
 H^{s-3/2}(\partial\gO), \quad
\gL (z)\varphi = G_1 P (z)\varphi.
\ee
To be precise  we denote the operator $\gL(\cdot):\  H^{s}(\partial\gO) \to
 H^{s-1}(\partial\gO)$ by $\gL_{s}(\cdot)$. Notice that
$\gL_{s}(z) \in [H^{s}(\partial\gO),H^{s-1}(\partial\gO)]$.
}
\ed
Further, let $\Delta_{\partial\Omega}$ be the Laplace-Beltrami operator in
$L^2(\partial \Omega)$. {\vio{Since ${\Delta_{\partial\Omega}\ge 0}$}},
the operator  $(-\Delta_{\partial\gO} + I)^{-s/2}$ isomorphically maps $L^2(\partial\Omega)$ onto
$H^s (\partial\Omega)$, $s\in\R$.

Notice that the classical Green  formula \eqref{6.5MFAT}  cannot be extended to
$u,v\in\dom(A^*)$ since the traces $G_0 u$ and $G_1 u$ belong to
the spaces $H^{-1/2}(\partial\Omega)$ and
$H^{-3/2}(\partial\Omega)$, respectively (see\cite{LIoMag72, Gru09}). A
construction of a boundary triplet for $A^*$ as well as the
respective regularization of the Green formula \eqref{6.5MFAT}
goes back to the classical papers by Vishik \cite{Vis63} and Grubb
\cite{Gru68}. An adaptation of this construction to the case of boundary triplets
in the sense of Definition \ref{II.1}  was done in  \cite{Ma2010}.
{\vio{First we recall a result  from \cite{Ma2010} that
modifies and completes  \cite[Theorem 3.1.2 ]{Gru68}}}
\begin{proposition}[{\cite[Proposition 5.1]{Ma2010}}]\label{prop5.1}
Let the Hypothesis \ref{h2.5} be satisfied and let
$0 \in \rho\big({\widehat A}_{G_0}\big)$.
Then the following statements are valid:

\item[\rm\;\;(i)]
 The totality  $\Pi=\{\cH,\gG_0,\gG_1\}$, where $\cH:=
 L^2(\partial\gO)$ and
\be\label{3.10A}
\begin{split}
\gG_0 u  &:= (-\Delta_{\partial\gO} + I)^{-1/4}G_0u,\\
\gG_1 u  &:= (-\Delta_{\partial\gO} + I)^{1/4} (G_1  - \gL(0) G_0)u,
\end{split} \qquad u\in \dom(A_{\max}),
\ee
forms a boundary triplet for $A^*$. In particular,
the Green formula
\be\label{4.1}
(A^*u,v)_{L^2(\Omega)} - (u,A^*v)_{L^2(\Omega)} =
(\gG_1 u,\gG_0v)_{L^2(\partial\gO)}-(\gG_0 u,\gG_1v)_{L^2(\partial\gO)},
\ee
$u, v \in \dom(A^*)$, holds and $A_0 := A^*\upharpoonright\ker(\gG_0)
= \wh A_{G_0}$.

\item[\rm\;\;(ii)]
The  operator valued function $\gL_{-\frac{1}{2}}(z) - \gL_{-\frac{1}{2}}(0)$
has the regularity property
\begin{equation}\label{3.5AWeyl}
\gL_{-\frac{1}{2}}(z) - \gL_{-\frac{1}{2}}(0) \colon\ H^{-1/2}(\partial\Omega) \to
H^{ 1/2}(\partial\Omega), \qquad z\in\rho\big({\wh
A}_{G_0}\big).
\end{equation}
Moreover, $\gL_{-\frac{1}{2}}(z) - \gL_{-\frac{1}{2}}(0)\in [
H^{-1/2}(\partial\Omega),  H^{1/2}(\partial\Omega)]$
for any  $z\in\rho\big({\wh A}_{G_0}\big).$

\item[\rm\;\;(iii)]
The corresponding Weyl function is given by
\be\label{7.29}
 M(z) = {(-\Delta_{\partial\gO} + I)^{1/4}\bigl(\gL_{-\frac{1}{2}}
(z)-\gL_{-\frac{1}{2}}(0)\bigr) (-\Delta_{\partial\gO} + I)^{1/4}}, \quad
z\in \C_{\pm}.
\ee
\end{proposition}
In {\vio{contrast}} to the mapping $\Gamma= \{\Gamma_0, \Gamma_1\}:  \dom(A_{\max})\to {\vio{L^2(\partial\Omega)\times L^2(\partial\Omega)}}$,
the mapping
\bed
G=\{G_0,G_1\}:\  \dom(A_{\max})\to H^{-1/2}(\partial\Omega)\times H^{-3/2}(\partial\Omega)
\eed
is not surjective. The following statement  describes  the range  $\ran(G)$.
\bc\label{cor7.6}
{\vio{Let the assumptions of Proposition  \eqref{prop5.1} be satisfied}}. Then for any pair
$\{h_0,h_1\}\in  H^{-1/2}(\partial\Omega)\times H^{-3/2}(\partial\Omega)$ the system $G_j f= h_j, j\in\{0,1\}$,
has a solution $f\in\dom(A_{\max})$ if and only if
\be\label{7.23}
h_1-\gL_{-\frac{1}{2}}(0)h_0\in H^{1/2}(\partial\Omega).
\ee
\end{corollary}
\begin{proof}
If condition \eqref{7.23} is satisfied, then it follows from
Proposition  \eqref{prop5.1}(i) and the surjectivity of
$\Gamma= \{\Gamma_0, \Gamma_1\}:  \dom(A_{\max})\to L^2(\partial\Omega)\times L^2(\partial\Omega)$
that the system
\bed
\begin{cases}
\Gamma_0f = & (-\Delta_{\partial\gO} + I)^{-1/4}h_0\\
\Gamma_1 f= & (-\Delta_{\partial\gO} + I)^{1/4} \bigl(h_1-\gL_{-\frac{1}{2}}(0)h_0\bigr)
\end{cases}
\eed
has a (non-unique) solution $f\in\dom(A_{\max})$. According to
definition  \eqref{3.10A},  $f$ also satisfies the system $G_j f=h_j$,
$j\in \{0,1\}$.

{\vio{Conversely, assume that there is a vector $f \in \dom(A_{\rm max})$ satisfying
$G_jf = h_j$,\  $j \in \{0,1\}$. Hence,
\bed
h_1 - \gL_{-\frac{1}{2}}(0)h_0 = G_1f - \gL_{-\frac{1}{2}}(0)G_0f = (-\gD_{\partial\gO} + I)^{-1/4}\gG_1f \in H^{1/2}(\partial\gO)
\eed
which verifies \eqref{7.23}.}}
\end{proof}

Let $\wt A$ be any proper extension of $A := A_{\rm min}$. We set
\bed
\Xi_{\wt A} := G\dom(\wt A) \subseteq H^{-1/2}(\partial\gO) \times H^{-3/2}(\partial\gO).
\eed
Clearly,  $\Xi_{\wt A}$ is not necessarily the graph of an operator.
\bl\la{VII.7}
Let the Hypothesis \ref{h2.5} be satisfied.

\item[\rm\;\;(i)]
Let  $\wt A$ be a proper extension of $A := A_{\rm min}$.
There exists an operator $K: H^{-1/2}(\partial\gO) \longrightarrow
H^{-3/2}(\partial\gO)$ such that $\Xi_{\wt A} = \graph(K)$ if and only if
$\wt A$ and $\wh A_{G_0}$ are disjoint.

\item[\rm\;\;(ii)]
Let  $\wt A_1$ and $\wt A_2$ be proper extensions of $A := A_{\rm
  min}$ which are disjoint from $\wh A_{G_0}$. Let $\Xi_j := G\dom(\wt
A_j) = \graph(K_j)$, $j =1,2$. If $K_1 = K_2$, then $\wt A_1 = \wt
A_2$.
\el
\begin{proof}
(i) Assume that the  extensions $\wt A$ and $\wh A_{G_0}$ are disjoint. Let
$\{h_0,h_1\} \in \Xi_{\wt A} = G\dom(\wt A)$, $h_0 \in H^{-1/2}(\partial\gO)$,
$h_1 \in  H^{-3/2}(\partial\gO)$. If $h_0 = 0$, then $0 = G_0f$, $f
\in \dom(\wt A) \subseteq \dom(A_{\rm max})$. Hence $f \in \dom(\wh
A_{G_0})$ which yields $f \in \dom(\wt A) \cap \dom(\wh A_{G_0})$.
Since $\dom(\wt A)$ and $\dom(\wh A_{G_0})$ are disjoint we find $f
\in \dom(A)$. Hence $\gG_1f = 0$ which implies $h_1 = 0$. Therefore,
$\Xi_{\wt A}$ is the graph of a linear operator $K: H^{-1/2}(\partial\gO) \longrightarrow
H^{-3/2}(\partial\gO)$.

Conversely, if $\wt A$ and $\wh A_{G_0}$ are not disjoint, then there
is a non-trivial $f \in \dom(\wt A) \cap \dom(\wh A_{G_0})$ such that
$f \not\in \dom(A_{\rm min})$. Notice that $h_0 = G_0f = 0$. Since
$\Xi_{\wt A}$ is the graph of an operator we get $0 = h_1 = G_1f$. Hence $f
\in \dom(A_{G_0}) \cap \dom(A_{G_1})$ which yields $f \in \dom(A_{\rm
  min})$.

(ii) Let $\{h_0,h_1\} \in \Xi_1 = \Xi_2$. Then there are elements
$f_1 \in \dom(\wt A_1)$ and $f_2 \in \dom(\wt A_2)$ such that $h_0 =
G_0f_1 = G_0f_2$ and $h_1 = G_1f_1 = G_1f_2$. Setting $f := f_1 - f_2
\in \dom(A_{\rm \max})$ we find $G_0f  = 0$ and $G_1f = 0$. Hence
$f \in \dom(A_{G_0}) \cap \dom(A_{G_1}) = \dom(A_{\rm min})$. Hence
$f_1 = f_2 + f$ where $f_2 \in \dom(\wt A_2)$ and $f \in \dom(A_{\rm
  min})$ which yields $f_1 \in \dom(\wt A_2)$. Similarly we prove
$f_2 \in \dom(\wt A_2)$. Therefore $\dom(\wt A_2) = \dom(\wt A_2)$
which yields $\wt A_1 = \wt A_2$.
\end{proof}

If $\wt A$ is a proper extension of $A_{\rm min}$ which is disjoint
from $\wh A_{G_0}$, then there is an operator
$K: H^{-1/2}(\partial\gO) \longrightarrow H^{-3/2}(\partial\gO)$
such that $G\dom(\wt A) = \graph(K)$.
The converse is in general not true.
That means, not for every operator $K: H^{-1/2}(\partial\gO)
\longrightarrow H^{-3/2}(\partial\gO)$ there exists a proper
extension $\wt A$ of $A_{\rm min}$ such that $\graph(K) = G\dom(\wt A)$.
\bl\label{lem3.6}
Let the Hypothesis \ref{h2.5} be satisfied and let
$0 \in \rho\big({\widehat A}_{G_0}\big)$. Further, let
$K: H^{-1/2}(\partial\gO) \longrightarrow  H^{-3/2}(\partial\gO)$.
There exists a unique proper extension $\wt A$ of $A_{\rm min}$
such that $\wt A$ is disjoint from $\wh A_{G_0}$ and
$\graph(K) = G\dom(\wt A)$ if and only if
the regularity condition $\ran(K - \gL_{-\frac{1}{2}}(0)) \subseteq  H^{1/2}(\partial\Omega)$
is satisfied.
\el
\begin{proof}
Let $\{h_0,h_1\} \in \graph(K)$. Notice that $h_1 = Kh_0$. Using the
assumption $\ran(K - \gL_{-\frac{1}{2}}(0)) \subseteq  H^{1/2}(\partial\Omega)$
we get $h_1 - \gL_{-\frac{1}{2}}(0)h_0 = Kh_0 - \gL_{-\frac{1}{2}}(0)h_0 \in H^{1/2}(\partial\Omega)$. By Corollary
\ref{cor7.6} there is an element $f \in \dom(A_{\rm max})$ such that
$h_j = G_jf$, $j = 0,1$. We set
\bed
\dom(\wt A) := \{f \in \dom(A_{\rm max}): h_j = G_jf, \; j= 0,1, \; \{h_0,h_1\}
\in \graph(K)\}.
\eed
Obviously, we have $\dom(A_{\rm min} \subseteq \dom(\wt A) \subseteq \dom(A_{\rm
  max})$. Setting $\wt A := A_{\rm max}\upharpoonright\dom(\wt A)$
we define a proper extension of $A_{\rm min}$ such that $\graph(K) =
G\dom(\wt A)$. By Lemma \ref{VII.7}(i) we get that $\wt A$ is
disjoint from $\wh A_{G_0}$. The uniqueness follows
from Lemma \ref{VII.7}(i).

Conversely, if there is a proper extension $\wt A$ such that $G\dom(\wt A) =
\graph(K)$, then $h_j = G_jf$, $j = 0,1$, $f \in \dom(\wt A)$,
$\{h_0,h_1\} \in \graph(K)$. Applying Corollary \ref{cor7.6} we
get $h_1 - \gL_{-\frac{1}{2}}(0)h_0 \in H^{1/2}(\partial\gO)$. Hence $Kh_0 -
\gL_{-\frac{1}{2}}(0)h_0   \in H^{1/2}(\partial\gO)$ for any $h_0 \in
\dom(K)$. Therefore $\ran(K - \gL_{-\frac{1}{2}}(0)) \subseteq
H^{1/2}(\partial\gO)$.
\end{proof}

Let $K: H^{-1/2}(\partial\gO) \longrightarrow
H^{-3/2}(\partial\gO)$. We set
\be\la{7.24b}
\begin{split}
\wh A_K &:= A_{\rm max} \upharpoonright\dom(\wh A_K),\\
\dom(\wh A_K) &:= \{f \in \dom(A_{\rm max}): G_1f = KG_0f\}.
\end{split}
\ee
Obviously, $\wh A_K$ is a proper extension of $A_{\rm min}$. Let $\wt
A$ be a proper extension disjoint from $\wh A_{G_0}$ and let $K$ the
operator defined by Lemma \ref{VII.7}. Then a straightforward
computation shows that $\wt A = \wh A_K$. In general, for any operator
$K: H^{-1/2}(\partial\gO) \longrightarrow
H^{-3/2}(\partial\gO)$ one has only $G\dom(\wh A_K) \subseteq
\graph(K)$.
\bc\la{VII.9}
Let the assumptions of Lemma \ref{lem3.6} be satisfied and let
$K: H^{-1/2}(\partial\gO) \longrightarrow  H^{-3/2}(\partial\gO)$.
The condition $\graph(K) = G\dom(\wh A_K)$ is satisfied if and only if
$\ran(K - \gL_{-\frac{1}{2}}(0)) \subseteq H^{1/2}(\partial\gO)$.
\ec
\begin{proof}
Let $\{h_0,h_1\} \in \graph(K)$,
then by the assumption $\ran(K - \gL_{-\frac{1}{2}}(0))$ we get $h_1 - \gL_{-\frac{1}{2}}(0)h_0 \in
H^{1/2}(\partial\gO)$.
By Corollary \ref{cor7.6} there is a $f \in \dom(A_{\rm max})$ such
that $h_1 = G_1f$ and $h_0 = G_0f$. Since $G_1f = h_1 = Kh_0 = KG_0f$
we get $f \in \dom(A_{\rm max})$. Hence $f \in \dom(\wh A_K)$ and
$\graph(K) = G\dom(\wh A_K)$.

Conversely, if $\graph(K) = G\dom(\wh A_K)$, then for any $\{h_0,h_1\}
\in \graph(K)$ there is $f\in \dom(\wt A_K)$ such that $h_0= G_0f$ and
$h_1 = G_1f$. By Corollary \ref{cor7.6} we get $h_1 - \gL_{-\frac{1}{2}}(0)h_0 = (K
- \gL_{-\frac{1}{2}}(0))h_0 \in H^{1/2}(\partial\gO)$ which yields $\ran(K
- \gL_{-\frac{1}{2}}(0)) \subseteq H^{1/2}(\partial\gO)$.
\end{proof}

Corollary \ref{VII.9} suggests besides the operator $K: H^{-1/2}(\partial\gO)
\longrightarrow H^{-3/2}(\partial\gO)$  to consider its restriction $K': H^{-1/2}(\partial\gO)
\longrightarrow H^{-3/2}(\partial\gO)$ defined by
\be\label{7.24}
\begin{split}
K' &:= K\upharpoonright\dom(K'), \\
\dom(K') &:= \{h \in \dom(K): Kh - \gL_{-\frac{1}{2}}(0)h \in H^{1/2}(\partial\gO)\} \subseteq \dom(K).
\end{split}
\ee
Clearly, $\graph(K') = G\dom(\wh A_K)$, i.e.
$\wh A_K = \wh A_{K'}$.
For instance, if  $\bO: H^{-1/2}(\partial\gO) \longrightarrow
H^{-3/2}(\partial\gO)$ is  the zero operator, then obviously, $\wh A_\bO = \wh A_{G_1}$. However,
$\bO' := \bO \upharpoonright\dom(\bO')$,
 \be\label{7.24A}
\dom(\bO') := \{f \in
H^{-1/2}(\partial\gO): -\gL_{-\frac{1}{2}}(0)f \in H^{1/2}(\partial\gO)\} = H^{3/2}(\partial\gO).
  \ee
Hence
$\wh A_{G_1} = \wh A_{\bO'}$ and $\dom(\wh A_{G_1}) = \{f \in
H^2(\gO): G_1f = 0\}$.

Obviously, the proper extension $\wh A_K$ admits a description with respect
to the boundary triplet $\Pi = \{L^2(\partial\gO),\gG_0,\gG_1\}$ given
by Proposition \ref{prop5.1}
\begin{proposition}[{\cite[Proposition 3.8]{Ma2010}}]\label{prop3.5}
Assume the conditions of Proposition \ref{prop5.1}.
Let  $K: H^{-1/2}(\partial\gO) \longrightarrow  H^{-3/2}(\partial\gO)$
and let  $\Pi=\{L^2(\partial\gO),
\gG_0,\, \gG_1\}$ be the boundary triplet for $A^*$ given by \eqref{3.10A}.
Then the following holds:

\item[\;\;\rm (i)]
$\wh A_K = A_{B_K}$, where $A_{B_K} :=
A^*\upharpoonright\ker(\gG_1 - B_K\gG_1)$ and
\be\label{3.24}
B_K := (-\Delta_{\partial\gO} + I)^{1/4}(K'
-\gL_{-\frac{1}{2}}(0))(-\Delta_{\partial\gO} + I)^{1/4}: L^2(\partial\gO)
\longrightarrow L^2(\partial\gO).
\ee
\item[\;\;\rm (ii)]
The operator $\wh A_K$ is closed and disjoint from $\wh A_{G_0}$ if and only if the
operator $K' - \gL_{-\frac{1}{2}}(0): H^{-1/2}(\partial\gO) \longrightarrow H^{1/2}(\partial\gO)$
is closed.

\item[\;\;\rm (iii)]
Let $\wh A_K$ be not closed (but necessary closable). Its closure is disjoint from
$\wh A_{G_0}$ if and only if  the operator $K' - \gL_{-\frac{1}{2}}(0)$ is closable.

\item[\;\;\rm (iv)]
If $z \in \rho(\wh A_{G_0})$, then $z \in \rho(\wh A_K)$ if and only
if  the operator $K' -\gL_{-\frac{1}{2}}(z)$ maps
$\dom(K') \subset H^{-1/2}(\partial\gO)$  onto $H^{1/2}(\partial\Omega)$ and its  kernel is trivial.
\end{proposition}
\begin{proof}
(i) From $\wh A_K = \wh A_{K'}$ and \eqref{7.24} we find that
\bed
G_1f - \gL_{-\frac{1}{2}}(0)G_0f = (K' - \gL_{-\frac{1}{2}}(0))G_0f, \quad f \in \dom(\wh A_K),
\eed
which yields
\bed
\begin{split}
\gG_1 f &= (-\gD_{\partial\gO} + I)^{1/2}(K'
-\gL_{-\frac{1}{2}}(0))(-\gD_{\partial\gO} + I)^{1/2}(-\gD_{\partial\gO} + I)^{-1/2}G_0f\\
&= (-\gD_{\partial\gO} + I)^{1/2}(K' -\gL_{-\frac{1}{2}}(0))(-\gD_{\partial\gO} + I)^{1/2}\gG_0f, \quad f \in \dom(\wh A_K).
\end{split}
\eed
Hence, if $f \in \dom(\wh A_K)$, then $f \in \ker(\gG_1 -
B_K\gG_0)$. Therefore $\wh A_K \subseteq A_{B_K}$.

Conversely, if $f \in \ker(\gG_1 -B_K\gG_0)$, then
\bed
(-\gD_{\partial\gO} + I)^{1/2}G_1f =(-\gD_{\partial\gO} + I)^{1/2}(K' -
\gL_{-\frac{1}{2}}(0))G_0f, \quad f \in \dom(A_{B_K}),
\eed
which implies $G_1f = (K' - \gL_{-\frac{1}{2}}(0))G_0f$, $f \in \dom(B_K)$. Hence
$\dom(A_{B_K}) \subseteq \dom(\wh A_K)$. Consequently, $\dom(\wh A_K) =
\dom(A_{B_K})$ or $\wh A_K = A_{B_K}$.

(ii) $\wh A_K$ is closed and disjoint from $\wh A_{G_0}$ if and only if $B_K$ is closed. However, $B_K$ is
closed if and only if $K' - \gL_{-\frac{1}{2}}(0): H^{-1/2}(\partial\gO)
\longrightarrow H^{1/2}(\partial\gO)$ is closed.

(iii)
The closure of $\wh{A_K}$ is disjoint from $\wh A_K$ if and only if $B_K$ is closable. However,
$B_K$ is closable if and only if $K' - \gL_{-\frac{1}{2}}(0)$ is closable.

(iv) By Proposition \ref{t1.12} we get that $z \in \rho(\wh A_K)$ if
and only if $z \in \rho(B_K - M(z))$ where $M(z)$ is the Weyl function
given by \eqref{7.29}. Obviously, we have
\be\label{7.25}
B_K - M(z) = (-\gD_{\partial\gO} + I)^{1/4}(K' - \gL_{-\frac{1}{4}}(z))(-\gD_{\partial\gO} + I)^{1/4}.
\ee
However, the operator $B_K - M(z)$ is invertible if and only if the
operator $K' - \gL_{-\frac{1}{2}}(z): H^{-1/2}(\partial\gO)
\longrightarrow H^{1/2}(\partial\gO)$ is invertible.
\end{proof}
\subsubsection{Perturbation determinants}

To state the next  result  we recall  the following definition.
\bd\label{def4.1}
Let $\cS_p(\gH)=\{T\in \gotS_\infty(\gH): s_j(T)
=O(j^{-1/p}), \; \text{as} \;j\to \infty\}$, $p>0$, where $s_j(T)$,
$j\in\N$, denote the singular values of $T$ (i.e., the
eigenvalues of $(T^*T)^{1/2}$ decreasingly ordered
counting multiplicity).
\ed

It is known that $\cS_p(\gH)$ is
a two-sided (non-closed) ideal in $[\gH]$.  Clearly,
$\mathcal S_{p_1}\subset \cS_{p_2}$ if $p_1 > p_2$.   An important property
of  the classes $\cS_p(\gH)$ needed in the sequel is
\be\label{4.1A}
\mathcal S_{p_1}\cdot \mathcal S_{p_2} \subset \mathcal S_{p}, \quad \text{where} \quad
p^{-1} = p_1^{-1} + p_2^{-1}.
\ee
\bt[{\cite[Theorem 4.13]{Ma2010}}]\label{VII.11}
Assume the Hypothesis \ref{h2.5}. Let $A_0 := \wh A_{G_0}$ and
$0 \in \rho\big({\widehat A}_{G_0}\big)$. Further, let
$K: H^{-1/2}(\partial\gO) \longrightarrow  H^{-3/2}(\partial\gO)$ be an
operator satisfying  $\dom(K)\subseteq L^2(\partial\gO)$ and
$\ran(K)\subseteq L^2(\partial\gO)$.   Then
\be\label{3.90AA}
(\wh A_{K} - z)^{-1} - (A_0 - z)^{-1}\in \mathcal S_{\frac{2n-2}{3}}(L^2(\Omega)), \quad z\in\rho(\wh A_{K})\cap\rho(A_0).
\ee
%
%
%
%
For $n = 2$ the resolvent difference in \eqref{3.90AA} is the trace class operator.
\et

\bl\la{VII.13}
Let $\gotX$ and $\gotY$ be Banach spaces and let $X: \gotX
\longrightarrow \gotY$ be a closed operator which is boundedly invertible.
Further, let $\gotX_0$ be another Banach space such that $\gotX_0$
is a  dense subset of $\gotX$ and  the embedding  $J : \gotX_0 \longrightarrow
\gotX$ is continuous. If $\dom(X) \subseteq J\gotX_0$, then the operator
$X_0 := XJ: \gotX_0 \longrightarrow \gotY$, $\dom(X_0) := \{f \in
\gotX_0: Jf \in \dom(X)\}$ is well defined, closed and
boundedly invertible. In particular,  $X^{-1} = JX^{-1}_0$.
\el
\begin{proof}
Let $\dom(X_0) \ni f_n \longrightarrow f$ and $A_0f_n \longrightarrow g
\in \gotY$ as $n \to \infty$. Obviously we have $Jf_n \longrightarrow
Jf$ and $XJf_n \longrightarrow g$  as $n \to \infty$. Since $X$ is
closed we get $Jf \in \dom(X)$ and $XJf = g$. Hence $f \in \dom(X_0)$
which shows that $X_0$ is closed. If $f \in \ker(X_0)$, then $XJf =
0$. hence $Jf = 0$ which yields $\ker(X_0) = \{0\}$. Finally, we have
$\ran(X) = \ran(X_0) = \gotY$ which shows that $X_0$ is boundedly
invertible.
\end{proof}
In what follows we apply Lemma  \ref{VII.13} with  $\gotX := H^{-1/2}(\partial\gO)$, $\gotX_0 :=
H^0(\partial\gO)$, $\gotY := H^{1/2}(\partial\gO)$ and $\gotY' := H^{-3/2}(\partial\gO)$.
Denote  by $J$  the embedding operator,
  \be\label{7.29A}
J:\ H^0(\partial\gO) = L^2(\partial\gO) \longrightarrow H^{-1/2}(\partial\gO),\qquad Jf=f.
 \ee
%
%
Since $\dom(K) \subseteq JL^2(\partial\gO)$, we can set
\bed
\begin{split}
K_0 &:= KJ: H^0(\partial\gO) \longrightarrow H^{-3/2}(\partial\gO), \\
\dom(K_0) &:= \{f\in H^0(\partial\gO): Jf \in \dom(K)\}.
\end{split}
\eed
Clearly,   $\gL_0(z) = \gL_{-\frac{1}{2}}(z)J$,\ \  $\dom(\gL_0(z)) := JH^0(\partial\gO)$, and
%
%
\bed
\begin{split}
K'_0 &:= K_0\upharpoonright\dom(K'_0),  \\
\dom(K'_0) &:= \{f \in \dom(K_0):\ (K_0 - \gL_0(0))f \in
H^{1/2}(\partial\gO)\}.
\end{split}
\eed
Clearly, $K'_0 = K'J: L^2(\partial\gO) \longrightarrow
H^{-3/2}(\partial\gO)$.

Now we are in position to state the  first main result of this section.
\begin{proposition}\la{VII.12}
Let the assumptions of Theorem  \ref{VII.11} be satisfied.
Further, let $0 \in \rho(A_0) \cap \rho(\wh A_K)$.  Then the following holds:

\item[\;\;\rm (i)]  For any  $z\in\rho(\wh A_{K})\cap\rho(A_0)$ the operator $K' - \gL_{-\frac{1}{2}}(z): H^{-1/2}(\partial\gO)
\longrightarrow H^{1/2}(\partial\gO)$ is boundedly invertible and
\be\la{7.26}
(K' - \gL_{-\frac{1}{2}}(z))^{-1} \in
\mathcal S_{\frac{2n-2}{3}}(H^{1/2}(\partial\gO),H^{-1/2}(\partial\gO)).
\ee
In particular, if $n=2$ then
\begin{equation}\label{7.28A}
 (K' - \gL_{-\frac{1}{2}}(z))^{-1} \in
\gotS_1(H^{1/2}(\partial\gO),H^{-1/2}(\partial\gO)).
\end{equation}
\item[\;\;\rm (ii)]  Let  $n =2$. Then the boundary triplet $\wt \Pi= \{\cH, -\gG_1,\gG_0\}$, where $\cH:= L^2(\partial\gO)$ and $\gG_0,\gG_1$  are given by  \eqref{3.10A},  is  regular  for the pair $\{\wt A,A_0\}$,
    $\{\wh A_K,A_0\} \in \gotD^{\wt \Pi}$, and the perturbation determinant  $\gD^{\wt \Pi}_{\wh A_K/\wh A_0}(\cdot)$ is
\be\label{7.30A}
\begin{split}
\gD^{\wt \Pi}_{\wh A_K/\wh A_0}(z) &= 
{\det}_{H^{-\frac{1}{2}}}\left(I - \bigl(K' - \gL_{-\frac{1}{2}}(0)\bigr)^{-1}\bigl(\gL_{-\frac{1}{2}}(z)-\gL_{-\frac{1}{2}}(0)\bigr)\right)\\
&= {\det}_{H^{\frac{1}{2}}}\left(I - \bigl(\gL_{-\frac{1}{2}}(z) - \gL_{-\frac{1}{2}}(0)\bigr)\bigl(K' - \gL_{-\frac{1}{2}}(0)\bigr)^{-1}\right)
\end{split}
 \ee
where $z\in \rho(\wh A_K)\cap \rho(A_0)$.

\item[\;\;\rm (iii)]  Let $n=2$. Then $(\gL_0(z)-\gL_0(0))(K'_0 - \gL_0(0))^{-1}\in \gotS_1(H^{1/2}(\partial\Omega))$ and the perturbation determinant
$\gD^{\wt \Pi}_{\wh A_K/A_0}(\cdot)$ admits the representation
\be\la{7.30}
\gD^{\wt \Pi}_{\wh A_K/A_0}(z) =  {\det}_{H^{1/2}}\left(I - (\gL_0(z)-\gL_0(0))(K'_0 - \gL_0(0))^{-1}\right).
\ee
\end{proposition}
\begin{proof}
(i) By Proposition \ref{t1.12}(i),  $0 \in \rho(B_K - M(z))$ for any $z\in \rho(\wh A_K)\cap \rho(A_0)$.  Moreover, combining  Proposition \ref{prop2.9} with Theorem \ref{VII.11} we get  $(B_K - M(z))^{-1} \in \mathcal S_{\frac{2n-2}{3}}(H^0(\partial\Omega))$.
Combining this fact  with \eqref{7.25}  we obtain
\bead
\lefteqn{
(B_K - M(z))^{-1}}\\
& &
 = (I-\gD_{\partial\gO})^{-1/4}(K' - \gL_{-\frac{1}{2}}(z))^{-1}(I - \gD_{\partial\gO})^{-1/4} \in \mathcal S_{\frac{2n-2}{3}}(H^0(\partial\Omega)).
\eead
Since $(I - \gD_{\partial\gO})^{-1/4}$ isomorphically maps
$H^s(\partial\gO)$
onto $H^{s+1/2}(\partial\gO)$ for $s \in \R$ we arrive at \eqref{7.26}.
Further,  for $n=2$ inclusion \eqref{7.26} implies
$$
(K' - \gL_{-\frac{1}{2}}(z))^{-1} \in
\mathcal S_{\frac{2}{3}}(H^{1/2}(\partial\gO),H^{-1/2}(\partial\gO))\subset \gotS_1(H^{1/2}(\partial\gO),H^{-1/2}(\partial\gO)).
$$
This proves  the last statement.

(ii) By  Theorem \ref{VII.11},  $(\wh A_K - z)^{-1} - (A_0 -
z)^{-1} \in \gotS_1(L^2(\gO))$ since  $n=2$. Further, by Proposition \ref{t1.12}(i), the condition $0 \in \rho(A_0) \cap
\rho(\wh  A_K)$ is equivalent to   $0 \in \rho(B_K - M(0)) = \rho(B_K)$. By
 Proposition \ref{IV.8}(iii), a boundary triplet  $\wt \Pi= \{\cH, -\gG_1,\gG_0\}$ is regular for the pair $\{\wh A_K,A_0\}$, hence
   $\{\wh A_K,A_0\} \in \gotD^{\wt \Pi}$, and the perturbation determinant $\gD^{\wt \Pi}_{\wh A_K/A_0}(\cdot)$ is given by
\eqref{4.11} with $\mu=0$ and  constant $c=1$ (see formula  \eqref{4.11New}).  Inserting in this expression  formulas \eqref{7.29}
and \eqref{3.24}  we arrive at  the following formula for the determinant
\bed
\begin{split}
&\gD^{\wt \Pi}_{\wh A_K/A_0}(z) = {\det}_{L^2(\partial\gO)}(I - B^{-1}_KM(z)) =\\
&{\det}_{H^0}\left(I -
(I -\Delta_{\partial\gO})^{-1/4}(K' - \Lambda_{-\frac{1}{2}}(0))^{-1}
(\Lambda_{-\frac{1}{2}}(z) - \Lambda_{-\frac{1}{2}}(0))(I -\Delta_{\partial\gO})^{1/4}\right)
\end{split}
\eed
for  $z \in \rho(\wh A_K) \cap \rho(A_0)$.  Further,
according to Proposition \ref{prop5.1}(ii),  $\gL_{-\frac{1}{2}}(z) - \gL_{-\frac{1}{2}}(0)\in [H^{-1/2}(\partial\gO), H^{1/2}(\partial\gO)]$.
Combining this inclusion  with \eqref{7.28A} we get $T_2(z) \in \gotS_1(H^0(\partial\gO), H^{-1/2}(\partial\gO))$ where
\bed
T_2(z):= (K' - \Lambda_{-\frac{1}{2}}(0))^{-1}
(\Lambda_{-\frac{1}{2}}(z) - \Lambda_{-\frac{1}{2}}(0))(I -\Delta_{\partial\gO})^{1/4}, \quad z \in \rho(\wh A_K) \cap \rho(A_0).
\eed
Noting that $T_1 =  (I -\Delta_{\partial\gO})^{-1/4}$ isomorphically maps $H^{-1/2}(\partial\gO)$ onto $H^0(\partial\gO)$ we see that $T_2(z)T_1$ is well defined and $T_2(z)T_1\in \gotS_1(H^{-1/2}(\partial\gO))$. Moreover, due to the inclusion \eqref{7.28A}, $T_1T_2(z)\in \gotS_1(H^{0}(\partial\gO))$. Taking both last inclusions into account
and applying \eqref{2.26} we arrive at
the equality
$$
\gD^{\wt \Pi}_{\wh A_K/A_0}(z) =  {\det}_{L^2(\partial\gO)}\bigl(I- T_1T_2(z)\bigr) = c\;{\det}_{H^{-1/2}(\partial\gO)}\bigl(I - T_2(z)T_1\bigr)
$$
coinciding with  the first identity in  \eqref{7.30A}.
The second identity in   \eqref{7.30A}  is implied by combining the first one with  the property \eqref{2.26}. Note that the applicability of  \eqref{2.26} is possible due to
inclusion \eqref{7.28A} and Proposition \ref{prop5.1}(ii).

(iii) By Lemma \ref{VII.13}, the operator
$(K'_0 - \gL_0(0))^{-1}: H^{1/2}(\partial\gO) \longrightarrow
L^2(\partial\gO)$ is bounded and
\be\la{7.28a}
(K' - \gL_{-\frac{1}{2}}(0))^{-1} = J(K'_0 - \gL_0(0))^{-1}. \qquad  z \in \rho(\wh A_K) \cap \rho(A_0).
\ee
Combining this formula with \eqref{7.28A} we get $J(K'_0 - \gL_0(0))^{-1}\in \gotS_1(H^{1/2}(\partial\gO)).$
Therefore inserting \eqref{7.28a} into the second formula in  \eqref{7.30A} we get
\bed
\gD^{\wt \Pi}_{\wh A_K/A_0}(z) = c\;{\det}_{H^{1/2}}\left(I -
\left(\gL_{-\frac{1}{2}}(z) - \gL_{-\frac{1}{2}}(0)\right)J
\left(K'_0 - \gL_0(0)\right)^{-1}\right).
\eed
To  arrive at  \eqref{7.30} it remains to note that  $\gL_{-\frac{1}{2}}(z)J = \gL_0(z)$.
\end{proof}

Combining the chain rule \ref{3.15}  with Proposition \ref{VII.12} one  arrives at the following statement.
\bc\label{cor7.15a}
Assume the Hypothesis \ref{h2.5}. Let $A_0 := \wh A_{G_0}$ and
$0 \in \rho\big({\widehat A}_{G_0}\big)$.
Further, let
$K_j: H^{-1/2}(\partial\gO) \longrightarrow  H^{-3/2}(\partial\gO)$ be an
operator satisfying  $\dom(K_j)\subseteq L^2(\partial\gO)$ and
$\ran(K_j)\subseteq L^2(\partial\gO)$,  and let $A_j := A_{K_j}$,  $j\in \{1,2\}$.
Assume also that  $0 \in \rho(A_0) \cap
\rho(\wh A_{K_j})$, \ $j\in \{1,2\}$.
Then the boundary triplet
$\wt \Pi= \{\cH, -\gG_1,\gG_0\}$ for $A_{\rm max}$ with  $\cH:= L^2(\partial\gO)$ and $\gG_0,\gG_1$ given by  \eqref{3.10A},  is  regular  for the family  $\{\wh A_{K_1},\wh A_{K_2}, A_0\}$,
and the perturbation determinant is
\be\la{7.36a}
\gD^{\Pi}_{\wh A_2/\wh A_1}(z) =
\frac{ {\det}_{H^{-\frac{1}{2}}}\left(I - \bigl(K'_2 - \gL_{-\frac{1}{2}}(0)\bigr)^{-1}\bigl(\gL_{-\frac{1}{2}}(z)-\gL_{-\frac{1}{2}}(0)\bigr)\right)}
{ {\det}_{H^{-\frac{1}{2}}}\left(I - \bigl(K'_1 - \gL_{-\frac{1}{2}}(0)\bigr)^{-1}\bigl(\gL_{-\frac{1}{2}}(z)-\gL_{-\frac{1}{2}}(0)\bigr)\right)},
\ee
$z\in \rho(\wh A_{K_1})\cap \rho(\wh A_{K_2})\cap \rho(A_0)$.
\end{corollary}
Our next goal is to show that under  additional restrictions on $K$ the perturbation determinant  $\gD^{\wt \Pi}_{\wh A_K/A_0}(\cdot)$ can be computed in $L^2(\partial\gO)$. To this end we introduce the  operator-valued function
$\gL_{0,0}(\cdot): L^2(\partial\gO) \longrightarrow L^2(\partial\gO)$ by setting
\be\label{7.33}
\begin{split}
\gL_{0,0}(z) &:= \gL_0(z)\upharpoonright\dom(\gL_{0,0}(z)),\\
\dom(\gL_{0,0}(z)) &:= \{f \in \dom(\gL_0(z)): \gL_0(z)f \in L^2(\partial\gO)\} 
\end{split}
\ee
\bl\label{lem7.16}
Let $0 \in \rho(\wh A_{G_0})$. Then
\be\label{7.33a}
\dom(\gL_{0,0}(z)) =  H^1(\partial\gO), \qquad z\in \rho(\wh A_{G_0}),
\ee
and,  for any $z\in \rho(\wh A_{G_0}) \cap \rho(\wh A_{G_1})$ the operator $(\gL_{0,0}(z))^{-1}$ exits and  satisfies  $(\gL_{0,0}(z))^{-1} \in \mathcal S_{{1}}(H^{0}(\partial\gO))$.
Moreover, if $0 \in \rho(\wh A_{G_0}) \cap \rho(\wh A_{G_1})$, then the operator $\gL_{0,0}(0)$ is selfadjoint,
has discrete spectrum,  and $(\gL_{0,0}(0))^{-1} \in \mathcal S_{{1}}(H^{0}(\partial\gO))$.
\el
\begin{proof}
It follows from Definition \ref{def3.2A} that
$\dom(\gL_{0,0}(\cdot)) \supseteq H^1(\partial\gO)$.
Let us prove the equality \eqref{7.33a}.
Since both  realizations $\wh A_{G_0}$ and $\wh A_{G_1}$ are selfadjoint, $\rho(\wh A_{G_0})\cap \rho(\wh A_{G_1})\supset \C_{\pm}$. Let $z\in \rho(\wh A_{G_0})\cap \rho(\wh A_{G_1})$.
Then $\dom(\gL_{1}(z)) = H^1(\partial\gO)$ and
$\gL_{1}$  isomorphically maps $H^1(\partial\gO)$ onto  $H^0(\partial\gO)$ (see \cite[Theorem 5.2]{Gru68}).
Since $\gL_{0,0}(z)h = \gL_{1}(z)h$ for  $h\in H^1(\partial\gO)$ we conclude
that $\dom(\gL_{0,0}(z)) = H^1(\partial\gO)$
and   $\ran(\gL_{0,0}(z)) = H^0(\partial\gO)$.

Next let  $x_0= \bar x_0\in \rho(\wh A_{G_0}) \setminus  \rho(\wh A_{G_1})$.
We cam assume without loss of generality that $x_0=0.$ Otherwise we
replace the expression  ${\mathcal A}$ by ${\mathcal A} - x_0I$.
Then,  by Proposition  \ref{prop5.1},
the difference $T(z) := \gL_{-\frac{1}{2}}(z) - \gL_{-\frac{1}{2}}(0) \colon\ H^{-1/2}(\partial\Omega) \to H^{ 1/2}(\partial\Omega)$ is bounded.  Hence the difference $\gL_{0,0}(z) - \gL_{0,0}(0)$ being a restriction of $T(z)$ is bounded in $H^{0}(\partial\Omega)$ and $\dom(\gL_{0,0}(0)) = \dom(\gL_{0,0}(z)) = H^1(\partial\gO)$ for $z \in \rho(\wh A_{G_0})$.
Further, since
$\ran(\gL_{0,0}(z))^{-1}) =  H^{1}(\partial\gO)$  for $z \in \rho(\wh A_{G_0})\cap \rho(\wh A_{G_1})$, we have  $(\gL_{0,0}(z))^{-1} \in \mathcal S_{{1}}(H^{0}(\partial\gO))$.

Clearly, for any $x_0 = \bar x_0\in \rho(\wh A_{G_0})$ the operator $\gL_{0,0}(x_0)$ is symmetric.
If, in addition, $x_0\in \rho(\wh A_{G_1})$, then the operator $\gL_{0,0}(x_0)$ is selfadjoint  since $\ran(\gL_{0,0}(x_0)) =
H^{0}(\partial\gO)$.  
If $0 \in \rho(\wh A_{G_0}) \setminus  \rho(\wh A_{G_1})$, then
the self-adjointness of  $\gL_{0,0}(0)$ is implied by the self-adjointness  of  $\gL_{0,0}(x_0)$
with   $x_0 =  \bar x_0\in \rho(\wh A_{G_0}) \cap  \rho(\wh A_{G_1})$ and  the boundedness of
$\gL_{0,0}(x_0) - \gL_{0,0}(0)$ in $H^{0}(\partial\gO)$.

Further, since the boundary  $\partial\gO$ is compact, the spectrum of $\gL_{0,0}(0)$
is discrete. Moreover, since  by \eqref{7.33a},
$\ran(\gL_{0,0}(z))^{-1}) =  H^{1}(\partial\gO))$  for $\rho(\wh A_{G_0})\cap \rho(\wh A_{G_1})$, we have  $(\gL_{0,0}(z))^{-1} \in \mathcal S_{{1}}(H^{0}(\partial\gO))$.
\end{proof}
\begin{proposition}\la{VII.15}
Assume the Hypothesis \ref{h2.5}.
Let $K: H^{-1/2}(\partial\gO) \longrightarrow  H^{-3/2}(\partial\gO)$ be an
operator satisfying  $\dom(K)\subseteq L^2(\partial\gO)$ and
$\ran(K)\subseteq L^2(\partial\gO)$.
Assume also that  $0 \in \rho(\wh{A_{G_0}}) \cap  \rho(\wh{A_{G_1}}) \cap \rho(\wh A_K)$ and
\be\la{7.31a}
\widehat K_0 := KJ:  L^2(\partial\gO) \longrightarrow L^2(\partial\gO), \quad \dom(K) = J\dom(\widehat K_0),
\ee
where $J$ is the embedding operator given by \eqref{7.29A}.
If $\wh K_0$ is relatively compact with respect to $\gL_{0,0}(0)$, then
\be\label{7.37}
\bigl(\gL_{0,0}(z) -  \gL_{0,0}(0))(\wh K_0 - \gL_{0,0}(0)\bigr)^{-1}   \in \mathcal S_{\frac{1}{2}}(H^{0}(\partial\gO))\subset \gotS_1(H^{0}(\partial\gO)),
\ee
$z \in \rho(\wh A_K) \cap  \rho(\wh A_{G_0})$, and the perturbation determinant $\gD^{\wt \Pi}_{\wh A_K/A_0}(\cdot)$  given by \eqref{7.30A}
admits the representation
\be\la{7.32}
\gD^{\wt \Pi}_{\wh A_K/A_0}(z) =  {\det}_{L^2}\left(I - \bigl(\gL_{0,0}(z) -  \gL_{0,0}(0))(\wh{K_0} - \gL_{0,0}(0)\bigr)^{-1}\right),
\ee
for $z \in \rho(\wh A_K) \cap  \rho(\wh A_{G_0})$. In particular,  representation \eqref{7.32} holds
whenever $\widehat K_0$ is bounded, i.e. $\widehat K_0 \in [H^0(\partial\Omega)]$.
\end{proposition}
\begin{proof}
(i)  Let us prove the inclusion  \eqref{7.37}.
According to \eqref{3.5AWeyl}, $\gL_0(z) - \gL_0(0): H^{0}(\partial\gO) \longrightarrow  H^{1/2}(\partial\gO)$. Hence
\be\label{7.38}
\overline{\gL_0(z)-\gL_0(0)} \in \cS_{{2}}(H^{0}(\partial\gO)),\qquad z\in \rho(A_0).
\ee
Further, by definition  \eqref{7.29A},  $J^*$ continuously embeds  $H^{1/2}(\partial\gO)$ into $L^2(\partial\gO)$.
Therefore
\be\label{7.35}
J^*(\gL_0(z) - \gL_0(0))h = (\gL_{0,0}(z) - \gL_{0,0}(0))h, \quad h \in H^1(\partial\gO).
\ee
Combining relations \eqref{7.35} and \eqref{7.38}, using $J^*\in S_{{2}}(H^{1/2}, H^{0})$ and taking property
\eqref{4.1A} into account we obtain
\be\label{7.38A}
\overline{\gL_{0,0}(z) - \gL_{0,0}(0)} \in \cS_{{1}}(H^{0}(\partial\gO)),\qquad z\in \rho(A_0).
\ee
On the other hand, by Lemma \ref{lem7.16} the operator $\gL_{0,0}(0)$ is selfadjoint, and due to the assumption $0\in \rho(\wh{A_{G_1}})$,  $\gL_{0,0}(0)$ is invertible
and $(\gL_{0,0}(0))^{-1} \in \cS_{{1}}(H^{0}(\partial\gO))$.

Further,  by Proposition \ref{prop3.5}(iv) the inclusion $0\in \rho(\wh{A_K})$ is equivalent
to the inclusion  $0\in \rho({\wh K}_0 - \gL_{0,0}(0))$. Hence
$0\in \rho\bigl(I- \wh K_0(\gL_{0,0}(0))^{-1}\bigr)$.
Thus,  the inverse operator
$\bigl(I - \wh K_0(\gL_{0,0}(0))^{-1}\bigr)^{-1}\in [H^0(\partial\gO)]$ and
\bed
\bigl(\widehat K_0 - \gL_{0,0}(0)\bigr)^{-1}
= - \left(\gL_{0,0}(0)\right)^{-1}\left(I - \widehat K_0(\gL_{0,0}(0))^{-1}\right)^{-1} \in \mathcal S_{{1}}(H^{0}(\partial\gO)).
\eed
Using \eqref{7.38A} and taking  into account \eqref{4.1A} we arrive at  \eqref{7.37}.

(ii) In this step we prove formula \eqref{7.32}.
Since  $\ran(K) \subseteq L^2(\partial\gO)$, the operator  $K_0 =KJ: H^0(\partial\gO) \longrightarrow H^{-3/2}(\partial\gO)$ satisfies   $\ran(K_0) \subseteq L^2(\partial\gO)$.
However, we distinguish it from the operator  $\widehat K_0$ defined by \eqref{7.31a}.
Since $\dom(\widehat K_0) = \dom(K_0)$  and  $\ran(\widehat K_0) \subseteq L^2(\partial\gO)$, one
gets  from definition \eqref{7.24} and \eqref{7.33}
\bed
\dom(K'_0) := \{h\in\dom(\widehat K_0) \cap H^1(\partial\gO): \widehat K_0h - \gL_{0,0}(0)h \in H^{1/2}(\partial\gO)\},
\eed
and  in  accordance with \eqref{7.29A}
\be\label{7.36}
J^*(K'_0 - \gL_0(0))h = (\widehat K_0 - \gL_{0,0}(0))h, \qquad h \in \dom(K'_0).
\ee
Clearly, $J^*(K'_0 - \gL_0(0)) \subseteq  \widehat K_0 - \gL_{0,0}(0)$ (in fact, the inclusion is always strict).\marginpar{Why?}
Since  $\widehat K_0$ is relatively compact with respect to $\gL_{0,0}(0)$,
 $\dom(\widehat K_0 - \gL_{0,0}(0)) = \dom(\gL_{0,0}(0))$.
Moreover, since $0\in \rho(\widehat A_{K_0})$ and $\dom(K')\subseteq H^0(\partial\gO)$, Proposition \ref{prop3.5}(iv) yields  $\ran(K'_0 - \gL_{0}(0)) =\ran(K' - \gL_{-\frac{1}{2}}(0)) = H^{1/2}(\partial\gO)$.
Hence the range $\ran(\widehat K_0 - \gL_{0,0}(0))$ is dense in $H^0(\partial\gO)$.

On the other hand, by Lemma \ref{lem7.16},  the operator $\gL_{0,0}(0)$ is selfadjoint and   its spectrum is discrete.
In particular, $\gL_{0,0}(0)$ is a Fredholm operator with zero index.
In turn,  since  $\widehat K_0$  is $\gL_{0,0}(0)$-compact, the operator $\widehat K_0 -  \gL_{0,0}(0)$ is also  Fredholm operator with zero index  \cite[Theorem 4.5.26]{Ka76} (in fact, it has discrete spectrum too).  Therefore the range
$\ran(\widehat K_0 - \gL_{0,0}(0))$ is closed and being dense in $H^0(\partial\gO)$, coincides with $H^0(\partial\gO)$. Since ${\rm ind}(\widehat K_0 - \gL_{0,0}(0)) =0,$ the latter is equivalent to
$0\in \rho(\widehat K_0 - \gL_{0,0}(0))$.
From Lemma \ref{VII.13} we get the existence of $(K'_0f - \gL_0(0))^{-1}: H^{1/2}(\partial\gO) \longrightarrow H^0(\partial\gO)$. Combining this fact with \eqref{7.36}
we find
\be\la{7.31}
(\widehat K_0 -  \gL_{0,0}(0))^{-1}J^* = (K'_0 - \gL_0(0))^{-1}
\ee
Inserting \eqref{7.31} into \eqref{7.30} we obtain
\bed
\gD^{\wt \Pi}_{\wh A_K/A_0}(z) = c\; {\det}_{H^{1/2}}\left(I - (\gL_0(z)-\gL_0(0))(\wh K_0 - \gL_{0,0}(0))^{-1}J^*\right).
\eed
Using \eqref{7.37}  we get by the cyclicity property (see \eqref{2.26}) that
\bed
\gD^{\wt \Pi}_{\wh A_K/A_0}(z) = c\; {\det}_{H^0}\left(I - J^*(\gL_0(z)-\gL_0(0))(\wh K_0 - \gL_{0,0}(0))^{-1}\right),
\eed
$z \in \rho(\wh A_K) \cap \rho(A_0)$. Combining this identity with  \eqref{7.35} we arrive at \eqref{7.32}.
\end{proof}

\bc\label{cor7.15}
Assume the Hypothesis \ref{h2.5}.
Let
$K_j: H^{-1/2}(\partial\gO) \longrightarrow  H^{-3/2}(\partial\gO)$ be an
operator satisfying  $\dom(K_j)\subseteq L^2(\partial\gO)$ and
$\ran(K_j)\subseteq L^2(\partial\gO)$, $j\in \{1,2\}$.
Further, let $0 \in \rho(\wh{A_{G_0}}) \cap \rho(\wh{A_{G_1}}) \cap
\rho(\wh A_{K_j})$, \  and
\bed
\wh{K_{j,0}}:= K_jJ:  L^2(\partial\gO) \longrightarrow L^2(\partial\gO), \quad \dom(K_j) = J\dom(\wh{K_0}),
\quad j \in \{1,2\}.
\eed
If the operator $\wh{K_{j,0}}$ is relatively compact with respect to $\gL_{0,0}(0)$, then the perturbation determinant
$\gD^{\Pi}_{\wh A_2/\wh A_1}(\cdot)$ given by \eqref{7.36a} admits the representation
\bed
\gD^{\Pi}_{\wh{A_{K_2}}/\wh{A_{K_1}}}(z) =
\frac{{\det}_{L^2}\left(I - \bigl(\gL_{0,0}(z) -  \gL_{0,0}(0))(\wh{K_{2,0}} - \gL_{0,0}(0)\bigr)^{-1}\right)}
{{\det}_{L^2}\left(I - \bigl(\gL_{0,0}(z) -  \gL_{0,0}(0))(\wh{K_{1,0}} - \gL_{0,0}(0)\bigr)^{-1}\right)},
\eed
for $z \in \rho(\wh{A_{K_2}}) \cap  \rho(\wh{A_{K_1}}) \cap \rho(\wh{A_{G_0}})$.
\ec
\begin{proof}
The proof is immediate by combining  Corollary \ref{cor7.15a} with  Proposition \ref{VII.15}.
\end{proof}

Consider  Robin-type realizations
\bed
\begin{split}
\wh A_{\gs} &:= A_{\max}\upharpoonright \dom(\wh A_{\gs}),\\
\dom(\wh A_{\gs}) &:= \{f\in H^2(\gO): G_1f = \gs G_0f\}.
\end{split}
\eed
%
It follows from the classical a priory estimate (see \cite[Theorem 15.2]{Agmon1959})
that the realization $\wh A_{\sigma}$ is closed whenever $\sigma\in C^2(\partial\Omega)$.
Moreover, in this case $\rho(\wh A_{\sigma})\not =\emptyset$ and $\wh A_{\sigma}$ is selfadjoint whenever $\sigma$ is real.
\bc
Assume the conditions  of Proposition \ref{VII.15}.  Let $\gs\in C^2(\partial\gO)$
and  let $\wh \gs$ denote the multiplication operator induced by $\gs$  in $L^2(\partial\gO)$.
If $0 \in \rho(\wh{A_\gs}) \cap \rho(\wh{A_{G_0}}) \cap \rho(\wh{A_{G_1}})$, then the boundary triplet $\wt \Pi= \{\cH, -\gG_1,\gG_0\}$ given in Proposition  \ref{VII.12}, is regular for the  pair $\{\wh A_\gs, \wh{A_{G_0}}\}$,  $\{\wh A_\gs, \wh{A_{G_0}}\} \in \gotD^{\wt \Pi}$, and the corresponding
perturbation determinant $\gD^{\wt \Pi}_{\wh A_\gs/A_0}(\cdot)$ is
\bed
\gD^{\wt \Pi}_{\wh A_\gs/A_0}(z) = 
{\det}_{L^2(\partial\gO)}\left(I - (\gL_{0,0}(z) - \gL_{0,0}(0))(\wh \gs - \gL_{0,0}(0))^{-1}\right), \
\eed
$z \in \rho(\wh A_\gs) \cap \rho(\wh{A_{G_0}})$.
\ec
\begin{proof}
Setting $K=\wh{\sigma}$ and noting that $\gs\in C^2(\gO)$ we easily  get from \eqref{6.12c}
\bed
\dom\bigl(K-\Lambda_{-1/2}(0)\bigr) = \dom(K)=\dom(\wh{\sigma})\subset \ran(G_0) = H^{3/2}(\partial\Omega).
\eed
Since $\Lambda_{3/2}(0)$ is a restriction of  $\Lambda_{-1/2}(0)$  one gets from  \eqref{3.5A} that
$\ran\bigl(\Lambda_{-1/2}(0)\upharpoonright H^{3/2}(\partial\Omega)\bigr)\subset H^{1/2}(\partial\Omega)$.
Further, the assumption  $\gs\in C^2(\gO)$ yields
$\ran\bigl(K \upharpoonright H^{3/2}(\partial\Omega)\bigr)\subset H^{3/2}(\partial\Omega)$.
Combing these inclusions  we arrive at  the regularity property
$\ran\bigl(K-\Lambda_{-1/2}(0)\bigr)\subset H^{1/2}(\partial\Omega)$ (see \eqref{7.23}).

 Hence $K'=K$ (see definition \eqref{7.24}) and $\dom(\wh A_K) =
\dom(\wh A_\gs)$. Moreover, since $\dom(K),\  \ran(K)\subset  H^{3/2}(\partial\Omega)$, then according to \eqref{7.31a},  $\widehat K_0 = \widehat \sigma$.
Finally, since $\widehat K_0 = \widehat \sigma \in [H^{0}(\partial\Omega)]$, one completes the proof by applying Proposition \ref{VII.15}.
\end{proof}

\begin{appendix}

\section*{Appendix}\la{App}

\section{Infinite determinants}\la{A.I}

Let us briefly recall the definition of determinants and their
basic properties following {\cite{GK69}}.
\bd
{\em
Let  $T$ be a trace class operator, i.e. $T\in{\mathfrak S}_1(\cH)$,
and let $\{\lambda_j(T)\}^{\infty}_{j=1}$ be its
eigenvalues counted with respect to their algebraic multiplicity.
The determinant $\det(I+T)$ is defined by
$\det(I+T) := \Pi^{\infty}_{j = 1}\bigl(1+\lambda_j(T)\bigr)$.
}
\ed

The perturbation determinant has the following interesting properties.
\begin{proposition}[{\cite[Section 4.1]{GK69}}]\label{classprop}
Let $T_1\in[\cH_1, \cH_2]$ and $T_2\in[\cH_2, \cH_1]$.
\begin{enumerate}

\item[\rm (i)] If $T_1 T_2\in{\mathfrak S}_1(\cH_2)$ and $T_2
T_1\in{\mathfrak S}_1(\cH_1)$, then
 \be\label{2.26}
{\det}_{\cH_2}(I+T_1T_2) = {\det}_{\cH_1}(I + T_2 T_1).
  \ee
\item[\rm (ii)] If $\cH := \cH_1 =\cH_2$ and $T_1,T_2\in{\mathfrak S}_1(\cH)$, then
\be\label{2.27} \det[(I+T_1)(I+T_2)]=\det(I+T_1)\cdot\det(I+T_2).
\ee

\item[\rm (iii)] If $T\in{\mathfrak S}_1(\cH)$, then
\bed \det(I+T^*)={\overline{\det(I+T)}}. \eed
\end{enumerate}
\end{proposition}

For technical reasons we need a slightly improved version of the
property \eqref{2.26}.
\bl\la{A.1}
Let $T = T^* \ge 0$ such that $0 \in \rho(T)$.
Further, let $C$ be linear operator such that $\dom(C) \supseteq
\ran(T)$. If $\overline{T^{-1}C} \in \gotS_1(\gotH)$ and $CT^{-1} \in
\gotS_1(\gotH)$, then
\be\la{a.1}
\det(I + \overline{T^{-1}C}) = \det(I + CT^{-1}).
\ee
\el
\begin{proof}
Let $K = K^*$ be a bounded operator such that $\|K\| < \inf\gs(T)$
In this case we have $T+ K \ge 0$ and $0 \in  \rho(T+K)$.
Moreover, we find $\overline{(T + K)^{-1}C} \in
\gotS_1(\gotH)$ and $C(T+K)^{-1} \in \gotS_1(\gotH)$. Both relations
immediately follow from the identities
\bed
(T + K)^{-1} = (I - (T+K)^{-1}K)T^{-1} = T^{-1}(I - K(T+K)^{-1})
\eed
Moreover, if $\{K_n\}_{n\in\N}$ is a sequence of selfadjoint bounded
operators such $\sup_n\|K_n\| < \inf\gs(T)$ and
$\lim_{n\to\infty}\|K_n\| = 0$, then
\be\la{a.2}
\lim_{n\to\infty} \det(I + \overline{(T + K_n)^{-1}C}) = \det(I + T^{-1}C)
\ee
and
\be\la{a.3}
\lim_{n\to\infty} \det(I + C(T + K_n)^{-1}) = \det(I + CT^{-1}).
\ee
For any $\varepsilon \in (0,\inf\gs(T))$ there is a Hilbert-Schmidt
operator $K_\varepsilon$ such that $\|K_\varepsilon\|_{\gotS_2}
< \varepsilon$ and $T + K_\varepsilon$ has
pure point spectrum. Let $\gl_n$, $n \in \N$, be an enumeration of the
eigenvalues of $T+K_\varepsilon$ where without loss of generality we may assume
that the eigenvalues are simple. Let $P_m$ the orthogonal projection onto the
subspace which is spanned by the first $m$-eigenspaces. Obviously, we
have $s-\lim_{m\to\infty}P_m = I$ and $(T+K_\varepsilon)P_m = P_m(T+K_\varepsilon)$, $m \in
N$. We find
\bead
\lefteqn{
s-\lim_{n\to\infty}\det(I + (T+K_\varepsilon)^{-1}P_mC) =}\\
& &
s-\lim_{n\to\infty}\det(I + P_mC(T+K_\varepsilon)^{-1}) = \det(I + C(T+K)^{-1}).
\eead
and
\bead
\lefteqn{
s-\lim_{n\to\infty}\det(I + (T+K_\varepsilon)^{-1}P_mC) = }\\
& &
s-\lim_{n\to\infty}\det(I + P_m\overline{(T+K_\varepsilon)^{-1}C}) =
\det(I + \overline{(T+K_\varepsilon)^{-1}C}
\eead
which yields
\bed
\det(I + \overline{(T+K_\varepsilon)^{-1}C}) = \det(I + C(T+K_\varepsilon)^{-1}).
\eed
Choosing a sequence $\{\varepsilon_n\}_{n\in\N}$, $\varepsilon_n > 0$,
which tends to zero as $n\to\infty$ we get
\bed
\det(I + \overline{(T+K_{\varepsilon_n})^{-1}C}) = \det(I + C(T+K_{\varepsilon_n})^{-1}).
\eed
for each $n \in \N$. Taking into account \eqref{a.2} and \eqref{a.3}
we verify \eqref{a.1}.
\end{proof}
\bc\la{A.2}
Let $T$ be a densely defined closed symmetric operator such that $0
\in \rho(T)$. Further, let $C$ be a linear operator such that $\dom(C)
\supseteq \ran(T)$. If $\overline{T^{-1}C} \in \gotS_1(\gotH)$ and
$CT^{-1} \in \gotS_1(\gotH)$, then
\be\la{a.4}
\det(I + \overline{T^{-1}C}) = \det(I + CT^{-1}).
\ee
\ec
\begin{proof}
By the polar decomposition we have $T* = U|T^*|$ where $U$ is
unitary. Hence, $T^{-1} = U|T^*|^{-1}$. Since $U^*T^{-1}C = |T^*|^{-1}C$
one has $\overline{|T^*|^{-1}C} \in \gotS_1(\gotH)$. From
$CT^{-1} = \wt C |T^*|^{-1} \in \gotS_1(\gotH)$ where $\wt C = CU$.
Obviously, we find
\bed
\det(I + T^{-1}C) = \det(I + U|T^*|^{-1}C) = \det(I + |T^*|^{-1}\wt C)
\eed
and
\bed
\det(I + CT^{-1}) = \det(I + \wt C|T^*|^{-1}).
\eed
Applying Lemma \ref{A.1} we get $\det(I + |T^*|^{-1}\wt C) = \det(I + \wt C|T^*|^{-1})$
which yields \eqref{a.4}
\end{proof}

\section{Perturbation determinant and their properties}\la{B}

Let us summarize some important properties of perturbation determinants for additive
perturbations $\gD_{H'/H}(\cdot)$, cf. \cite[Section 8.1]{Yaf92} and \cite{BY92a,GK69}.
\bd[{\cite[Chapter I.2]{GK69}}]
{\em
\item[\;\;\rm (i)]
A vector $\varphi\in\gH\setminus\{0\}$ is called
a root vector of a closed operator $T\in \cC(\gotH)$ corresponding
to its eigenvalue $\lambda_0\in\sigma_p(T)$ if there exists
$n\in{\N}$ such that $(T - \lambda_0)^n\varphi=0$. The closure
of the set $\mathfrak L'_{\lambda_0}(T)$ of all root vectors of $T$
corresponding to $\lambda_0$ is called the root subspace.
\bed
\gotL_{\lambda_0}(T)=\overline{\gotL'_{\lambda_0}(T)}, \quad
\gotL'_{\lambda_0}(T)=\{f \in \gotH:\ (T-\lambda_0)^n f=0\; \text{for
some}\; n\in{\N}\}.
\eed
\item[\;\;\rm (ii)]
The dimension $m_0=m_{\lambda_0}(A)=\dim\mathfrak
L_{\lambda_0}(A)$ is called the algebraic multiplicity of
$\lambda_0$.
\item[\;\;\rm (iii)]
An eigenvalue $\lambda_0\in\sigma_p(T)$ is called a
normal eigenvalue of $T$ if it is isolated and its
algebraic multiplicity $m_{\lambda_0}(T)$ is finite,
$m_{\lambda_0}(T)<\infty$.
}
\ed

If $m_0< \infty,$ then $\gotL'_{\lambda_0}(A)$ is closed and $=\gotL_{\lambda_0}(A)$ is closed.
An isolated eigenvalue $\lambda_0\in\sigma_p(T)$ is a normal one
if and only if
\bed
m_0 = \dim P_{\lambda_0} <\infty, \qquad P_{\lambda_0} =
-\frac{1}{2\pi i}\int_{|z- \lambda_0|=\delta}R_T(z)dz.
\eed
We set $m_{\lambda_0}(T) = 0$, if  $z_0$ is regular, i.e. $z_0\in \rho(T)$.

Further, if the function $f(\cdot)$ is analytic in a punctured
neighborhood of $z_0 (\in \C)$ and $z_0$ is not an essential
singularity of it, then the order $\ord(f(z_0))$ of  $f(\cdot)$ at $z_0 \in \C$ is the
integer $k \in \Z$ in the representation $f(z) = (z-z_0)^kg(z)$
where $g(z)$ is analytic at $z_0$ and $g(z_0) \not= 0$,
\cite[Chapter IV.3]{GK69}.

If $H'$ and $H$ are densely defined closed operators in the Hilbert space $\gotH$
such that $\{H',H\} \in \gotD$, then the perturbation determinant defined by \eqref{1.0} has the following properties:

\begin{enumerate}

\item
If $H',H \in \gotS_1(\gotH)$, then $\{H',H\} \in \gotD$ and
   \begin{equation}\la{1.1AA}
\gD_{H'/H}(z) = \frac{\det(I - z^{-1}H')}{\det(I - z^{-1}H)},
\qquad z\in \rho(H) \setminus \{0\}.
\end{equation}

\item
If $\{H'',H'\} \in \gotD$ and $\{H',H\} \in \gotD$, then
$\{H'',H\} \in \gotD$ and the following  chain rule holds
  \begin{equation}\la{1.1AB}
\gD_{H'',H'}(z)\gD_{H',H}(z) = \gD_{H'',H}(z), \qquad  z \in
\rho(H') \cap (H).
   \end{equation}
\item
If $\{H',H\} \in \gotD$, then $\{H,H'\} \in \gotD$ and
$\gD_{H'/H}(z)\gD_{H/H'}(z) = 1$ for $z \in \rho(H') \cap
\rho(H)$.

\item
If $\{H',H\} \in \gotD$ and $z$ is either a common regular
point or a normal eigenvalue of both $H'$ and $H$ of algebraic
multiplicities $m_z(H')$ and $m_z(H)$, then $\ord(\gD_{H'/H}(z)) =
m_z(H') - m_z(H)$.

\item
If $\{H',H\} \in \gotD$, then
\bea\la{1.1}
\lefteqn{
\frac{1}{\gD_{H'/H}(z)}\frac{d}{dz}\gD_{H'/H}(z)^{-1}= }\\
& &
\tr((H - z)^{-1} - (H' - z)^{-1}), \quad z \in \rho(H')\cap \rho(H).
\nonumber
\eea

\item
If $\{H',H\} \in \gotD$  and $\{H'^*,H^*\} \in \gotD$, then
$\gD_{H'^*/H^*}(z) = \overline{\gD_{H'/H}(\overline{z})}$ for $z
\in \rho(H^*)$.

\item If $\{H',H\} \in \gotD$, then the following  identity holds
\bed
\frac{\gD_{H'/H}(z)}{\gD_{H'/H}(\zeta)} = \det\left(I +
(z - \zeta)(H'-\zeta)^{-1}V(H - z)^{-1}\right),
\eed
$z \in \rho(H)$ and $\zeta \in \rho(H') \cap \rho(H)$.
\end{enumerate}

\section{Logarithm}\la{App.II}

In the following we need the definition of the logarithm $\log(z)$ of a
complex number $z \in \C$. We shall define the logarithm by
\be\la{2.19}
\log(z) := -i\int^\infty_0 \left((z + i\gl)^{-1} -
(1 + i\gl)^{-1}\right)d\gl, \quad z \in \C_+ \setminus -i\R_+,
\ee
with a cut along the negative imaginary semi-axis. One proves that
\bed
\log(e^{z}) = z, \quad e^z \in \C_+ \setminus -i\R_+,
\eed
which yields
\bed
e^{\log(z)} = z, \quad z  \in \C_+ \setminus -i\R_+.
\eed
Let $f(\cdot)$
and $g(\cdot)$ be  holomorphic functions in a domain $\gO$
satisfying  $f(z) \not= 0$ and $f(z) = e^{g(z)}.$  Then for a
neighborhood $\cO$ of a fixed point $z_0 \in \gO$ such that
$f(z_0)$ does not belong to the negative imaginary semi-axis one
has
\bed
\log(f(z)) = g(z) + 2n\pi i, \quad z \in \cO, \quad n \in \Z.
\eed
By analytical continuation this equality can be extended to the whole
$\gO$. Using definition \eqref{2.19} we find
\bed
\frac{d}{dz}\log(f(z)) = \frac{1}{f(z)}\frac{d}{dz}f(z), \quad z \in
\gO.
\eed
Furthermore, we need the definition of the logarithm of a
dissipative operator $G$ given in \cite{GM00b}. Let a $G$ be a
bounded dissipative operator such that $0 \in \rho(G)$. Then we set
\be\la{2.17}
\log(G) := -i\int^\infty_0 \left((G + i\gl)^{-1} -
(1 + i\gl)^{-1}\right)d\gl
\ee
where the integral is understood in the operator norm. One proves that
$e^{\log(G)} = G$, cf. \cite[Lemma 2.5(e)]{GM00b}. Moreover, from
\cite[Lemma 2.6]{GM00b} we find $0 \le \IM(\log(G)) \le \pi I$. If
$G \in \gotS_1(\gotH)$ and dissipative, then $\log(I
+ G) \in \gotS_1(\gotH)$. Moreover, one gets
\be\la{2.20}
\det(I + G) = e^{\tr(\log(I+G))}, \quad G \in \gotS_1(\gotH).
\ee

\section{Holomorphic functions in {$\C_+$}}\la{App.III}

A holomorphic function $F_+(\cdot)$ in $\C_+$ belongs to $F_+(\cdot) \in
H^\infty(\C_+)$ if $\sup_{z\in\C_+}|F_+(z)| < \infty$. Let  $\{z^+_k\}_{k\in\N} (\subset \C_+)$ be  the set of
its zeros and   $m_k$  the corresponding  multiplicities.
It is well known (see for instance \cite[Section VI C]{Koo80})
that the set of zeros  satisfy the condition
\be\la{4.12}
\sum_k\frac{m_k\IM(z^+_k)}{1 + |z^+_k|^2}<\infty.
\ee
If the sequence $\{\ga^+_k\}_{k\in\N}\subset \R$ is chosen such that
$e^{i\ga^+_k}({i - z^+_k})({i - \overline{z}{^{\,+}_k}})^{-1} \ge
0$,\ $k\in \N,$ then  the corresponding Blaschke product
\be\la{7.17A}
\cB_+(z) = \prod_k (b_{z_k})^{m_k}  :=  \prod_k \left(e^{i\ga^+_k}\frac{z -
z^+_k}{z - \overline{z^+_k}}\right)^{m_k}, \qquad  z \in \C_+,
\ee
converges uniformly on compact subsets of $\C_+$.

Moreover, $F_+(\cdot)$  admits (see \cite[Section VI C]{Koo80})
the following  representation
\be\la{4.9a}
F_+(z) = \varkappa_+\cB_+(z)
\exp\left\{\frac{i}{\pi}\int_\R\left(\frac{1}{t-z} -
\frac{t}{1+t^2}\right)d\mu_+(t)\right\} e^{i\ga_+ z},
\ee
$z \in \C_+$, where $\varkappa_+ \in \T$, $\ga_+ \ge 0$ and  $\mu_+(\cdot)$  is a non-decreasing
function on $\R$ generating a non-negative  Borel
measure  and satisfying
\be \la{4.10}
\int_\R \frac{1}{1 + t^2}d\mu_+(t) < \infty.
\ee
If $F_+(z)$ has no zeros in $\C_+$, the Blaschke product
$\cB_+(\cdot)$ in  \eqref{4.9a} is missing. Let $\mu_+ = \mu^s_+ +
\mu^{ac}_+$ be the Lebesgue  decomposition  of $\mu_+$, where
$\mu^s_+$  and $\mu^{ac}_+$ are  the singular  and the absolutely
continuous measures, respectively. Setting
  \be\label{4.10Inner}
\begin{split}
I_{F_+}(z) & := \cB_+(z) S_{F_+}(z) e^{i\ga_+ z},\\
S_{F_+}(z) & := \exp\left\{\frac{i}{\pi}\int_R \left(\frac{1}{t-z} -
\frac{t}{1+t^2}\right)d\mu^s_+(t)\right\}, 
\end{split}
\ee
where $\ga_+ \ge 0$ and
\bed
\cO_{F_+}(z) := \exp\left\{\frac{i}{\pi}\int_R \left(\frac{1}{t-z}
- \frac{t}{1+t^2}\right)d\mu^{ac}_+(t)\right\},\quad z \in \C_+,
\eed
one gets the unique factorization $F_+(z) =
\varkappa_+I_{F_+}(z)\cO_{F_+}(z)$, $z \in \C_+$, where
$\varkappa_+ \in \T$ and $I_{F_+}(z)$ and $\cO_{F_+}(z)$ are the
inner  and the outer factors, respectively.   Note, that
$|I_{F_+}(t + i0)| = 1$, $|\cO_{F_+}(t+i0)| = |F_+(t+i0)|$ for
a.e. $t \in \R$ and $d\mu^{ac}_+(t) = - \ln(|F_+(t+i0)|)dt$.
 Note, that
$F_+(\cdot)$ is an outer function in $\C_+$, if and only if it
admits the representation
\be\la{4.15a}
F_+(z) = \varkappa_+ \exp\left\{ -\frac{i}{\pi}\int_R
\left(\frac{1}{t-z}-
\frac{t}{1+t^2}\right)\ln(|F_+(t+i0)|)\,dt\right\}, \quad  z \in
\C_+.
\ee
Clearly, $F_+(i)$ is real, if and only if  $\varkappa_+ = 1$.

A holomorphic function $F$ belongs to the
Smirnov class $\cN^+(\C_+)$ if it admits the representation $F =
F_+/G$ where $F_+,G \in H^\infty(\C_+)$ and $G$ is an outer
function. Any function $F \in \cN^+(\C_+)$ admits the
representation
\be\la{D.7}
\begin{split}
F(z) = & \varkappa B_+(z)\exp\left\{\frac{i}{\pi}\int_\R\left(\frac{1}{t-i} - \frac{t}{1+t^2}\right)d\mu^s_+(t)\right\}\times\\
     & \exp\left\{-\frac{i}{\pi}\int_\R\left(\frac{1}{t-i} - \frac{t}{1+t^2}\right)h(t)dt\right\}e^{i\ga z},\quad z \in \C_+,
\end{split}
\ee
where $\varkappa \in \T$, $\ga \ge 0$, $\mu^s_+(\cdot)$ is a non-negative Borel
measure, which is singular with respect to the Lebesgue measure, and
$h \in L^1(\R,\tfrac{1}{1+t^2}dt)$.
Using \eqref{D.7} one easily verifies that the $H^\infty(\C_+)$-functions $F_+$ and $G$ can be chosen contractive.
Indeed, let $\eta(t) := \max\{h(t),0\} \ge 0$, $t \in \R$, and
$k(t) := \eta(t) - h(t) \ge 0$, $t \in \R$.
Notice that $h(t) = \eta(t) - k(t)$, $t\in \R$.
Setting $d\mu_+(\cdot) := d\mu^s_+(\cdot) + k(\cdot)dt$,
\bed
F_+(z) := \varkappa B_+(z)\exp\left\{\frac{i}{\pi}\int_\R\left(\frac{1}{t-i} - \frac{t}{1+t^2}\right)d\mu_+(t)\right\}e^{i\ga z},
\quad z \in \C_+,
\eed
and
\bed
G(z) := \exp\left\{\frac{i}{\pi}\int_\R\left(\frac{1}{t-i} - \frac{t}{1+t^2}\right)\eta(t)dt\right\}
\quad z \in \C_+,
\eed
we define contractive analytic functions, where $G$ is an outer function, such that $F = F_+/G$.
Summing up we have proved the following lemma.
\bl\la{D.I}
If $F \in \cN^+(\C_+)$, then there exists a non-negative Borel measure
$\mu_+(\cdot)$ satisfying $\int_\R \tfrac{d\mu_+(t)}{1 + t^2}dt < \infty$ and
a non-negative function $\eta(\cdot) \in L^1(\R,\tfrac{1}{1+t^2}dt)$
as well as constants $\varkappa \in \T$ and $\ga \ge 0$
such that the representation
\be\la{D.8}
F(z) = \varkappa B_+(z)\exp\left\{\frac{i}{\pi}\int_\R\left(\frac{1}{t-i} - \frac{t}{1+t^2}\right)d\mu(t)\right\}
\ee
holds where $d\mu(\cdot) = d\mu_+(\cdot) - \eta(\cdot)dt$.
\el

\section{On {$H^1(\D)$} functions}\la{App.IV}

Let $\D := \{w \in \C: |w| < 1\}$.
By $\ln(\cdot)$ we denote a branch of the logarithm such that
$\ln(z) \in \R$ for $z \in \R_+$ and $\IM(\ln(z)) \in
(-\pi/2,\pi/2)$ for $\RE(z)
> 0$.
\bl\la{A.III}
Let $H(w)$ be a holomorphic function in $\D$ such that $\RE(H(w))
\ge 0$ for $w \in \D$. Let $G(w) := \ln(1 + H(w))$ for $w \in \D$.
Then $G(w) \in H^1(\D)$ and the following  estimate holds
        \be\la{A.80}
0 \le\int^\pi_{-\pi} \RE(G(e^{i\gth})d\gth \le 2\pi\,|H(0)|.
     \ee
\el
\begin{proof}
Obviously we have $|\IM(G(w))| \le \pi/2$, $w \in \D$. Furthermore,
we have
\bed
G(0) = \frac{1}{2\pi}\int^{\pi}_{-\pi} G(re^{i\gth})d\gth,
\quad r \in (0,1),
\eed
which yields
\bed
2\pi\RE(G(0)) = \int^{\pi}_{-\pi}\RE(G(re^{i\gth}))d\gth.
\eed
Since $\RE(G(re^{i\gth})) \ge 0$ we obtain
      \bed
\|G\|_{H^1} \le 2\pi\RE(G(0)) + \pi^2
       \eed
which yields $G \in H^1(\D)$. In particular, we have
\be\la{A.8}
\|G_R\|_{L^1} = 2\pi G_R(0), \quad G_R(w) = \RE(G(w)), \quad w \in \D.
\ee
Using the estimate $\RE(G(0)) = \ln(|1 + H(0)|)  \le |H(0)|$ we
arrive at \eqref{A.80}.
\end{proof}

The result can be carried over to upper half-plane.
\bc\la{A.4}
Let $h(z)$, $z \in \C_+$, be a holomorphic function such that
$\RE(h(z)) \ge 0$ for $z \in \C_+$. Let $g(z) := \ln(1 + h(z))$
for $z \in \C_+$. Then the following  estimate
   \be\la{A.10}
\int_\R|g(x+i0)|\frac{dx}{1+x^2} \le  2\pi\;|h(i)|
           \ee
is valid where $g(x+i0) := \lim_{y\downarrow 0}g(x + iy)$.
\ec
\begin{proof}
We set
\bed
H(w) := h\left(i\frac{1+w}{1-w}\right)
\eed
and
     \bed
G(w) := \ln(1 + H(w)) = \ln\left(1 +
  h\left(i\frac{1+w}{1-w}\right)\right)
= g\left(i\frac{1+w}{1-w}\right).
     \eed
Since
\bed
\int^{\pi}_{-\pi}|G(e^{i\gth})|d\gth = \int_\R |g(x + i0)|\frac{dx}{1+x^2}
\eed
and $h(i) = H(0)$ we obtain \eqref{A.10} from \eqref{A.80}.
\end{proof}

\section{Riesz-Dunford functional calculus}\la{App.V}

Let $T$ be a densely defined closed operator. We say the
function $\Phi$ belongs to the class $\cF(T)$ if there is a
a simple closed curve $\gG$ in $\C$ which does not intersect the real
axis such that
\begin{enumerate}

\item[(i)]
its open exterior domain $\gO^{\rm ext}_\gG$ contains the spectra
of $T$ and

\item[(ii)] there is a neighborhood $\cO$ of the closed set
$\overline{\gO^{\rm ext}_\gG}$ such that $\Phi$ is holomorphic in
$\cO$ including infinity.
\end{enumerate}
We note that if $\rho(T)$ is not empty, then the class $\cF(T)$ is
not an empty too. In this case one defines $\Phi(T)$  by
\be\la{F0}
\Phi(T) := \Phi(\infty)I + \frac{1}{2\pi i}\oint_\gG \Phi(z)(T - z)^{-1} dz
\ee
where the integral $\oint_\gG$ is taken in mathematical positive
sense with respect to open inner domain $\gO^{\rm in}_\gG$, see
\cite[Section VII.9]{DSch88}. We note that
\be\label{F3}
\Phi(\xi) = \Phi(\infty) -\frac{1}{2\pi i}\oint_\gG
\frac{\Phi(z)}{z - \xi}dz, \quad \xi \in \gO^{\rm ext}_\gG,
   \ee
Since
\be\la{4.14b}
\oint_\gG |\Phi(z)|\;|dz| < \infty
\ee
the integral $\oint_\gG \Phi(z) dz$ is well-defined. Hence the
residuum $\res_\infty(\Phi)$, $\Phi \in \cF(T)$, is well defined by
\be\la{4.14a}
\res_\infty(\Phi) := -\frac{1}{2\pi i} \oint_\gG \Phi(z) dz
\ee
Since the curve $\gG$ does not
intersect the real axis, we get from \eqref{F3} and \eqref{4.14b} that $\sup_{t\in\R}(1+t^2)|\Phi'(t)| <
\infty$ for $\Phi \in \cF(T)$. Therefore, if $\nu$ is a
complex-valued Borel measure satisfying $\int_\R \frac{1}{1+t^2}|d\nu(t)| < \infty$, then
\bed
\int_\R |\Phi'(t)||d\nu(t)| < \infty
\eed
which guarantees the existence of the integral $\int_\R \Phi'(t)d\nu(t)$
and
\be\label{F4}
\int_\R \Phi'(t)d\nu(t) = -\frac{1}{2\pi i}\oint_\gG \Phi(z)\left(\int_\R
\frac{1}{(t-z)^2}d\nu(t)\right) dz
   \ee
for $\Phi \in \cF(T)$ provided $\nu$
satisfies $\int_\R \frac{1}{1+t^2}|d\nu(t)| < \infty$.

Let $T$ and $T'$ be two densely defined closed operators. We set
$\cF(T,T') := \cF(T) \cap \cF(T')$, that is, there is a simple
closed curve $\gG$ such that $\gO^{\rm ext}_\gG$ contains the
spectra of both $T$ and $T'$. If $\rho(T) \cap \rho(T') \not=
\emptyset$, then $\cF(T,T')\not = \emptyset$.
\bl\la{A.Va}
Let $T$ and $T'$ be two densely defined closed operators in $\gotH$
such that $\rho(T) \cap \rho(T') \not= \emptyset$.
If the condition $(T'- \xi)^{-1} - (T- \xi)^{-1} \in \gotS_1(\gotH)$
for some $\xi \in \rho(T) \cap \rho(T')$, then $\Phi(T') - \Phi(T) \in \gotS_1(\gotH)$
for $\Phi \in \cF(T,T')$.
\el
\begin{proof}
Obviously we have
\bed
\Phi(T') - \Phi(T) = \frac{1}{2\pi i}\oint_\gG \Phi(z)((T' - z)^{-1} - (T - z)^{-1})dz
\eed
which yields the estimate
\bed
\|\Phi(T') - \Phi(T)\|_{\gotS_1} \le \sup_{z\in\gG}\|(T' - z)^{-1} - (T - z)^{-1}\|_{\gotS_1}\frac{1}{2\pi }\oint_\gG |\Phi(z)||dz| < \infty
\eed
which proves the assertion.
\end{proof}

\section{Root vectors}\la{App.VI}

Let $T$ be an unbounded operator in $\gotH$. If $0 \not\ni \psi \in \ker((T -
\gl)^n)$ for some $\gl \in \C$ and $n \in \N$, then $\psi$ is called
a root vector of $T$ belonging to $\gl$. The point $\gl$ is necessarily an
eigenvalue of $T$. The set of all root vectors belonging to $\gl$
is denoted by $\cV_T(\gl)$. The set $\cV_T = \bigcup_{\gl \in \gs(T)}\cV_T(\gl)$ is
called the root vector system of $T$. If $\cV_T$ is a total set,
then the root vector system is called complete.
\bl\la{A.V}
Let $T$ be a densely defined closed operator such that $\xi \in \rho(T)$.
The root vector system of $T$ is complete if and only if the root
vector system of $R := (T -\xi)^{-1}$ is complete.
\el
\begin{proof}
Using
\be\la{5.820}
(R - \mu)^n\psi = (-1)^n\mu^n(T - \xi)^{-n}(T - \gl)^n\psi, \quad \mu
:= \frac{1}{\gl - \xi}, \quad \psi \in \gotH,
\ee
we get that $\psi \in \cV_T(\gl)$ yields $\psi \in
\cV_R(\mu)$. Conversely, if $\psi \in \cV_R(\mu)$, then
\bed
0 = (R - \mu)^n\psi = \sum^n_{k=0}\binom{n}{k}\mu^{n-k}R^k\psi =
\mu^n\psi + \sum^n_{k=1}\binom{n}{k}\mu^{n-k}R^k\psi
\eed
which implies
\bed
\psi = -R\sum^{n-1}_{k=0}\binom{n}{k+1}\mu^{-(k+1)}R^k\psi.
\eed
Hence, $\psi \in \dom(T)$, and
\bed
\psi = -R\sum^{n-1}_{k=0}\binom{n}{k+1}\mu^{-(k+1)}R^{k+1}(T-\gl)\psi.
\eed
This yields
\bed
\psi = -R^2\sum^{n-1}_{k=0}\binom{n}{k+1}\mu^{-(k+1)}R^k(T-\gl)\psi.
\eed
and $\psi \in \dom(T^2)$. If we proceed further in this way we
finally get $\psi \in \dom((A-\gl)^n)$. Applying again formula
\eqref{5.820} we obtain $\psi \in \cV_T(\gl_0)$.
\end{proof}

\section{The class {$C_0$}}\la{App.VII}

Let us briefly recall some basic concepts
and facts on contractions following \cite{FN70}. In \cite{FN70}
Nagy and Foias  using theory of dilations have extended
the Riesz-Dunford functional calculus for a contraction $T$
to the class $H^{\infty}_T(\D)$ (see \cite[Section
3.2]{FN70} for precise definitions). If a contraction $T$ is
completely non-unitary, then $H^{\infty}_T(\D)=
H^{\infty}(\D)$ is just the Hardy class in the unit disc $\D$.
The extended functional calculus makes it possible to
introduce concepts of $C_0$-contractions and minimal
annihilation function.
\bd[{\cite{FN70}}]
\item[\;\;\rm (i)]
A contraction $T$ in $\gH$ is put in the class $C_{0\cdot}$
($C_{\cdot 0}$ ) if $s-\lim_{n\to \infty}T^n = 0$ (s-$\lim_{n\to
\infty}T^{*n} = 0$). It is set $C_{00} := C_{\cdot 0}\cap
C_{0\cdot}.$

\item[\;\;\rm (ii)]
It is said that  a completely non-unitary operator
$T$ belongs to the class $C_0$ if there exists a function $u(\cdot)\in
H^{\infty}(\D)\setminus\{0\}$ such that $u(T)=0$.
The function  $u(\cdot)$  is called an annihilation function for $T$.

\item[\;\;\rm (iii)]
An annihilation  function $u_0(\cdot)$  is called minimal if
it is a divisor in $H^{\infty}(\D)$ of any other annihilation
function $u(\cdot)$ for $T$.
\ed

It is well known that $C_{0}\subset C_{00}.$ Moreover, it is known
\cite[Proposition 3.4.4]{FN70} that for any $T\in C_0$ the minimal
function exists and is unique up to a multiplicative constant. The
minimal function is denoted by $m_T(\cdot)$. {\vio{It}} is always an inner one.

Alongside a $m$-dissipative operator $D$ in $\gH$ we consider its
Cayley transform
\bed
T := T_{D} :=(D - i)(D + i)^{-1} = I -2i(D + i)^{-1}.
\eed
Clearly, $T_D$ is a contraction. Moreover,  $T_D$ is
completely non-unitary if and only if $D$ is completely
non-selfadjoint.

For any function $v(\cdot) \in H^{\infty}({\C}_+)$ and any completely
non-selfadjoint $m$-dissipative operator $D$ we set $v(D) := \wt v(T_D)$
where
\bed
H^\infty(\D) \ni \wt v(\zeta) := v\left(i\frac{1+\zeta}{1-\zeta}\right), \quad \zeta \in \D.
\eed
We say the $m$-dissipative operator $D$ belongs to the class  $C_{0\cdot}$
($C_{\cdot 0},\  C_0$) if $T_D$ belongs to $C_{0\cdot}$
(resp. $C_{\cdot 0},\  C_0$). In other words,  $D\in C_{0}$ if it is
completely non-selfadjoint and there  is  a function $v(\cdot) \in
H^{\infty}({\C}_+)$ such that $v(D) = \wt v(T_D) = 0$.  Clearly,  there  always exists a minimal function
$m_D(\cdot) \in H^{\infty}({\C}_+)$ which is an inner function.
\end{appendix}


\def\cprime{$'$} \def\cprime{$'$} \def\cprime{$'$} \def\cprime{$'$}
  \def\cprime{$'$} \def\cprime{$'$} \def\cprime{$'$}
  \def\lfhook#1{\setbox0=\hbox{#1}{\ooalign{\hidewidth
  \lower1.5ex\hbox{'}\hidewidth\crcr\unhbox0}}} \def\cprime{$'$}
  \def\cprime{$'$} \def\cprime{$'$} \def\cprime{$'$} \def\cprime{$'$}
  \def\cprime{$'$} \def\cprime{$'$} \def\cprime{$'$} \def\cprime{$'$}
  \def\lfhook#1{\setbox0=\hbox{#1}{\ooalign{\hidewidth
  \lower1.5ex\hbox{'}\hidewidth\crcr\unhbox0}}} \def\cprime{$'$}

\end{document}